\documentclass[11pt]{article}%
\usepackage{amsmath}
\usepackage{amsfonts}
\usepackage{sw20elba}
\usepackage{amssymb}
\usepackage{hyperref}
\usepackage{graphicx}
\usepackage{multirow}
\usepackage{float}
\usepackage[longnamesfirst]{natbib}%
\usepackage{enumitem} 
\providecommand{\U}[1]{\protect\rule{.1in}{.1in}}
\newtheorem{theorem}{Theorem}[section]
\newtheorem{assumption}{Assumption}[section]

\newtheorem{corollary}{Corollary}[section]

\newtheorem{definition}[theorem]{Definition}

\newtheorem{lemma}{Lemma}[section]

\newenvironment{proof}[1][Proof]{\noindent \textbf{#1.} }{\  \rule{0.5em}{0.5em}}
\newcommand{\GG}[1]{}
\usepackage{xcolor}
\hypersetup{
	colorlinks,
	linkcolor={red!50!black},
	citecolor={blue!50!black},
	urlcolor={blue!80!black}
}

\begin{document}

\title{PAC-Bayesian Treatment Allocation Under Budget Constraints\footnote{The author is grateful to his advisor Yixiao Sun for feedback and support on this paper.  Kaspar Wuthrich, Xinwei Ma, Richard Carson, Ross Starr, and Nikolay Kudrin have also provided helpful comments and useful discussion.}}
\author{Daniel F. Pellatt\thanks{Department of Economics, University of California, San Diego. 
Address: 9500 Gilman Drive, La
Jolla, CA 92093.  Email: dpellatt@ucsd.edu}}
\date{ \today \\%
}
\maketitle

\begin{abstract}
This paper considers the estimation of treatment assignment rules when the policy maker faces a general budget or resource constraint.  Utilizing the PAC-Bayesian framework, we propose new treatment assignment rules that allow for flexible notions of treatment outcome, treatment cost, and a budget constraint.  For example, the constraint setting allows for cost-savings, when the costs of non-treatment exceed those of treatment for a subpopulation, to be factored into the budget. It also accommodates simpler settings, such as quantity constraints, and doesn't require outcome responses and costs to have the same unit of measurement.  Importantly, the approach accounts for settings where budget or resource limitations may preclude treating all that can benefit, where costs may vary with individual characteristics, and where there may be uncertainty regarding the cost of treatment rules of interest.  For a class of stochastic treatment rules, we derive non-asymptotic generalization bounds for the target population costs and sharp oracle-type inequalities that compare the rules' welfare regret to that of optimal policies in relevant budget categories.  A closely related, non-stochastic, model aggregation treatment assignment rule is shown to inherit desirable attributes.

\bigskip



\end{abstract}

\section{Introduction \label{sec: Intro}}
This paper proposes new statistical decision rules for treatment assignments under a general budget or resource constraint.  A key objective in the empirical analysis of treatment data is identifying policies that result in the most beneficial outcomes.  There is a large literature (e.g. \cite{Manski2004} and \cite{hirano2009asymptotics}) that examines how to determine which policies are optimal to implement in the absence of constraints such as one on policy cost.  In practice, however, policy makers are rarely free from constraints when it comes to the policies they may enact.  Several recent papers in the econometrics literature, including  \cite{KT2018}, \cite{athey2021policy}, and \cite{mbakop2021model}, consider the treatment estimation problem from an empirical welfare maximization (EWM) perspective that allows for arbitrary constraints on the functional form of the decision rule.  However, these papers do not address general budget constraints nor cost uncertainty that varies with the characteristics of individual agents.  For example, while \cite{KT2018} consider quantity constraints via random rationing, this treats costs as fixed and hence cannot identify which policies most efficiently balance cost vs. outcome trade-offs when costs vary with individual characteristics.  

Here we focus on the setting where costs may be uncertain, current resource limitations
may preclude treating all that can benefit, and where individual characteristics can influence
treatment responses and costs.  Compared to the unconstrained setting, the theoretically optimal treatment rule involves population objects that are more difficult to estimate and analyze in concert.  For example, \cite{BHATTACHARYA2012168} show that under a quantity constraint, which is simpler than the setting with variable costs, the optimal rule is to assign treatment when the conditional average treatment effect exceeds its $(1-c)$th quantile.  Here $c$ is the maximal proportion of treatments assignable under the constraint.   As a result, it can be difficult to evaluate properties of interest for proposed approaches and each existing approach has limitations.  

The contributions of the paper are as follows.  First, we propose new treatment rules that expand the tool set available to policy makers in the budget constrained setting.  Second, we show they possess several potential benefits in terms of theoretical guarantees, the variety of settings in which they can be applied, and ease of estimation.   Third, we show expert knowledge can be incorporated when the policy maker has non-data-dependent insights into the problem.  However, the ability to integrate expert knowledge is a secondary feature of the approach.  In our primary implementation we assume no such knowledge. 

PAC-Bayesian analysis applies the probably approximately correct learning framework to objects of interest that involve probability distributions over model or parameter families. These objects can include, for example, treatment rules formed by aggregating over a family of potential rules.  Our  work can be seen as extending the PAC-Bayesian learning approach to the treatment setting in a way that incorporates a secondary cost objective.  This motivates the proposed rules and allows us to derive generalization bounds for the costs and oracle-type inequalities for the welfare regret of proposed rules.  Here, the welfare regret associated with a treatment rule is the loss in expected welfare of the decision rule relative to the theoretically optimal decision rule (cf. \cite{Manski2004}).  To work within the regret framework, we also derive the form of a theoretically optimal treatment policy if the data generating process (DGP) were known under a general budget constraint.

Individualized treatment policies under budget restrictions are of interest in a variety of settings.  Often policy makers with limited resources face uncertainty regarding the costs and benefits of potential policies where this uncertainty is driven due to the fact that costs and benefits vary with the individual characteristics of those who decide to participate in a program.  For example, \cite{FinkelsteinEtAl2012OHI} examine outcomes such as health care utilization and self-reported health measures following a randomized expansion of household access to Medicaid in Oregon.  A policy maker may be interested in identifying policies to maximize a well-defined weighted average of such outcomes given a binding expenditure constraint.  The government has control over eligibility rules defined on characteristics such as age, income, and the number of children in a household that directly influence expected cost and cost uncertainty.  

Insecticide-treated nets (ITNs) for protection against malaria in regions of Africa represent another common example.  \cite{lengeler1998insecticide}, for instance, documents reductions in child mortality while \cite{kuecken2014does}  document returns to education related to ITN provisions.  \cite{Teklehaimanot2007} estimate the cost of providing an ITN to every at-risk individual in sub-Saharan Africa to be 2.5 billion dollars.  However, government and aid funding was below that level at the time of the study. \cite{BHATTACHARYA2012168} look at a treatment policy estimator under quantity constraints derived from data from a randomized experiment assigning ITNs to rural households in Kenya.  They use fixed costs to estimate rules that satisfy quantity constraints.  Our approach makes it possible to target policies in such a way as to account for cost heterogeneity (e.g. different distribution channels) and hence improve efficiency and achieve a higher overall outcome level.

Beyond aid and social safety net policies, the budget constrained treatment assignment problem can also arise in a commercial context for firms considering potentially costly promotions aimed at obtaining new customers.  For instance, \cite{sun2021treatment} recently proposed a budget constrained treatment estimator aimed at determining which customers should be offered trial access to a premium service.  They seek to use customers' individual characteristics to discriminate against making offers to customers likely to heavily utilize the service in the trial (high cost) while being unlikely to use the service after the trial period expires.  Rather than the simple notion of not wanting to implement a policy that leads to long-term losses, many companies will also face a short-term constraint on how much they can ``lose" in the trial phase to gain market share. For other firms, like Uber which is considered in \cite{sun2021treatment}, a deeper issue may arise. Increasing sales or trial offers may fundamentally alter the firm's cost structure (e.g., increasing driver compensation to induce enough new drivers to work to handle the increased number of trips).

The rules we develop start from a user-specified family of (non-stochastic) treatment models $\mathcal{F}$ that map an individual's covariates that are observable pre-treatment to the $\{0,1\}$ treatment indicator space.  Rather than choosing the model that maximizes the empirical welfare in $\mathcal{F}$, for example, we instead consider stochastic treatment rules derived from $\mathcal{F}$ and a measure of budget penalized empirical welfare.  Given an individual's pre-treatment covariates, their treatment probability is calculated as an exponentially weighted average over the treatments specified  by members of $\mathcal{F}$.  The treatment probability is similar to a weighted majority vote taken over $\mathcal{F}$.  The exponential weighting received by members of the model family is greatest for models with a large budget-penalized empirical welfare.  The magnitude of the penalization term related to cost is determined by a parameter $u$ that  modulates the trade-off between maximizing welfare and reducing costs.  Any choice for $u$ will correspond to a different maximal empirical budget, with $u=0$ corresponding to an unlimited budget (no constraint).  Typically, for larger sample sizes, the rule is unlikely to assign identical covariates to different treatments unless there are subsets of the model family with  similarly high values of penalized welfare that prescribe different treatments.   We also consider closely related, non-stochastic, model aggregation treatment rules that aggregate over $\mathcal{F}$ to make treatment decisions.

Utilizing a PAC-Bayesian framework, under reasonable conditions we show that for a set of $u$ values, in large samples, with high probability we obtain increasingly accurate estimates of the target population costs associated with corresponding stochastic treatment rules.  We can use these estimates to select $u$ or, alternatively, $u$ can be chosen via cross-validation.  At the same time, with $u$ chosen in either manner, with high probability the resulting rule achieves a welfare regret comparable to that of the best models in the model family that have a similar target population cost.  Starting from a set of budget penalty parameters, the policy maker can trace out good estimates of the feasible target population budgets, select the parameter associated with one of these estimates, and obtain a treatment rule with desirable regret properties.  Regarding the non-stochastic, model aggregation treatment rules, we show that they inherit desirable properties  from the stochastic rules.  We also consider the setting where $u$ is chosen to meet a predetermined target population budget level.  The procedure in this case is still reasonably motivated, as the rule minimizes an upper bound on the target population regret among rules that satisfy an empirical budget constraint.  However, the generalization bounds for the target population cost and the oracle-type inequalities in this case become more complex to interpret.

The remainder of the paper is organized as follows.  Section \ref{sec: literature review} discusses related literature and papers with alternative budget constrained treatment estimators.  Section \ref{sec: Setup} details the statistical setting, treatment model formulation, and initial properties useful for later results.  Section \ref{sec: PAC Analysis} provides theoretical motivation for the proposed treatment rules, utilizing the PAC-Bayesian analysis framework to examine (frequentist) properties of the proposed rules.  Section \ref{sec: Simulation and Implementation}  conducts a simulation experiment and discusses implementation and estimation.  Lastly, Section \ref{Sec: Empirical Illustration Gibbs} conducts a short empirical illustration utilizing data from the Job Training Partnership Act Study and Section \ref{sec: Conclusion} concludes.

\section{Related Literature \label{sec: literature review}}
The topic of budget constrained treatment allocation is the subject of a small but growing literature.  \cite{sun2021treatment} and \cite{wang2018learning} empirically implement treatment rules starting from the notion of a theoretically optimal rule.  They estimate unknown population level objects that appear in the optimal rules and then plug in the empirical counterparts to the corresponding theoretical  formulas to obtain rules.  The standard drawback of this sort of approach is that the estimation technique doesn't directly target policies that maximize the welfare problem of interest.  For example, the regressions utilized to fit the conditional average treatment and cost functions in \cite{wang2018learning} might yield parameters that are most accurate in regions of the covariate space that are less important for distinguishing individuals with a high outcome-to-cost ratios in the population.    \cite{wang2018learning} also consider a second method that shares similarities with the approach taken by \cite{huang2020estimating}.  These approaches add the budget constraint to the outcome-weighted treatment learning approach considered, for example, in \cite{zhao2012estimating}.  These approaches work from optimization problems that directly target an empirical version of the problem of interest.  

One drawback of the aforementioned techniques is a lack of theoretical insight regarding the true target population cost and risk attributes of the proposed rules.   \cite{sun2021empirical} adapts the EWM setting of \cite{KT2018} to account for a general budget constraint. She considers a conservative rule that will satisfy the budget constraint asymptotically.  She also considers a modified rule where a Lagrange multiplier parameter is capped during estimation. This will, asymptotically, approach the welfare of the budget constrained welfare maximizing policy among the user-specified model class.   This methodology extends the arbitrary form features of EWM to the budget constraint setting.  However, the rules involve a non-convex estimation procedure that may become difficult if the model class includes more flexible functional forms.  While our methodology sacrifices some ability to satisfy functional form constraints due to its stochastic nature, one benefit is that we can take advantage of Bayesian estimation machinery as discussed in Section \ref{sec: Simulation and Implementation}.  Lastly, although the modified rules of \cite{sun2021empirical} will approach the optimal rule within the original budget constraint, it is worth noting that the modified rule may violate that budget constraint.  One benefit of our approach is we can compare our rules to those with the highest welfare among rules with the same target population cost as the proposed rules.

In a broader context, this paper contributes to a growing literature on  statistical treatment rules in econometrics, including \cite{Manski2004},  \cite{dehejia2005program}, \cite{hirano2009asymptotics}, \cite{BHATTACHARYA2012168}, \cite{KT2018}, \cite{viviano2019policy}, and \cite{athey2021policy}.  This literature has overlap with additional fields including statistics and machine learning.  For examples, see \cite{qian2011performance} and \cite{beygelzimer2009offset}, respectively.  Additional references and a discussion of the links between these fields can be found in \cite{athey2021policy}.  In the machine learning literature, \cite{london2019bayesian} utilize a PAC-Bayesian approach to policy estimation for the logged bandit feedback problem which is closely related to treatment policy estimation.  We also note that \cite{kitagawa2023stochastic} examine stochastic treatment assignment rules from a PAC-Bayesian perspective.  Their paper's approach has overlap with ours, however the papers diverge in a number of dimensions stemming from our focus on the setting with a general budget constraint which is not considered there.  

Lastly, our analysis and proposed treatment rules are heavily influenced by the PAC-Bayesian machine learning literature.  Seminal works in this area include \cite{shawe1997pac},  \cite{mcallester1999some}, \cite{mcallester1999pac}, \cite{seeger2002pac}, and \cite{mcallester2003pac}.  In particular, we utilize techniques stemming from \cite{catoni2007pac}, \cite{lever2010distribution}, \cite{maurer2004note},  \cite{germain15aJMLR},  and \cite{alquier2016properties}.  The theoretical contribution of our paper is, first, to modify and adapt relevant tools and generalization bounds to the treatment choice setting.  We also develop the incorporation of a secondary objective or loss function (the treatment cost cost) into the analysis that yields informative oracle-type inequalities and generalization bounds relevant to the constrained budget setting.  

\section{Setup and Assumptions \label{sec: Setup}}

\subsection{Statistical Setting and Policy Maker's Problem \label{sec: stat setting and policy maker's problem}}
We consider the setting where a policy maker has data consisting of observations 
\[Z_{i}=(Y_{i}, C_{i}, D_{i}, X_{i}), \ i=1,\dots, n.\] 
Here,   $X_{i}\in\mathcal{X}\subset\mathbb{R}^{d_{x}}$, where $d_{x}\in\mathbb{N}$, denotes a vector of covariates for individual or unit $i$ observed prior to treatment assignment, $Y_{i}\in\mathbb{R}$ is unit $i$'s outcome that is observed after treatment assignment, $C_{i}\in\mathbb{R}$ is the cost incurred and $D_{i}\in\{0,1\}$ is a treatment assignment indicator that is $1$ if unit $i$ was assigned the treatment and is zero otherwise.  $C_{i}$ may be uncertain at the time of treatment assignment and is allowed to be observed after treatment assignment.    

To account for heterogeneous treatment responses and costs, we work from a potential outcomes and costs framework.  For unit $i$ and for $j\in\{0,1\}$, let $Y_{i,j}$ and $C_{i,j}$ denote the outcome and cost, respectively, that would have been observed if unit $i$ had been assigned $D_{i}=j$.  Ignoring the index $i$, we can relate the observed outcome and cost to their potential outcomes and costs by writing
\begin{equation}
\label{eqn: Y, C in terms of potential outcomes}
Y=Y_{1}D+Y_{0}\left ( 1-D \right ), \ \ C=C_{1}D+C_{0}\left ( 1-D \right ).
\end{equation}
The following assumption formalizes this setting.  It also includes conditions needed to identify properties related to potential outcomes and costs when they are not observed directly in sample data.  

\begin{assumption}
	\label{Assumption: treatment identification and boundedness}
	\begin{enumerate}[label=(\roman*)]
		\item Random Sample:  Let $Q$ be the joint distribution of $(Y_{0},Y_{1}, C_{0},C_{1}, D, X )$,
		where $Y_{0}, Y_{1}, C_{0}, C_{1}\in\mathbb{R}$, $D\in\{0,1\}$, $X\in\mathcal{X}\subseteq \mathbb{R}^{d_{x}}$.  Let $Z=(Y,C,D,X)\in \mathcal{Z}$ be distributed according to $P$ where $P$ is determined by $Q$ and \eqref{eqn: Y, C in terms of potential outcomes}.  We assume the sample $S=\{Z_{i}\}_{i=1}^{n}$
		$\sim P^{\otimes n}$
		is a size $n$ i.i.d. sample\footnote{To denote the probability of an event $A$ under this sampling distribution, we will use the notation $P^{n}(A)$.  To denote the probability of an event $B$ under the distribution $P$, we write $P(B)$.}.  We denote the sample space $S\in\mathcal{S}=\mathcal{Z}^{n}.$
		\item Unconfoundedness:   $(Y_{1},Y_{0},C_{1}, C_{0})\bot D|X$.
		\item Bounded Outcomes and Costs:  There exist positive $M_{y}, M_{c}<\infty$ such that the support of Y is contained in $[-M_{y}/2, M_{y}/2]$ and the support of $C$ is contained in $[-M_{c}/2, M_{c}/2]$.
		\item Strict overlap: Define
		$e(X) = E_{P}[D|X],$
		where $E_{P}(\cdot)$ is the expectation with respect to $P$.\footnote{Similarly, we denote expectation with respect to $Q$ by $E_{Q}(\cdot)$. Expectation with respect to the distribution of the sample, $P^{\otimes n}$, will be denoted $E_{P^{n}}(\cdot)$.}
		It is assumed that there exists $\kappa\in (0,1/2)$ such that $e(x)\in[\kappa, 1-\kappa]$ for all $x\in \mathcal{X}$.
	\end{enumerate}
\end{assumption} 

Assumption \ref{Assumption: treatment identification and boundedness} mirrors treatment assumptions in \cite{KT2018} and \cite{mbakop2021model} and also includes similarly-formulated conditions for cost-related variables.  Unconfoundedness states that, conditional on the covariates, the potential outcomes and costs are independent of the treatments assigned to the observed data.  This and strict overlap will hold in randomized controlled trials (RCTs) which is our primary setting of interest.  As such, we assume $e(x)$ is known.  It is possible to adjust our procedures to a setting where $e(x)$ is estimated similarly to the e-hybrid rules utilized in \cite{KT2018} and \cite{mbakop2021model} while maintaining some of the theoretical motivations considered in Section \ref{sec: PAC Analysis}.   We leave a complete exploration of this topic to future research and work under the presumption that  $e(x)$ is known.   

Define the conditional average treatment effect (CATE) and the conditional average treatment cost (CATC), respectively, by
\begin{equation}
\label{eqns: CATE and CAC}
\delta_{y}(x)\equiv E_{Q}[Y_{1}-Y_{0} \rvert X=x], \ \  \delta_{c}(x) \equiv E_{Q}[C_{1}-C_{0}|X=x].
\end{equation}
Assumption \ref{Assumption: treatment identification and boundedness} (iii) implies that $|\delta_{y}(X)|$ and $|\delta_{c}(X)|$ are bounded almost surely by $M_{y}$ and $M_{c}$, respectively.  Our procedures can be implemented without knowledge of $M_{y}$ or $M_{c}$ and several of the motivating regret bounds in Section \ref{sec: PAC Analysis} could be derived in slightly altered forms if instead we required that objects related to $|\delta_{y}(X)|$ and $|\delta_{c}(X)|$ are sub-Guassian or even sub-exponential with additional constraints on a hyper-parameter. Assumption \ref{Assumption: treatment identification and boundedness} (iii) is typically a mild requirement that is often adopted in the treatment and classification literature; here it simplifies our exposition and path to generalization bounds.  Note that $Y$ and $C$ may belong to any interval.  The upper and lower bounds are taken to be symmetric around zero for convenience and  without loss of generality. 

In section \ref{sec: PAC-Bayesian Setting} we propose treatment assignment rules that aim to balance two prevailing objectives.  We seek rules that will maximize the expected outcome  $Y$  while also accounting for a potential budget constraint when we anticipate that resource, policy, or other limitations may preclude treating everyone with a positive CATE. Our proposed rules contain a parameter $u$, which can be chosen in a data-dependent manner, that modulates how much the second (budgetary) objective is prioritized.  In particular, any choice of $u$  corresponds to a different maximum expected cost in a budget-constrained welfare optimization problem.  Before describing the treatment model and empirical approach, we first state the policy maker's problem at the population level under a given maximum budget $B$ if the distribution $Q$ were known.  

The policy maker's goal is to obtain a treatment rule that maximizes welfare subject to a budget or quantity constraint.  The treatment rule is intended for application to a target population wherein the joint distribution of $(Y_{0}, Y_{1}, C_{0}, C_{1},X)$ follows that associated with $Q$.  We will consider stochastic treatment assignment rules,  defining such a rule as a measurable map $f:\mathcal{X}\rightarrow [0,1]$ from the covariate space to a treatment assignment probability.  If $f(x)\in\{0,1\}$, the treatment assignment for $x$ is non-random.  If $0<f(x)<1$, treatment is assigned randomly with treatment probability $f(x)$.   

The utilitarian welfare associated with $f$ is given by 
\begin{align}
\label{Definition: utilitarian welfare}
E_{Q}[Y_{1}f(X)+Y_{0}(1-f(X))].
\end{align}
This is the expected value of $Y$ when treatment is administered according to $f(X)$.  Dropping terms that do not vary with $f$, the policy maker's objective function evaluated at $f$ is defined by 
\begin{equation}
\label{Definition: policy maker's objective}
W(f)\equiv E_{Q}[(Y_{1}-Y_{0})f(X)].
\end{equation}
Choosing $f$ that maximizes $W(f)$ is equivalent to choosing $f$ that maximizes utilitarian welfare.  Thus we will refer to $W(f)$ as the welfare associated with $f$.   Note that by the law of iterated expectations, 
$W(f)=E_{Q}[\delta_{y}(X)f(X)]$.  Next, define the expected cost of $f$ by
\begin{equation}
\label{Definition: cost differential}
K(f) \equiv E_{Q}\left [ (C_{1}-C_{0})f(X) \right ],
\end{equation}
which can similarly be written $K(f)=E_{Q}[\delta_{c}(X)f(X)]$. 
Given a budget constraint $B$, the policy maker's problem is to identify
\begin{equation}
\label{eqn: policy makers budget problem}
f^{*}_{B}\in \underset{f}{\arg \max} \left \{ W(f): K(f) \leq B \right \},
\end{equation}
where the maximization is taken over all measurable functions from $\mathcal{X}$ to $[0,1]$.

Note that $K(f)=E_{Q}[C_{1}f(X)+C_{0}(1-f(X))]- E_{Q}[C_{0}].$  The budget constraint states that the expected additional cost due to implementing treatment policy $f$, that beyond what would be expected if treatment were never assigned, cannot exceed $B$. This is flexible, as it allows for cost savings (i.e. when $C_{1}<C_{0}$ with positive probability) to be factored into the budget.  Provided such savings are possible, a policy maker could be interested in, for example, $B = 0$.  In this scenario the policy maker is looking for treatment policies that may improve welfare without increasing the expected cost beyond the setting were no treatments are administered.  On the other hand, if the policy maker has a fixed budget allocated to treatments and cost savings do not feed back into the budget, one can simply define $C_{0}=0$, so that the observed $C$ is equal to the cost of treatment when treatment is provided and is zero otherwise.  If there is a a fixed quantity constraint consisting of a set number of treatments and no other budgetary concerns, one can set $C_{0}=0$ and $C_{1}=1$ so that the observed $C$ is the treatment indicator.  In this case $B$ denotes the maximum proportion of the target population for which treatments are available.  

If there is no budget constraint and the policy maker is able to choose any measurable $f:\mathcal{X}\rightarrow [0,1]$, it is straightforward to verify that an optimal treatment allocation rule is given by 
\begin{equation}
\label{Optimal treatment rule when no budget constraint}
f^{*}(x) = 1\{\delta_{y}(x)>0\}.
\end{equation}
$f^{*}$ assigns treatment to any unit with a positive CATE.
Here, and throughout the paper, the indicator function $1\{A\}$ takes the value $1$ if event $A$ occurs and is zero otherwise.  Given a particular budget constraint $B$, a solution to the policy maker's problem is characterized in the following theorem. 

\begin{theorem}
	\label{Theorem: theoretically optimal policy choice}
	Let $(Y_{0},Y_{1}, C_{0}, C_{1},X)$ be distributed according to $Q$.  Assume that $E_{Q}|\delta_{y}(X)|<\infty$,  $ E_{Q}|\delta_{c}(X)|<\infty$, and that $B> E_{Q}[\delta_{c}(X)1\{\delta_{c}(X)<0\}]$.   Then there exist constants $\eta_{B}\geq 0$ and $a_{1},a_{2}\in[0,1]$ such that
	\begin{equation}
	f^{*}_{B}(x)  =%
	\begin{cases}
	0 &
	\mathrm{if}\ \delta_{y}(x) < \eta_{B} \delta_{c}(x), \\
	a_{1} 1\{\delta_{c}(x)>0\} +a_{2}1\{\delta_{c}(x)<0\} & \mathrm{if} \ \delta_{y}(x)=\eta_{B} \delta_{c}(x), \\
	1 & \mathrm{if} \ \delta_{y}(x) > \eta_{B} \delta_{c}(x),
	\end{cases}
	\end{equation}
	satisfies \eqref{eqn: policy makers budget problem}.  In particular, if $K(f^{*})\leq B$, then one can take $\eta_{B}=a_{1}=a_{2}=0$ and $f_{B}^{*}=f^{*}$; if $K(f^{*})> B$ then  $(\eta_{B},a_{1},a_{2})$ are chosen such that $K(f_{B}^{*})=B$.  If $E_{Q}[1\{\delta_{y}(X)=\eta_{B} \delta_{c}(X)\}]=0$, $f_{B}^{*}$ is deterministic and is the unique budget-constrained, welfare-optimizing policy in the sense that for any $f'$ satisfying \eqref{eqn: policy makers budget problem} it holds that $f'(X)=f^{*}_{B}(X)$ a.s.   
\end{theorem}

The choice of $\eta_{B}$ in Theorem \ref{Theorem: theoretically optimal policy choice} is unique, however in general there may be different choices of $a_{1}, a_{2}$ that produce optimal rules when  $E_{Q}[1\{\delta_{y}(X)=\eta_{B} \delta_{c}(X)\}]\neq 0$.  Apart from this difference, Theorem \ref{Theorem: theoretically optimal policy choice} is a generalization of a result in \cite{sun2021treatment} which restricts itself to the setting where $C_{1}\geq C_{0}$ almost surely.  In practice, of course, $Q$ is unknown to the researcher who must estimate a suitable model $f$ empirically.  Section \ref{sec: PAC-Bayesian Setting} introduces the PAC-Bayesian setting for the empirical strategy we employ.

When $E_{Q}[1\{\delta_{y}(X)=\eta_{B}\delta_{c}(X)\}]=0$, for example when $\delta_{y}(X)$ and $\delta_{c}(X)$ have bounded densities, Theorem \ref{Theorem: theoretically optimal policy choice} says the optimal treatment rule is deterministic and unique in terms of the resulting treatment decisions.  However, the function $\delta_{y}(x)-\eta_{B}\delta_{c}(x)$ in the optimal rule in this setting, given by
\[f^{*}_{B}(x)= 1\{\delta_{y}(x)-\eta_{B}\delta_{c}(x)>0\},\]
is not unique.  Any measurable function $m(x):\mathcal{X}\rightarrow\mathbb{R}$ that satisfies
\[\mathrm{sign}\left [ m(x) \right ] = \mathrm{sign}\left [ \delta_{y}(x)-\eta_{B}\delta_{c}(x) \right ],\]
yields an optimal treatment rule via $f_{m}(x)=1\{ m(x)>0\}$.  This situation is similar to that in the binary forecasting problem (cf. \cite{elliott2013predicting}) and is illustrated in Figure \ref{fig:plotnonuniquefinal}.

\begin{figure}[H]
	\centering
	\includegraphics[width=0.8\linewidth]{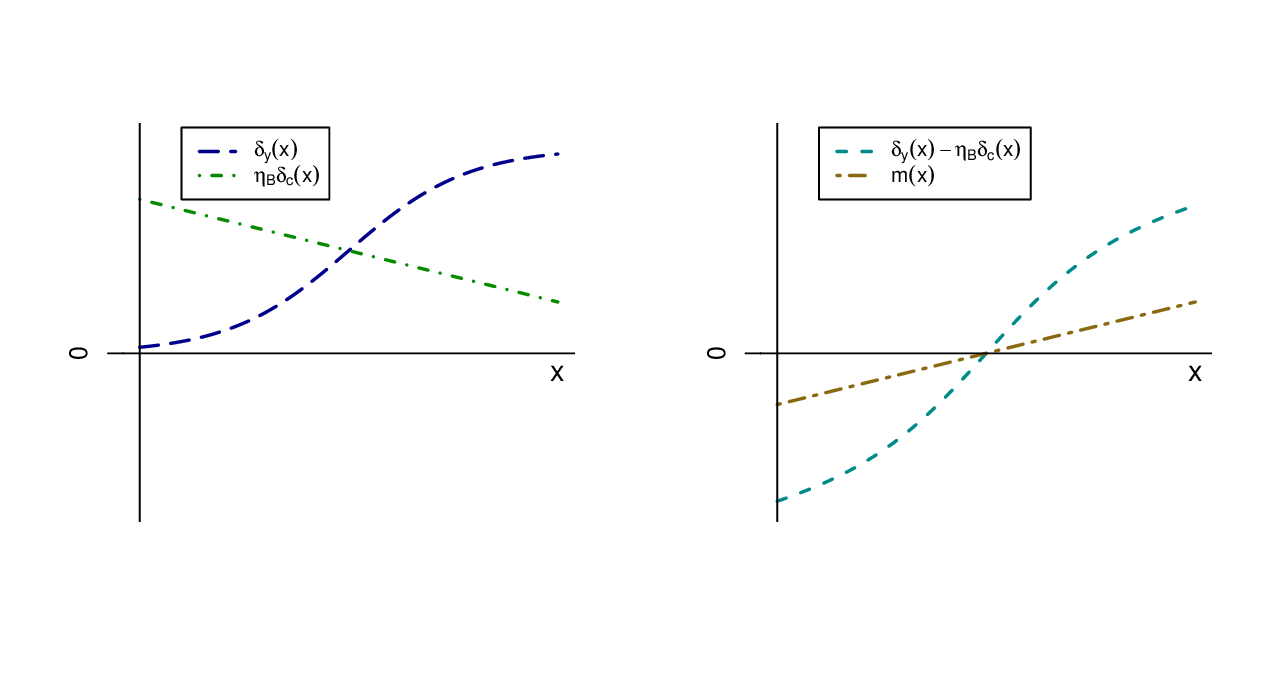}
	\vspace*{-10mm}
	\caption[Non-uniqueness of the optimal treatment rule]{On the left, a plot of $\delta_{y}(x)$ and $\eta_{B}\delta_{c}(x)$ in a simple setting with a single crossing point and a single covariate.  On the right, the corresponding $\delta_{y}(x)-\eta_{B}\delta_{c}(x)$ is plotted along with a second function, $m(x)$.  Here, $m(x)$ differs from $\delta_{y}(x)-\eta_{B}\delta_{c}(x)$ everywhere except at the crossing point yet $1\{m(x)>0\}$ and $1\{\delta_{y}(x)-\eta_{B}\delta_{c}(x)\}=f^{*}_{B}(x)$ yield identical treatment decisions.}
	\label{fig:plotnonuniquefinal}
\end{figure}

In Section \ref{sec: PAC-Bayesian Setting}, we propose treatment rules that aggregate over a user-specified family of treatment rules in a way that is weighted towards models with high empirical budget-penalized welfare.  There, we introduce Gibbs treatment rules, which aggregate over the rule family to derive a treatment probability, and related majority vote rules which aggregate over the rule family to assign treatment directly.  Aside from the desirable theoretical properties derived in Section \ref{sec: PAC Analysis}, some intuition behind such an approach is as follows. Two functions $\hat{m}(x)$ and $\hat{m}^{*}(x)$, with corresponding treatment rules $1\{\hat{m}(x)>0\}$ and $1\{\hat{m}^{*}(x)>0\}$, respectively, could yield identical or very similar treatment decisions over the sample covariate values.  In a setting where different rules may have the same or very similar observable properties, it is reasonable to aggregate or average over rules with high empirical welfare.  Rather than trying to select a single solution, we take the identification issue above as motivation for an ensemble approach. 

\subsection{Empirical Approach and PAC-Bayesian Setting  \label{sec: PAC-Bayesian Setting}}

Underpinning the treatment rules we will consider is a family of non-stochastic treatment rules, indexed by $\theta\in\Theta$, denoted   
\begin{equation}
\label{Definition: class of models}
\mathcal{F}_{\Theta}=\{f_{\theta}(x):\mathcal{X}\rightarrow \{0,1\}; \theta\in\Theta\}.
\end{equation}
For a concrete example, we could let $\{\phi_{1}(x),\dots, \phi_{q}(x)\}$ be a set of feature transformations where $\phi_{j}(x):\mathcal{X}\rightarrow \mathbb{R}$ for $j=1,\dots, q$.  Denoting $\phi(x)=(\phi_{1}(x),\dots,\phi_{q}(x))^{\intercal}$, we could then have 
\begin{equation}
\label{Example treatment model class}
f_{\theta}(x)=1\{\phi(x)^{\intercal}\theta>0\} \ \mathrm{for} \ \theta\in\Theta=\mathbb{R}^{q},
\end{equation}
where $q\in \mathbb{N}$ need not be equal to $d_{x}$, the dimension of $\mathcal{X}$.

For any treatment assignment rule $f$, we define the welfare regret relative to the first-best prediction rule $f^{*}$ in \eqref{Optimal treatment rule when no budget constraint} by 
\[R(f) \equiv W\left ( f^{*} \right ) - W\left (f \right ).  \]
Note that $R(f)$ is defined relative to the first-best treatment assignment without a budget constraint. We can also define 
\begin{equation}
\label{Definition: R_B(f)}
 R_{B}(f)\equiv W(f_{B}^{*})-W(f),
\end{equation}
the welfare-regret under a maximum expected budget of $B$ where $f_{B}^{*}$ is defined in Theorem \ref{Theorem: theoretically optimal policy choice}.  With simple manipulations, the oracle-type inequalities involving $R(f)$ in Sections \ref{subsec: Regret Bounds and Oracle Inequalities} and \ref{subsec: Normal Prior} apply to $R_{B}(f)$ rather than $R(f)$.  For simplicity, we will mostly work with $R(f)$ which is non-negative.  Note that $R_{B}(f)$ is only non-negative when attention is constrained to treatment rules with a maximal budget $B$.    For particular models $f_{\theta}\in\mathcal{F}_{\Theta}$, with a slight abuse of notation, we will write
\[R(\theta)  \equiv  R\left ( f_{\theta} \right ), \ W(\theta) \equiv W\left ( f_{\theta} \right ), \ \mathrm{and} \ K(\theta)\equiv K\left ( f_{\theta} \right ). \]

Under the unconfoundedness and strict overlap conditions of Assumption \ref{Assumption: treatment identification and boundedness}, it holds that
\[W(f) = E_{Q} \left [ \left ( Y_{1}-Y_{0} \right ) f(X) \right ] = E_{P} \left [ \left ( \frac{YD}{e(X)}-\frac{Y(1-D)}{1-e(X)} \right )f(X) \right ]. \]
A similar statement can be written for $K(f)$, now with $C$ in place of $Y$.    
Defining 
\[\delta_{y,i}= \left ( \frac{Y_{i}D_{i}}{e(X_{i})}-\frac{Y_{i}(1-D_{i})}{1-e(X_{i})} \right ) \ \mathrm{and} \ \delta_{c,i} = \left ( \frac{C_{i}D_{i}}{e(X_{i})}-\frac{C_{i}(1-D_{i})}{1-e(X_{i})} \right ), \]
the (unbiased) empirical counterparts of $W(f)$, $R(f)$, and $K(f)$, along with their notation for $f_{\theta}\in\mathcal{F}_{\Theta}$, are given by
\begin{align}
\notag 
&W_{n}(f) \equiv \frac{1}{n} \sum_{i=1}^{n}\delta_{y,i}f\left (X_{i} \right ), \ \ &&W_{n}(\theta) \equiv \frac{1}{n} \sum_{i=1}^{n}\delta_{y,i}f_{\theta}\left (X_{i} \right ),
\\
\notag 
&R_{n}(f) \equiv \frac{1}{n} \sum_{i=1}^{n}\delta_{y,i}\left ( f^{*} \left ( X_{i} \right )- f\left (X_{i} \right ) \right ), \ \ &&R_{n}(\theta) \equiv \frac{1}{n} \sum_{i=1}^{n}\delta_{y,i}\left ( f^{*} \left ( X_{i} \right )- f_{\theta}\left (X_{i} \right ) \right ),
\\
\notag 
&K_{n}(f) \equiv  \frac{1}{n} \sum_{i=1}^{n}\delta_{c, i} f(X_{i}), \ \ && K_{n}(\theta) \equiv  \frac{1}{n} \sum_{i=1}^{n}\delta_{c, i} f_{\theta}(X_{i}).
\end{align}
As $f^{*}$ is unknown, the empirical regret $R_{n}(f)=W_{n}(f^{*})-W_{n}(f)$ or $R_{n}(\theta)$ for $\theta\in\Theta$ cannot be evaluated in practice.  $R_{n}(\theta)$ will arise in our analysis only as a theoretical object in relation to $R(\theta)$.  We stress that the treatment assignment rules we consider can be expressed solely in terms of $W_{n}(\theta)$.

$\mathcal{F}_{\Theta}$ consisting of treatment rules of the form in \eqref{Example treatment model class} will be considered in Sections \ref{subsec: Normal Prior} and \ref{sec: Simulation and Implementation}.  In general, to accommodate broader treatment rule model families, we make the following technical assumptions.   
\begin{assumption}
	\label{Assumption: measurability}
	(i)  We assume that $(\Theta, \mathcal{B}_{\theta})$ is a standard Borel space.  (ii)  We assume that $\mathcal{F}_{\Theta}$ is such that the maps $(S,\theta)\mapsto R_{n}(\theta):\mathcal{S}\times \Theta \rightarrow \mathbb{R}$ and $(S,\theta)\mapsto K_{n}(\theta):\mathcal{S}\times \Theta \rightarrow \mathbb{R}$ are measurable.
\end{assumption}

We now introduce the stochastic treatment rules of interest.  Let $\mathcal{P}(\Theta)$ be the set of probability measures on $(\Theta,\mathcal{B}_{\theta})$ and, for any $\pi\in\mathcal{P}(\Theta)$, let  $\mathcal{P}_{\pi}(\Theta)= \{\rho\in\mathcal{P}(\Theta): \rho \ll \pi \}$.  That is, $\mathcal{P}_{\pi}(\Theta)$ is the set of probability measures on $(\Theta,\mathcal{B}_{\theta})$ that are absolutely continuous with respect to $\pi$. Rather than selecting a single value $\hat{\theta}\in\Theta$, for example that which maximizes $W_{n}(\theta)$, and then assigning treatment via $f_{\hat{\theta}}$, we seek probability measures $\rho\in \mathcal{P}(\Theta)$ from which we form stochastic treatment rules.  Borrowing nomenclature from the classification literature, we work with  Gibbs treatment rules.  For $\rho\in\mathcal{P}(\Theta)$, the Gibbs treatment rule or method associated $\rho$, denoted $f_{G,\rho}:\mathcal{X}\rightarrow [0,1]$,  is defined by 
\[f_{G,\rho}(x) = \int_{\Theta}f_{\theta}(x) d\rho(\theta), \  x\in\mathcal{X}.\]

Assigning treatments via the Gibbs method is equivalent to assigning treatments as follows.  For an individual with covariates $X$, a parameter value $\theta_{\circ}$ is drawn randomly according to $\rho$, i.e. $\theta_{\circ}\sim\rho$.  Then, $f_{\theta_{\circ}}(X)\in\{0,1\}$ determines the treatment assignment.  This process, with an independent draw from $\rho$, is repeated each time treatment is to be assigned.  Note that, exchanging the order of integration, we can write
\[R(f_{G,\rho})= \int_{\Theta} R(\theta) d\rho(\theta) \ \ \mathrm{and} \  \ R_{n}(f_{G,\rho})= \int_{\Theta} R_{n} (\theta) d\rho(\theta),\]
which is called the Gibbs risk associated with $\rho$.  Similarly, the expected cost of $f_{G,\rho}$ and its empirical counterpart can be written
\[K(f_{G,\rho})= \int_{\Theta} K(\theta) d\rho(\theta) \ \ \mathrm{and} \ \ K_{n}(f_{G,\rho})= \int_{\Theta} K_{n} (\theta) d\rho(\theta). \]
We will frequently be concerned with the cost or empirical cost associated with a Gibbs treatment rule utilizing some $\rho\in\mathcal{P}_{\pi}(\Theta)$.  To simplify the exposition, we denote  
\begin{equation}
\label{Definition:  B(rho) and Bhat(rho)}
B(\rho) \equiv   K\left ( f_{G,\rho} \right ), \ \ \mathrm{and} \ \ \widehat{B}(\rho) \equiv  K_{n}\left ( f_{G,\rho} \right ). 
\end{equation}

A non-stochastic treatment rule that is closely related to the Gibbs rule is the so-called majority vote or Bayes method associated with $\rho\in\mathcal{P}(\Theta)$.  This is given by 
\begin{equation}
\label{Definition: Majority vote treatment rule}
f_{\mathrm{mv},\rho}(x) = 1\left \{ \int_{\Theta} f_{\theta}(x) d\rho(\theta) > \frac{1}{2} \right \}, \ x\in\mathcal{X}.
\end{equation}
In practice, majority vote rules can deliver treatment rules that are numerically more stable than their Gibbs counterpart.  If $\rho=\alpha\rho_{1}+(1-\alpha)\rho_{2}$ for some $\rho_{1},\rho_{2}\in\mathcal{P}(\Theta)$ and constant $\alpha$, then $R(f_{G,\rho})=\alpha R(f_{G,\rho_{1}})+(1-\alpha)R(f_{G,\rho_{2}})$.  That is, the Gibbs risk is a linear functional of $\rho$.  This linearity makes the Gibbs risk and Gibbs treatment rules more amenable to theoretical analysis.  Our analysis will therefore focus on a family of Gibbs treatment rules.  However, in Section \ref{subsec: The Majority Vote TR}, we show that the majority vote treatment rule associated with our Gibbs rules of interest inherit desirable properties from their Gibbs counterparts.  In practice, either method is an acceptable choice and we consider both in our simulation study in Section \ref{sec: Simulation and Implementation}.

In particular, we propose to utilize Gibbs treatment rules constructed from  data-dependent\footnote{In general, by data-dependent probability measures on $(\Theta,\mathcal{B}_{\theta})$ we mean regular conditional probability measures (RCPMs):   letting $\mathcal{B}_{s}$ denote the $\sigma$-algebra associated with the sample space $\mathcal{S}$,   $\rho(S,\cdot)$ is an RCPM on $(\Theta,\mathcal{B}_{\theta})$ if (i) for any fixed $A\in\mathcal{B}_{\theta}$, the map $S\mapsto \rho(S,A): (\mathcal{S},\mathcal{B}_{s}) \rightarrow \mathbb{R}_{+}$ is measurable; and (ii) for any $S\in\mathcal{S}$, the map $A\mapsto \rho(S,A):\mathcal{B}_{\theta}\rightarrow [0,1]$ is a probability measure.  For additional measure-theoretic details, for example the decomposition and measurability of the Kullback-Leibler divergence (utilized throughout the paper) between RCPMs, we refer the reader to \cite{catoni2004statistical}, in particular Proposition 1.7.1 and its proof on pages 50-54.} probability measures of the form $\hat{\rho}_{\lambda, u}$ defined below. 

\begin{definition}
	\label{Definition: optimal rho hat under a budget constraint}
	For $\lambda >0$, $u\geq 0$, and a reference measure $\pi\in\mathcal{P}(\Theta)$,  define $\hat{\rho}_{\lambda, u}$ to be the (random) probability measure on $\Theta$ with the following Radon-Nikodym (RN) derivative with respect to $\pi$:
	\begin{align*}
	\frac{d\hat{\rho}_{\lambda, u}}{d\pi}(\theta) &= \frac{\exp\left [ -\lambda \left ( R_{n}(\theta) +u K_{n}(\theta) \right ) \right ] }{\int_{\Theta} \exp\left [  -\lambda \left ( R_{n} \left ( \tilde{\theta} \right )+ u K_{n}\left ( \tilde{\theta} \right ) \right ) \right ] d\pi\left ( \tilde{\theta} \right ) }
	\\
	&=\frac{\exp\left [ - \lambda \left ( uK_{n}(\theta) -W_{n}(\theta) \right ) \right ] }{\int_{\Theta} \exp\left [  - \lambda \left ( uK_{n}\left ( \tilde{\theta} \right ) - W_{n} \left ( \tilde{\theta} \right ) \right ) \right ] d\pi\left ( \tilde{\theta} \right ) }.
	\end{align*}
	Define $\rho^{*}_{\lambda, u}$ to be the probability measure on $\Theta$ with the following RN derivative with respect to $\pi$:
	\begin{equation}
	\notag 
	\frac{d\rho^{*}_{\lambda,u}}{d\pi}(\theta) = \frac{\exp\left [ -\lambda \left ( R(\theta) +uK(\theta) \right ) \right ] }{\int_{\Theta}\exp \left [ -\lambda\left (  R\left ( \tilde{\theta} \right ) + u K\left (  \tilde{\theta} \right ) \right ) \right ] d\pi \left (\tilde{\theta} \right )}.
	\end{equation}
\end{definition}

$\hat{\rho}_{\lambda, u}$ is sometimes called a Gibbs posterior distribution or a Boltzmann distribution.   As $\lambda\rightarrow\infty$, $\hat{\rho}_{\lambda, u}$  concentrates around the value of $\theta$ such that $f_{\theta}$ minimizes the budget-penalized empirical regret criterion  $R_{n}(f_{\theta})+ uK_{n}(f_{\theta})$.  Equivalently, it concentrates around the value of $\theta$ the maximizes $W_{n}(\theta) -uK_{n}(\theta)$ over $\Theta$.  This reduces to the empirical welfare maximizer when $u=0$.  In general, $\hat{\rho}_{\lambda, u}$ assigns higher probability to regions of the parameter or model space with low budget-penalized empirical regret.  $u$ modulates the trade off between emphasis on low regret vs expected cost.  As subsequent analysis will show, different choices of $u$ correspond in a one-to-one manner with different budget constraints.  We will consider the setting where $u$ is cross-validated and the setting where it is determined by a particular choice of a budget constraint parameter $B$.  $\rho^{*}_{\lambda, u}$ is a theoretical counterpart to $\hat{\rho}_{\lambda, u}$ that will be useful when we analyze statistical properties related to $\hat{\rho}_{\lambda, u}$.  $\lambda$ is typically chosen via cross-validation while choices where $\lambda = \mathcal{O}(\sqrt{n})$ will yield optimal or near-optimal rates of convergence in Section \ref{sec: PAC Analysis}. 

In the PAC-Bayesian literature, probability measures over the model or parameter space that are traditionally chosen independently of the sample are often called prior probability measures.  In our setting, the choice of $\pi$ utilized in Definition \ref{Definition: optimal rho hat under a budget constraint} will fall into this category.  Probability measures utilized for treatment or prediction,  such as $\hat{\rho}_{\lambda, u}$, are called posterior distributions.  However, this nomenclature does not have the same connotation as in traditional Bayesian methodology.  While knowledge of the DGP could allow for a prior to be chosen that improves the performance of rules suggested from PAC-Bayesian analysis, often the prior is taken to be uniform or normal centered at the origin.  Additionally, the posterior, for example, does not need to be proportional to a likeilihood function.  The statistical analysis itself is frequentist in nature.  The role and choice of $\pi$ will be discussed further later in the paper.  For now we make the following assumption.

\begin{assumption}
	\label{Assumption: prior indep of data}
	$\pi\in\mathcal{P}(\Theta)$ is a (deterministic) probability measure that does not depend on the sample.
\end{assumption}

\subsection{Initial Properties of the Gibbs Posterior  \label{sec: Gibbs Posterior Section}}
Here we derive initial properties of $\hat{\rho}_{\lambda, u}$ that link the choice of $u$ to a particular budget constraint.  These provide intuition behind Definition \ref{Definition: optimal rho hat under a budget constraint} and are utilized in proving the results of Section \ref{sec: PAC Analysis}. 

Let $D_{\mathrm{KL}}(\rho,\pi)$ 
denote the Kullback–Leibler (KL) divergence between $\rho,\pi \in\mathcal{P}(\Theta)$, 
\begin{equation}
D_{\mathrm{KL}}\left(  \rho,\pi\right)  =%
\begin{cases}
\int_{\Theta}\log\left[  \frac{d\rho}{d\pi}(\theta)\right]  d\rho(\theta), &
\mathrm{if}\ \rho\ll\pi\\
\infty, & \mathrm{else}.
\end{cases}
\nonumber
\end{equation}
Suppose the policy maker has a maximum expected budget of $B\in\mathbb{R}\cup \{\infty\}$, where $B=\infty$ is the unconstrained setting.  If the data generating process were known, among Gibbs treatment rules we would be interested in a solution to
\begin{equation}
\label{Optimization problem: Gibbs if Q known}
\min_{\rho\in\mathcal{P}(\Theta)} \int_{\Theta} R(\theta) d\rho(\theta),\ \mathrm{subject \ to \ } \ \int_{\Theta} K(\theta) d\rho(\theta) \leq B.
\end{equation}
In practice, we will instead focus on a subset $\mathcal{P}_{\pi}(\Theta) \subset \mathcal{P}(\Theta)$ and solve the following empirical problem:
\begin{align}
\label{Optimization problem: sample, with budget}
&\underset{  \rho\in\mathcal{P}_{\pi}(\Theta)  }{\min } \left [ \int_{\Theta}R_{n}(\theta)d\rho ( \theta ) +\frac{1}{\lambda}D_{\mathrm{KL}}\left ( \rho,\pi  \right ) \right ], \ \ \mathrm{subject \ to \ }  \int_{\Theta}K_{n}(\theta)d\rho(\theta) \leq B.
\end{align}

\eqref{Optimization problem: sample, with budget} includes a regularization term in the form of $D_{\mathrm{KL}}(\rho,\pi)$, discouraging any choice for $\rho$ that has a large KL divergence from the reference measure $\pi$.  In practice,  $\mathcal{P}_{\pi}(\Theta)$ is flexible and optimal choices for $\lambda$ will entail $\lambda\rightarrow\infty$ as $n\rightarrow \infty$.  When adapted to our setting, Lemma \ref{Lemma constrained KL} below shows that, provided a feasibility or Slater condition holds, for some value $\hat{u}\geq 0$, $\hat{\rho}_{\lambda,\hat{u}}$ is the solution to \eqref{Optimization problem: sample, with budget}.  Of course, appearing to be a reasonable empirical counterpart of \eqref{Optimization problem: Gibbs if Q known} is not, in and of itself,  justification for $f_{G,\hat{\rho}_{\lambda, u}}$.  In Section \ref{sec: PAC Analysis} we provide additional theoretical motivation for  $f_{G,\hat{\rho}_{\lambda, u}}$, comparing it to alternative Gibbs rules and optimal (non-stochastic) models in $\mathcal{F}_{\theta}$.  

The following lemma yields solutions to \eqref{Optimization problem: sample, with budget} and a theoretical counterpart when $R_{n}(\theta)$ and $K_{n}(\theta)$ are replaced by $R(\theta)$ and $K(\theta)$, respectively.

\begin{lemma}
	\label{Lemma constrained KL}
	Let $\pi\in\mathcal{P}(\Theta)$, $\lambda>0$,  $B\in\mathbb{R}\cup \{\infty \}$, and let $A(\theta)$ and $H(\theta)$ be bounded, measurable functions defined on $(\Theta,\mathcal{B}_{\theta})$.  For $u\geq 0$, define
	$\tilde{\rho}_{A, H, \lambda, u}\in \mathcal{P}_{\pi}\left ( \Theta \right )$
	to be the probability measure with RN derivative with respect to $\pi$ given by
	\[\frac{d\tilde{\rho}_{A, H, \lambda, u}}{d\pi }(\theta) = \frac{\exp \left [ -\lambda \left ( A \left ( \theta \right ) + u H \left (\theta   \right ) \right ) \right ]}{\int_{\Theta} \exp \left [ -\lambda \left ( A \left ( \tilde{\theta} \right ) + u H \left ( \tilde{\theta}   \right ) \right ) \right ] d\pi \left ( \tilde{\theta} \right )}.\]
	Lastly, define 
	\[\Lambda(u) = \int_{\Theta}H(\theta) d \tilde{\rho}_{A, H, \lambda, u}\left  ( \theta \right ), \ u\geq 0, \ \ \mathrm{and} \ \ \mathcal{E}_{H,B}= \left \{\rho\in\mathcal{P}_{\pi}\left ( \Theta \right ) : \int_{\Theta}H(\theta)d\rho(\theta)\leq B  \right \} .\]
	
	We have the following result.  If   
	\begin{equation}
	\label{eqn: condition for constrained KL solution to hold}
	\pi \left ( \left \{ \theta: H\left ( \theta \right )< B  \right \} \right )>0,
	\end{equation}
	then, 
	\begin{equation}
	\label{Lemma constrained KL: minimization problem}
	\tilde{\rho}_{A,H,\lambda, \overline{u}_{B}} = \underset{\mathcal{E}_{H,B}}{\arg\min } \left [ \int_{\Theta} A\left ( \theta \right ) d\rho(\theta) +\frac{1}{\lambda} D_{\mathrm{KL}}(\rho,\pi) \right ],
	\end{equation}
	where $\overline{u}_{B}=0$ if $\Lambda(0) \leq B$ and otherwise, when $\Lambda(0)>B$,   $\overline{u}_{B}>0$ is the unique positive real number satisfying $ \Lambda(\overline{u}_{B})=B$.   Additionally\footnote{Throughout, we adopt the convention that $0\cdot -\infty = 0$ when $B=\infty$ in statements of this form.},
	\begin{align}
	\notag 
	&\int_{\Theta}A(\theta) d\tilde{\rho}_{A,H,\lambda, \overline{u}_{B}}(\theta)+\frac{1}{\lambda}D_{\mathrm{KL}} \left (\tilde{\rho}_{A,H,\lambda, \overline{u}_{B}},\pi \right ) 
	\\
	\label{eqn: sup useful for expanding the budget}
	&= \sup_{u\geq 0} \left [  \int_{\Theta}A(\theta) d\tilde{\rho}_{A,H,\lambda, u}(\theta) + u\left ( \int_{\Theta}H(\theta)d\tilde{\rho}_{A,H,\lambda, u}(\theta)-B \right ) +\frac{1}{\lambda}D_{\mathrm{KL}} \left (\tilde{\rho}_{A,H,\lambda, u},\pi \right )  \right ]. 
	\end{align}
\end{lemma}

When $B=\infty$, so that $\overline{u}_{B}=0$, the result in Lemma \ref{Lemma constrained KL} is a well known property that is commonly utilized in the PAC-Bayesian literature with $A(\theta)$ taken as some loss or regret function; see \cite{catoni2007pac} and \cite{alquier2016properties} among many possible examples.  Lemma \ref{Lemma constrained KL} extends this setting to accommodate a secondary constraint objective associated with $H(\theta)$.  When, for example $H(\theta)=R(\theta)$, $\Lambda(u)$ is the cost associated with the Gibbs treatment rule utilizing $\tilde{\rho}_{A,H,\lambda, u}$.  $\Lambda(u)$ is decreasing in $u$.  Intuitively, as the exponential re-weighting of $\pi$ depends more heavily on $H(\theta)$ for larger values of $u$, regions of the parameter or model space with greater cost receive a relatively lower weighting and the overall cost is reduced as $u$ increases.  Convex optimization problems where the objective or constraint set involves the Kullback-Liebler divergence have been considered in earlier work, for example in \cite{csiszar1975divergence}.  Rather than establishing Lemma \ref{Lemma constrained KL} from the more abstract setting there, the proof in the Appendix utilizes well known properties of the KL divergence, stated as Lemma \ref{Lemma KL} and Corollary \ref{Corollary KL} in the Appendix.  We note that Corollary \ref{Corollary KL} (b) is a well known change-of-measure inequality (c.f. \cite{csiszar1975divergence} and \cite{DonskerVaradhan1975}) that is widely utilized in deriving PAC-Bayesian generalization bounds.

The property in \eqref{eqn: sup useful for expanding the budget} is used in deriving the oracle-type inequalities in Section \ref{sec: PAC Analysis}.  The result states that the duality gap between the primal and dual of the minimization problem in \eqref{Lemma constrained KL: minimization problem} is  zero.  That is, 
\begin{align*}
&\min_{\rho\in\mathcal{P}_{\pi}(\Theta)} \sup_{u\geq 0} \left [ \int_{\Theta}A(\theta) d\rho(\theta) + \frac{1}{\lambda} D_{\mathrm{KL}}(\rho,\pi) +u\left ( \int_{\Theta}H(\theta)d\rho(\theta) - B \right ) \right ] 
\\
&= \sup_{u\geq 0} \min_{\rho\in\mathcal{P}_{\pi }(\Theta)} \left [ \int_{\Theta}A(\theta) d\rho(\theta) + \frac{1}{\lambda} D_{\mathrm{KL}}(\rho,\pi) +u\left ( \int_{\Theta}H(\theta)d\rho(\theta) - B \right ) \right ]. 
\end{align*}
Note that the left-hand side of the above equality, the primal problem, is equivalent to the optimization problem in \eqref{Lemma constrained KL: minimization problem}.  The right-hand side is the dual of this problem.  That the right-hand side above is equivalent to the expression on the right-hand side of \eqref{eqn: sup useful for expanding the budget} can be seen from a careful examination of \eqref{Lemma constrained KL: minimization problem} or from Corollary \ref{Corollary KL} (a) in the Appendix.  The condition in \eqref{eqn: condition for constrained KL solution to hold} constitutes a constraint qualification. 

We will apply Lemma \ref{Lemma constrained KL} with $A(\theta) = R_{n}(\theta)$ or $R(\theta)$ and $H(\theta)=K_{n}(\theta)$ or $K(\theta)$.  We consider two scenarios or perspectives.  In the first,  we have a (nonrandom) predetermined budget $B$ and utilize a corresponding,  sample dependent, choice of $\hat{u}$.  In the second scenario, we start from a predetermined, non-random choice of $u$ (or multiple values of $u$), which then corresponds to a sample dependent budget (or budgets) associated with $f_{G,\hat{\rho}_{\lambda, u}}$.  We will require the following assumptions in order to satisfy \eqref{eqn: condition for constrained KL solution to hold} in our analysis.  The first will correspond to the case with a predetermined $B$ while the second condition will be utilized when we start from predetermined $u$.

\begin{assumption}
	\label{Assumption: CQ}
	(i)  Let $B\in\mathbb{R}\cup \{\infty\}$ be a desired budget.  It is assumed that 
	\[\pi\left ( \theta\in\Theta: K(\theta) < B \right ) >0 \ \ \mathrm{and} \ \ \pi\left ( \theta\in\Theta : K_{n}(\theta) < B \right ) >0 \ \ P^{n} \ \mathrm{almost \ surely}. \]
	
	(ii)  It is assumed that  
	\[\mathbb{V}_{\theta\sim \pi}\left [ K(\theta) \right ]>0  \ \  \mathrm{and} \ \ \mathbb{V}_{\theta\sim \pi }\left [ K_{n}(\theta) \right ]>0 \ \ P^{n} \ \mathrm{almost \ surely} \]
	where, $\mathbb{V}_{\theta\sim\pi}$ denotes the variance of $K(\theta)$ when $\theta\sim \pi$ and,  for a fixed sample $S\in\mathcal{S}$, $\mathbb{V}_{\theta\sim \pi}[K_{n}(\theta)]$ denotes the variance of $K_{n}(\theta)$ when $\theta\sim \pi$.
\end{assumption}

Assumption \ref{Assumption: CQ} involves $\mathcal{F}_{\Theta}$, $\pi$ and the sampling distribution $P$.  Condition (i) requires that the budget of interest is not ruled out under the prior or reference measure $\pi$ and is not exactly at the boundary of theoretical or empirical feasibility.   With additional exposition, the condition that $\pi\left ( \theta\in\Theta : K_{n}(\theta) < B \right ) >0$ holds $P^{n}$ a.s. could be replaced by the condition that $\pi\left ( \theta\in\Theta : K_{n}(\theta) < B \right ) >0$ holds with high probability.  For example, with probability at least $1-\xi$, for some $\xi\in [0,1)$.  In this case the theorems in Section \ref{sec: PAC Analysis} will remain valid except that the high probability bounds there, that hold with probability at least $1-\epsilon$ for $\epsilon \in (0,1]$, will now hold with probability at least $1-\epsilon -\xi$.  Condition (ii) requires that there is always variation in actual and empirical costs within models in $\mathcal{F}_{\Theta}$ drawn by $\pi$.  

Given Lemma \ref{Lemma constrained KL} and the assumption above, the following definition will be relevant when the analysis starts with a predetermined budget $B$ for which we must find an appropriate value of $u$.

\begin{definition}
	\label{Def: u hat and u star}
	Let $\hat{\rho}_{\lambda, u}$ and $\rho^{*}_{\lambda, u}$ be defined with $\pi\in\mathcal{P}(\Theta)$ as in Definition \ref{Definition: optimal rho hat under a budget constraint}.  For $B\in\mathbb{R}$, define $\hat{u}(B,\lambda)$ by 
	\[\hat{u}(B, \lambda) = \underset{u\geq 0}{\arg\max } \int_{\Theta} R_{n}(\theta) d\hat{\rho}_{\lambda, u}(\theta) +u \left ( \int_{\Theta}K_{n}(\theta) d\hat{\rho}_{\lambda, u}(\theta) - B \right )+\frac{1}{\lambda}D_{\mathrm{KL}}\left ( \hat{\rho}_{\lambda, u},\pi \right ), \]
	\[u^{*}(B, \lambda) = \underset{u\geq 0}{\arg\max  } \int_{\Theta} R(\theta) d\rho^{*}_{\lambda, u}(\theta) +u \left ( \int_{\Theta}K(\theta) d\rho^{*}_{\lambda, u}(\theta) - B \right )+\frac{1}{\lambda}D_{\mathrm{KL}}\left ( \rho^{*}_{\lambda, u},\pi \right ).\]
	For $B=\infty$, define $\hat{u}(\infty,\lambda )=0$ and $u^{*} (\infty,\lambda)=0$.
\end{definition}

To conclude the section, we point out corollaries of Lemma \ref{Lemma constrained KL} and Assumption \ref{Assumption: CQ} relevant to our setting.  Define the sets 
\begin{equation}
\label{Definition: E_B}
\mathcal{E}_{B}=\left \{ \rho\in\mathcal{P}_{\pi}(\Theta): \int_{\Theta}K(\theta)d\rho(\theta) \leq B \right \}, \ B\in\mathbb{R}\cup \{\infty\}
\end{equation}
and
\begin{equation}
\label{Definition: E_B hat}
\widehat{\mathcal{E}}_{B}=\left \{ \rho\in\mathcal{P}_{\pi}(\Theta): \int_{\Theta}K_{n}(\theta)d\rho(\theta) \leq B \right \}, \ B\in\mathbb{R}\cup \{\infty\}.
\end{equation}
In the scenario where we start from a pre-selected $B$,   $\mathcal{E}_{B}$ is the (non-random) subset of $\mathcal{P}_{\pi}(\Theta)$ corresponding to Gibbs treatment rules with expected cost within the budget.  $\widehat{\mathcal{E}}_{B}$ a random set that serves as an empirical counterpart, denoting the $\rho\in\mathcal{P}_{\pi}(\Theta)$ with Gibbs rules that meet the budget constraint empirically.  

When analysis begins with a pre-determined value of $u$,   $B ( \hat{\rho}_{\lambda, u}  )$ as in Assumption \ref{Assumption: CQ} and its empirical counterpart $\widehat{B}( \hat{\rho}_{\lambda, u})$ both defined in \eqref{Definition:  B(rho) and Bhat(rho)}, are both random.  $B ( \hat{\rho}_{\lambda, u}  )$ is the expected cost of $f_{G,\hat{\rho}_{\lambda,u}}$ in the target population given the sample-dependent $\hat{\rho}_{\lambda, u}$.  This is not observed.  However, it is a key object of interest, as it tells the researcher the expected cost of the estimated policy $f_{G,\hat{\rho}_{\lambda, u}}$ associated with $u$.  Similarly, for a predetermined $u$,  both  $\mathcal{E}_{B(\hat{\rho}_{\lambda, u})}$ and $\widehat{\mathcal{E}}_{\widehat{B}(\hat{\rho}_{\lambda, u})}$ are random sets.  The former corresponds to all Gibbs treatment policies with an expected budget in the target population that is less than or equal to that of $f_{G,\hat{\rho}_{\lambda, u}}$.  The latter serves as an empirical counterpart for which membership can be evaluated from the sample.

Given Lemma \ref{Lemma constrained KL} and Assumption \ref{Assumption: CQ}, the following lemma  pertains to the empirical problem in \eqref{Optimization problem: sample, with budget} and is mostly a corollary to \ref{Lemma constrained KL}.  It says that, for a pre-specified $B$, $\hat{\rho}_{\lambda, \hat{u}(B,\lambda)}$ solves \eqref{Optimization problem: sample, with budget}.  Conversely, if we start with a predetermined value of $u$, $\hat{\rho}_{\lambda, u}$ solves an analogous problem where the budget is given by $\widehat{B}(\hat{\rho}_{\lambda, u})$.

\begin{lemma}
	\label{Lemma: solutions to empirical Gibbs problem rho_hat u_hat and rho_hat}
	(a)  Let Assumptions \ref{Assumption: measurability} and \ref{Assumption: CQ} (i) hold for $B\in\mathbb{R}\cup \{\infty\}$.  The following properties hold $P^{n}$ almost surely.  For any $\lambda >0$, $\hat{u}(B,\lambda)$ exists, is unique, and satisfies that $\hat{u}(B,\lambda)=0$ when $\int_{\Theta} K_{n}(\theta) d\hat{\rho}_{\lambda, 0}(\theta) \leq B$ and $\hat{u}(B,\lambda)$ is positive and satisfies $\int_{\Theta} K_{n}(\theta) d\hat{\rho}_{\lambda, \hat{u}}(\theta) = B$ when $\int_{\Theta} K_{n}(\theta) d\hat{\rho}_{\lambda, 0}(\theta) > B$.  Additionally,   
	\begin{align*}
	\hat{\rho}_{\lambda, \hat{u}(B,\lambda)}= \underset{\widehat{\mathcal{E}}_{B}}{\arg\min} \left [ \int_{\Theta}R_{n}(\theta) d\rho(\theta) +\frac{1}{\lambda}D_{\mathrm{KL}}(\rho,\pi) \right ],
	\end{align*}
	
	(b)  Let Assumptions \ref{Assumption: measurability} and \ref{Assumption: CQ} (ii) hold.  Then, $P^{n}$ almost surely,
	
	\[\hat{\rho}_{\lambda, u} = \underset{\widehat{\mathcal{E}}_{\widehat{B}(\hat{\rho}_{\lambda, u})}}{\arg\min} \left [ \int_{\Theta} R_{n}(\theta)d\rho(\theta) +\frac{1}{\lambda} D_{\mathrm{KL}}(\rho,\pi) \right ]. \]
\end{lemma}

\section{PAC-Bayesian Analysis \label{sec: PAC Analysis}}
Here we provide theoretical motivation for decision rules utilizing $\hat{\rho}_{\lambda, u}$ or $\hat{\rho}_{\lambda, \hat{u}(B,\lambda )}$.  In Section \ref{subsec: Regret Bounds and Oracle Inequalities}, we first construct PAC-Bayesian generalization bounds that are similar to counterparts in earlier literature.  Then we derive oracle-type inequalities that compare the proposed treatment rules to alternatives in terms of regret in the target population for a given budget.  The results in Section \ref{subsec: Regret Bounds and Oracle Inequalities} allow for a general choice of the prior or reference measure $\pi$ utilized in the definition of $\hat{\rho}_{\lambda, u}$ and $\hat{\rho}_{\lambda, \hat{u}(B,\lambda )}$.  As a result, several bounds there contain KL divergence terms related to the complexity of the learning problem and the model class $\mathcal{F}_{\Theta}$.  In Section \ref{subsec: Normal Prior}, we specify $\mathcal{F}_{\Theta}$ to consist of rules of the form in \eqref{Example treatment model class} and take $\pi$ to be an uninformative multivariate normal distribution.  In this setting, we obtain oracle-type inequalities that compare the regret of our proposed treatment assignment rules directly to that of the rules in $\mathcal{F}_{\Theta}$ with the lowest welfare regret that are in budget.  In section \ref{subsec: The Majority Vote TR}, we show that desirable properties for the majority vote rules associated with $\hat{\rho}_{\lambda, u}$ can be inherited by their majority vote counterparts.

Our analysis builds from results and techniques in the PAC-Bayesian literature that are not always stated in ways that are directly applicable to our setting.  Results from earlier literature are adapted to our setting in Appendix Section \ref{Subsec: Appendix Prelinimaries and Adaptations}, which also contains additional properties of interest.  For the most part, proofs are included there for completeness even when the adjustments are fairly minor.  This spares the reader from visiting multiple references requiring concerted adjustments at certain steps of our analysis.  Proofs specific to Section \ref{sec: PAC Analysis} are contained in Appendix Section \ref{Subsec: Proofs for PAC Analysis}.

\subsection{Regret Bounds and Oracle-Type Inequalities \label{subsec: Regret Bounds and Oracle Inequalities}}

The first step in our analysis, Theorem \ref{Theorem: PAC-Bayesian Generalization bounds}, obtains alterations of earlier PAC-Bayesian generalization bounds for the treatment assignment setting.  A variant of part (a) appears in \cite{catoni2007pac} which considers classification in the 0/1-loss setting.  In our setting, it can be derived as a special case of a bound appearing in \cite{alquier2016properties} or via a general approach to PAC-Bayesian bounds outlined, for example, in \cite{germain15aJMLR}.  We utilize the latter approach which is useful during additional steps of our analysis. The proofs of parts (b) and (c) utilize the approach of \cite{lever2010distribution}, with part (b) being an alteration of Theorem 3 in that work.

\begin{theorem}
	\label{Theorem: PAC-Bayesian Generalization bounds}
	Let $\pi\in\mathcal{P}(\Theta)$ and let Assumptions \ref{Assumption: treatment identification and boundedness}, \ref{Assumption: measurability}, and \ref{Assumption: prior indep of data} hold.  Set 
	\[\{V_{n}(\theta), V(\theta), M_{\ell} \} = \{R_{n}(\theta), R(\theta), M_{y}\} \  \mathrm{or \ else} \  \{V_{n}(\theta), V(\theta), M_{\ell} \} = \{K_{n}(\theta), K(\theta), M_{c}\}. \] 
	We have the following properties.
	
	(a)  Let $\epsilon\in(0,1]$, $\lambda >0$ and $s\in\{-1,1\}$.  With probability at least $1-\epsilon$, for all $\rho\in\mathcal{P}_{\pi}(\Theta)$ simultaneously it holds that
	\[ \int_{\Theta}s \left [ V_{n}(\theta) -  V(\theta) \right ] d\rho(\theta) \leq  \frac{1}{\lambda}  D_{\mathrm{KL}}(\rho,\pi) + \frac{1}{\lambda} \left [ \frac{\lambda^{2}M_{\ell}^{2}}{8n\kappa^{2}} +\log\frac{1}{\epsilon} \right ].   \] 
	
	(b) Let $\lambda >0$, $u\geq 0$, and $\epsilon\in(0,1]$.  With probability at least $1-\epsilon$, it holds that
	\begin{align*}
	&\left ( \int_{\Theta}V(\theta)d\hat{\rho}_{\lambda, u}(\theta) - \int_{\Theta}V_{n}(\theta)d\hat{\rho}_{\lambda, u}(\theta)  \right )^{2}
	\\
	&\leq \frac{M_{\ell}^{2}}{2n\kappa^{2}} \left [ \frac{\lambda \sqrt{2} \left ( M_{y} +uM_{c} \right ) }{\kappa \sqrt{n}}\sqrt{ \log \left ( 2\sqrt{n} \right ) + \log\frac{2}{\epsilon }  } + \frac{\lambda^{2}\left ( M_{y}+uM_{c} \right )^{2}}{2 n \kappa^{2} } + \log\left ( 2\sqrt{n} \right ) +\log\frac{2}{\epsilon} \right ].
	\end{align*}
	
	(c) Let $\lambda >0$,  $u\geq 0$, and $\epsilon\in(0,1]$.  With probability at least $1-\epsilon$, it holds that
	\begin{align*}
	&\int_{\Theta}V(\theta)d\hat{\rho}_{\lambda, u}(\theta) - \int_{\Theta}V_{n}(\theta)d\hat{\rho}_{\lambda, u}(\theta)  
	\\
	&\leq \frac{\sqrt{2} \left ( M_{y} +uM_{c} \right ) }{\kappa \sqrt{n}}\sqrt{ \log \left ( 2\sqrt{n} \right ) + \log\frac{2}{\epsilon }  } + \frac{\lambda\left ( M_{y}+uM_{c} \right )^{2}}{2 n \kappa^{2} } + \frac{1}{\lambda} \left [ \frac{\lambda^{2}M_{\ell}^{2}}{8n\kappa^{2}} +\log\frac{2}{\epsilon} \right ].
	\end{align*}
\end{theorem}

Theorem \ref{Theorem: PAC-Bayesian Generalization bounds} contains high probability bounds for notions of the generalization error between the target population regret (or alternatively, expected cost) and its empirical counterpart for Gibbs treatment rules. For example, one notion of generalization error for the cost of policy $f_{G, \hat{\rho}_{\lambda, u}}$  could be the absolute difference, 
\[ \left \lvert  K \left ( f_{G, \hat{\rho}_{\lambda, u}} \right )-K_{n} \left ( f_{G, \hat{\rho}_{\lambda, u}} \right )  \right \rvert .\]  
Suppose we take $\lambda = a\kappa \sqrt{n}/(M_{y}+uM_{c})$ for some constant $a>0$.  Then Part (b) says that with probability at least $1-\epsilon$, this absolute difference is less than or equal to 
\[\frac{M_{c}}{\kappa \sqrt{2n}} \left [ a\sqrt{\log\left ( 4n \right ) + 2\log\frac{2}{\epsilon}}+ \frac{a^{2}}{2} + \log(2\sqrt{n})+\log\frac{2}{\epsilon} \right ]^{1/2} = \mathcal{O}\left ( \frac{\log{n}}{\sqrt{n}} \right ).\]
When $M_{c}$ and $M_{y}$ are known, this upper bound can be evaluated for a given choice of $a$.  We say that $K_{n}(f_{G,\hat{\rho}_{\lambda, u}})$ is Probably (with probability at least $1-\epsilon$) and Approximately (the $\mathcal{O}( \sqrt{ \log(n)/n}  )$ upper bound on the absolute difference) Correct for $K(f_{G,\hat{\rho}_{\lambda, u}})$.  This suggests that for a predetermined choice of $u$, $K_{n}(f_{G,\hat{\rho}_{\lambda, u}})$ will give a reasonable estimate of the expected cost in the target population, $K(f_{G,\hat{\rho}_{\lambda, u}})$, provided that $\lambda$ is not too large.  Part (c) is a variation of the style of bound in (b) that is useful in deriving subsequent results.  We note that the above choice for $\lambda$ may not be best in practice, or even feasible if the upper bound $M_{c}$ is not known.  In practice $\lambda$ is chosen via cross-validation, which can be accommodated by Theorem \ref{Theorem: PAC-Bayesian Generalization bounds} similarly to the choice of $u$ as discussed below.

The bounds in Theorem \ref{Theorem: PAC-Bayesian Generalization bounds} can be adjusted to accommodate the setting where $\lambda$, $u$, or pairs $(\lambda, u)$ are selected from a finite set of values $\mathcal{W}$.  With $\lvert \mathcal{W} \rvert$ denoting the number of elements in $\mathcal{W}$, one can apply a union bound argument similar to that in the proof of part (b).  The theorem is applied once for each element of $\mathcal{W}$ with size $\epsilon/\lvert \mathcal{W} \rvert$ for each repetition.  Then, applying the union bound argument, the bounds as stated in Theorem \ref{Theorem: PAC-Bayesian Generalization bounds} remain valid for any element of $\mathcal{W}$ with the alteration that the term $\log\frac{1}{\epsilon}$ in part (a) is replaced by $(\log\frac{1}{\epsilon} + \log  \lvert \mathcal{W} \rvert )$ and the terms $\log \frac{2}{\epsilon}$ in parts (b) and (c) are replaced by $(\log\frac{2}{\epsilon} + \log  \lvert \mathcal{W} \rvert )$.  For example, when $\lambda = \mathcal{O}(\sqrt{n})$, this adds a term that is $\mathcal{O}( \log \lvert \mathcal{W} \rvert / \sqrt{n})$ to the right hand side of the high probability bound in part (a).  This observation is applicable to the remaining theorems in the paper, with minor adjustments.  Therefore, it is not unreasonable to start with multiple values for $u$. Then one may choose $u$ in $\hat{\rho}_{\lambda, u}$ for the final policy based on the empirical estimates of the associated budgets, $K_{n}(f_{G,\hat{\rho}_{\lambda, u}})$ for $u\in \mathcal{W}$, or via cross-validation.  

Before comparing our suggested treatment policies to alternative choices, we discuss a final insight from Theorem \ref{Theorem: PAC-Bayesian Generalization bounds}.  Part (a) yields that, with probability at least $1-\epsilon$, 
\begin{equation}
\label{Example PAC-Bound for discussion}
R(f_{G,\rho}) \leq \int_{\Theta} R_{n}(\theta)d\rho(\theta) + \frac{1}{\lambda}  D_{\mathrm{KL}}(\rho,\pi) +  \frac{1}{\lambda} \left [ \frac{\lambda^{2}M_{y}^{2}}{8n\kappa^{2}} +\log\frac{1}{\epsilon} \right ],
\end{equation}
for all $\rho\in\mathcal{P}_{\pi}$ simultaneously.  Given a budget $B$ such that Assumption \ref{Assumption: CQ} (i) holds,  Lemma \ref{Lemma: solutions to empirical Gibbs problem rho_hat u_hat and rho_hat} (a) states that $\hat{\rho}_{\lambda, \hat{u}(B,\lambda)}$ produces the smallest upper bound for the target population regret in \eqref{Example PAC-Bound for discussion} among all $\rho\in\mathcal{P}_{\pi}(\Theta)$ such that $K_{n}(f_{G,\rho})\leq B$.  Similarly, starting from a given value of $u$, under Assumption \ref{Assumption: CQ} (ii), Lemma \ref{Lemma: solutions to empirical Gibbs problem rho_hat u_hat and rho_hat} (b) shows that $\hat{\rho}_{\lambda , u}$ results in the smallest upper bound for the target population regret among Gibbs rules with an empirical budget less than or equal to $\widehat{B}(\hat{\rho}_{\lambda, u})$.  

Although Theorem \ref{Theorem: PAC-Bayesian Generalization bounds} (a) is most useful for our subsequent analysis, in the PAC-Bayesian literature there are alternative generalization bounds to \eqref{Example PAC-Bound for discussion} that apply for all $\rho\in\mathcal{P}_{\pi}(\Theta)$ and could be adapted to our setting.  Most notably, variants of the bounds in \cite{seeger2002pac} and \cite{catoni2007pac} are fairly ubiquitous in the literature.  Either directly or via a slight relaxation, these bounds also suggest choosing $\rho$ to minimize 
\begin{equation}
\label{Discussion PAC bound choose rho to minimize this}
\int_{\Theta}R_{n}(\theta)d\rho(\theta)+ \frac{1}{\lambda} D_{\mathrm{KL}}(\rho,\pi),
\end{equation}
for some $\lambda >0$.  Hence, if we impose an empirical budget constraint these would again lead back to  $\hat{\rho}_{\lambda,\hat{u}(B,\lambda)}$ and $\hat{\rho}_{\lambda, u}$.  We note that Seeger's bound is utilized in our analysis to derive parts (b) and (c) of Theorem \ref{Theorem: PAC-Bayesian Generalization bounds} and appears as Theorem \ref{Theorem: adaptation of Seeger's bound} in Appendix Section \ref{Subsec: Appendix Prelinimaries and Adaptations}.  While this bound does not yield a closed form solution $\tilde{\rho}$ that minimizes an upper bound on the regret, we refer to the discussion in \cite{pmlr-v76-thiemann17a} regarding a relaxation that suggests minimizing \eqref{Discussion PAC bound choose rho to minimize this} with $\lambda$ replaced by $\lambda n$, which will yield the an equivalent minimization problem when $\lambda$ is cross-validated.  The style of bound in \cite{catoni2007pac}, in particular Theorem 1.2.6 there, can be adapted to our setting via the approach in \cite{germain15aJMLR} and again suggests choosing $\rho$ to minimize \eqref{Discussion PAC bound choose rho to minimize this}.

Next we derive oracle-type inequalities that compare the target population regret associated with $\hat{\rho}_{\lambda, u}$ or $\hat{\rho}_{\lambda, \hat{u}(B,\lambda )}$ to that of alternative choices of $\rho$ among Gibbs treatment rules within a relevant budget.  It may be helpful to recall the definitions of $\mathcal{E}_{B}$ and $\mathcal{E}_{B(\hat{\rho}_{\lambda ,u})}$ from \eqref{Definition: E_B} and \eqref{Definition:  B(rho) and Bhat(rho)}, 
\[\mathcal{E}_{B} = \left \{ \rho\in\mathcal{P}_{\pi}(\Theta): K\left ( f_{G,\rho} \right ) \leq B \right \} \ \mathrm{and} \ \mathcal{E}_{B(\hat{\rho}_{\lambda, u})} = \left \{ \rho\in\mathcal{P}_{\pi}(\Theta): K\left ( f_{G,\rho} \right ) \leq  K\left ( f_{G,\hat{\rho}_{\lambda, u}} \right ) \right \}.\]
We have the following result.  
\begin{theorem}
	\label{Theorem: Main Oracle Inequality Constrained Case}
	Let $\pi\in\mathcal{P}(\Theta)$, $\lambda>0$, and $\epsilon\in(0,1]$.  Under Assumptions \ref{Assumption: treatment identification and boundedness}, \ref{Assumption: measurability}, and \ref{Assumption: prior indep of data}, we have the following properties.   
	
	(a) Let $B\in\mathbb{R}\cup\{\infty\}$, denote $\hat{u}=\hat{u}(B,\lambda)$ and let Assumption \ref{Assumption: CQ} (i) hold.   With probability at least $1-\epsilon $, it holds that
	\[R\left ( f_{G,\hat{\rho}_{\lambda, \hat{u}}} \right ) \leq \min_{\rho \in \mathcal{E}_{B}} \left \{ R\left ( f_{G,\rho} \right ) + \frac{2}{\lambda}D_{\mathrm{KL}}(\rho,\pi) \right \} + \frac{2}{\lambda}\left [ \frac{\lambda^{2}M^{2}_{y}}{8n\kappa^{2}} +\log\frac{3}{\epsilon }  \right ] + \hat{u} \sqrt{\frac{M_{c}^{2}\log\frac{3}{\epsilon }}{2n\kappa^{2}} } . \]	
	
	(b)  Fix $u\geq 0$ and let Assumption \ref{Assumption: CQ} (ii) hold.   With probability at least $1-\epsilon $, it holds that
	\begin{align*}
	R\left ( f_{G,\hat{\rho}_{\lambda, u}} \right ) & \leq  \min_{\rho\in\mathcal{E}_{B ( \hat{\rho}_{\lambda, u}  )}} \left \{ \vphantom{ \left [ \frac{\lambda^{2}}{8\kappa^{2}} \right ] +\log\frac{4}{\epsilon} }  R\left ( f_{G,\rho} \right )+\frac{1}{\lambda}D_{\mathrm{KL}}(\rho,\pi ) \right \} +  uU_{1}\left ( \epsilon ; \lambda, u, n \right )+U_{2}\left ( \epsilon ; \lambda, u, n \right )  .
	\end{align*}
	where
	\[U_{1}\left ( \epsilon ; \lambda, u, n \right ) =\frac{\sqrt{2}\left (M_{y}+uM_{c} \right )}{\kappa \sqrt{n}}\sqrt{\log\left (2\sqrt{n} \right ) +\log\frac{4}{\epsilon} }+\frac{\lambda \left ( M_{y}+u M_{c} \right )^{2}}{2n\kappa^{2}},\]
	and
	\[U_{2}\left ( \epsilon ; \lambda, u,n \right ) = \sqrt{\frac{(M_{y}+uM_{c})^{2}\log(4/\epsilon)}{2n\kappa^{2}}} + \frac{1}{\lambda} \left [ \frac{\lambda^{2} \left ( M^{2}_{y} +u M_{c}^{2} \right )}{8n\kappa^{2}} +(1+u)\log\frac{4}{\epsilon} \right ]. \]
\end{theorem}

Note that if $\lambda =\mathcal{O}(n^{1/2})$ , then for any $u\geq 0$ and $\epsilon\in (0,1]$, 
\[U_{1}\left ( \epsilon ; \lambda, u , n \right ) = \mathcal{O} \left ( \sqrt{\frac{\log(n)}{n}} \right ) \ \mathrm{and} \ U_{2}\left ( \epsilon ; \lambda, u , n \right ) =\mathcal{O} \left ( \frac{1}{\sqrt{n}} \right ). \] 
Theorem \ref{Theorem: Main Oracle Inequality Constrained Case} contains sharp oracle-type inequalities that hold with high probability.  They differ slightly from traditional oracle inequalities in that the right-hand sides contain objects that are random.  

Consider part (b) first.  In this case, the randomness on the right-hand side of the inequality stems from $\mathcal{E}_{B ( \hat{\rho}_{\lambda, u}  )}$ which depends on the sample through $B ( \hat{\rho}_{\lambda, u}  )=K (f_{G,\hat{\rho}_{\lambda, u}})$, the un-observable expected target population cost of $\hat{\rho}_{\lambda, u}$.  For a predetermined $u$, it is natural to ask if there are alternatives in $\mathcal{P}_{\pi}(\Theta)$ that would yield lower regret for the same or lower expected cost.  $\mathcal{E}_{B ( \hat{\rho}_{\lambda, u}  )}$ is therefore the natural set of interest for comparison with $\hat{\rho}_{\lambda, u}$ as it is the subset of $\mathcal{P}_{\pi}(\Theta)$ with Gibbs rules that have target population costs no greater than  $B ( \hat{\rho}_{\lambda, u}  )$.  Given a budget $B ( \hat{\rho}_{\lambda, u}  )$, an oracle with  knowledge of $R(\theta)$  could solve for $ \arg\min_{\rho\in\mathcal{E}_{B ( \hat{\rho}_{\lambda, u}  )}} R(f_{G,\rho})$. For $\lambda \rightarrow \infty$, we may consider $\arg\min_{\rho \in \mathcal{E}_{B ( \hat{\rho}_{\lambda, u}  )}} R(f_{G,\rho}) +\lambda^{-1}D_{\mathrm{KL}}(\rho,\pi) $ as a second-best oracle solution.  When $\lambda = \mathcal{O}(n^{1/2})$, for example, part (b) indicates that with high probability $\hat{\rho}_{\lambda, u}$ is close to the second best oracle solution.  In Section \ref{subsec: Normal Prior} we consider oracle-type inequalities without the KL penalty term appearing.  

In part (a), the interpretation is similar to that in part (b), except that now the set of alternative Gibbs estimators for comparison are those that satisfy the predetermined budget $B$.  This set is non-random, however now the right-hand side contains a term involving the random $\hat{u}=\hat{u}(B,\lambda)$.   Note that $\hat{u}$ is the value taken by the Lagrange multiplier $u$ in the problem
\[\min_{\rho\in\mathcal{E}_{B}} \sup_{u\geq 0} \left \{   \int_{\Theta}R_{n}(\theta) d\rho(\theta)+\frac{1}{\lambda }D_{\mathrm{KL}}(\rho,\pi) + u\left ( \int_{\Theta}K_{n}(\theta)d\rho(\theta) -B \right ) \right \}. \] 
It measures the marginal decrease in empirical penalized regret (alternatively, the increase in empirical penalized welfare) resulting from a marginal relaxation of the budget.  Recall the welfare and budget are measured per treatment.   For example, when benefits and costs are measured in dollars, how many dollars of penalized welfare are obtained (empirically) by increasing the maximum empirical cost by a dollar. In more extreme scenarios where a small increase in the budget produces a large increase in empirical welfare, the bound becomes less meaningful as the right-hand side approaches the maximum possible regret (if this level is exceeded, the bound becomes trivial).  An example of an extreme setting would be when  $\hat{u}=\mathcal{O}_{p}(n^{\alpha})$ for some $\alpha \geq 1/2$.   When $\hat{u}n^{-1/2}$ is large  relative to typical or maximal values of the regret (which ranges from zero to twice the maximal welfare), this situation is visible to the analyst.  For a fixed $\lambda$, a statement similar to part (a) can be obtained where $\hat{u}$ is replaced by a non-random constant if we make additional assumptions on the data generating distribution $P$.  For example, if we instead assume the marginal increase in population penalized regret associated with a small relaxation of the empirical budget is $\mathcal{O}_{p}(1)$.  As it stands, the bound produces a robustness check for the method’s motivation.  Intuitively, if it is easy to dramatically change the empirical welfare by relatively small budget changes, so that $\hat{u}n^{-1/2}$ is large, we may be in a situation where it is difficult to learn policies well for the given $B$ and the proposed rules should be treated cautiously.

If regions of the model space with desirable regret and budget are assigned lower probability by $\pi$, the distributions $\rho\in\mathcal{P}_{\pi}(\Theta)$ with the best trade-off between $R(f_{G,\rho})$ and $D_{\mathrm{KL}}(\rho,\pi)$ in Theorem \ref{Theorem: Main Oracle Inequality Constrained Case} will tend to have larger $D_{\mathrm{KL}}(\rho,\pi)$ terms.  As a result, the upper bounds will be larger and less informative.  Similarly, applying Theorem \ref{Theorem: PAC-Bayesian Generalization bounds} part (a) with $\rho=\hat{\rho}_{\lambda, c}$ for either $c=u\geq 0$ or $c=\hat{u}(B,\lambda)$, and noting Lemma \ref{Lemma: solutions to empirical Gibbs problem rho_hat u_hat and rho_hat}, the regret and budget bounds there are influenced by the trade-off between empirical regret (or cost) and $D_{\mathrm{KL}}(\hat{\rho}_{\lambda, c}, \pi)$.  $D_{\mathrm{KL}}(\hat{\rho}_{\lambda, c}, \pi)$ increases when $\hat{\rho}_{\lambda, c}$ involves a greater re-weighting of $\pi$ in definition \ref{Definition: optimal rho hat under a budget constraint}.  The impact of the KL terms in the bounds of this subsection are therefore related to the learning problem and model space complexity.  It is influenced by how large the model space is, how narrow the subset of the model space with low regret/budget is, the relative difference in between lower and higher regret regions and the noisiness of the data.  In parts (b) and (c) of Theorem \ref{Theorem: PAC-Bayesian Generalization bounds}, where the KL term is absent, this role falls more to the $\lambda$ parameter: if the problem is more complex, larger (relative to $n$) values of $\lambda$ are needed to achieve lower regret or cost.  If $\lambda$ is too large, remainder terms in the generalization error bounds increase.  See  \cite{lever2010distribution} for further discussion of complexity in the setting of bounds of the form in (b) and (c).

Conversely, when the policy maker has (sample independent) knowledge of the data generating process, they may be able to select or alter a given choice of $\pi$ to focus on the regions of the model space that best balance regret and cost.  Then $D_{\mathrm{KL}}(\rho, \pi)$ can be smaller for $\rho$ that put the greatest weight on the most desirable regions of the parameter space.  The result is smaller upper bounds in Theorem \ref{Theorem: Main Oracle Inequality Constrained Case} and Theorem \ref{Theorem: PAC-Bayesian Generalization bounds} (a).    A benefit of the Gibbs rules associated with $\hat{\rho}_{\lambda, u}$ and $\hat{\rho}_{\lambda, \hat{u}(B,\lambda)}$ is that economic theory or situation-specific knowledge can be factored into the treatment rule via $\pi$.  Compatibility with expert knowledge may be a valuable advantage in settings where resource limitations imply that some individuals with a positive CATE will not be treated.   As we will see in Section \ref{subsec: Normal Prior}, such knowledge is not required for the procedures to have desirable properties.

\subsection{Normal Prior \label{subsec: Normal Prior}}
As noted at the end of Section \ref{subsec: Regret Bounds and Oracle Inequalities}, perhaps unsurprisingly, knowledge about the data generating process can confer estimation benefits through the choice of $\pi$.  While it is a positive attribute that the proposed treatment rules can utilize this information when available, it is important to emphasize that such knowledge is not a requirement.  Learning procedures based on PAC-Bayesian analysis often utilize uninformative or less informative choices for $\pi$, such as normal distributions, uniform distributions when $\Theta$ is compatible, or sparsity inducing distributions.   

Here we take $\pi$ to be a multivariate normal distribution centered at the origin and utilize the models of the form in \eqref{Example treatment model class}.  We show that the proposed treatment rules maintain desirable properties.  In doing so, the KL divergence term is removed from the oracle inequalities, resulting in a clearer comparison to alternative treatment rules.  We leave an exploration of alternative prior choices and the settings where they may be desirable to future research.

We satisfy Assumptions \ref{Assumption: measurability} and \ref{Assumption: prior indep of data} with the following, more specific, condition.  Note that in the assumption below we are treating $q$ as fixed; it does not grow with the sample size.  

\begin{assumption}
	\label{Assumption: Model family and prior in normal prior setting}
	It is assumed that $\mathcal{F}_{\Theta}$ consists of treatment rules $f_{\theta}$ as described by \eqref{Example treatment model class}, with $\Theta=\mathbb{R}^{q}$.  Let 
	\[\Phi_{\mu,\sigma^{2}} \in\mathcal{P}(\mathbb{R}^{q})\]
	denote a multivariate normal distribution with mean vector $\mu$ and covariance matrix $\sigma^{2} I_{q}$ for some $\sigma>0$.  We assume that $\pi = \Phi_{0,\sigma^{2}_{\pi}} $ for some $\sigma_{\pi}>0$ that does not depend on the sample.  
\end{assumption}

Next, we define
\begin{equation*}
\Theta_{B} = \left \{ \theta\in\mathbb{R}^{q}: K(\theta) \leq B \right \} \ \mathrm{and} \ \Theta_{B ( \hat{\rho}_{\lambda, u}  )} = \left \{ \theta\in\mathbb{R}^{q}: K(\theta) \leq B ( \hat{\rho}_{\lambda, u}  ) \right \},
\end{equation*}
and denote  
\begin{equation}
\label{Definition: Optimal theta and theat_u}
\overline{\theta} \in \underset{ \Theta_{B}}{\arg\min} \left [   R(\theta) \right ] \ \mathrm{and} \ \overline{\theta}_{u} \in \underset{\Theta_{B ( \hat{\rho}_{\lambda, u}  )}}{\arg\min} \left [ R(\theta) \right ]. 
\end{equation}
Note that $\Theta_{B ( \hat{\rho}_{\lambda, u}  )}$ and $\overline{\theta}_{u}$ are random as they vary with $B ( \hat{\rho}_{\lambda, u}  )$.  $\Theta_{B ( \hat{\rho}_{\lambda, u}  )}$ is the set of parameters such that the corresponding models in $\mathcal{F}_{\Theta}$ have lower expected target population cost than $f_{G,\hat{\rho}_{\lambda, u}}$.   $\overline{\theta}_{u}$ is the minimizer of the population regret among this set.   With regard to $\overline{\theta}$ and $\overline{\theta}_{u}$, we assume the following condition.  
\begin{assumption}
	\label{Assumption: existence of regret minimizers}
	With probability one, $\overline{\theta}$ and $\overline{\theta}_{u}$ as defined in \eqref{Definition: Optimal theta and theat_u} exist and are nonzero.
\end{assumption}
This type of condition is implicitly assumed in, for example, \cite{KT2018} and in \cite{sun2021empirical}.  It simplifies the exposition rather than allowing that the models associated with these parameters have regret that is arbitrarily close to an infimum. The requirement that $\overline{\theta}$ and $\overline{\theta}_{u}$ are nonzero simply specifies that the covariates are relevant to the budget constrained welfare problem.  Lastly, our analysis will also require the following technical condition. 
\begin{assumption}
	\label{Assumption: technical condition for comparisons to EWM}
	There exists a constant $\nu >0$ such that
	\[P\left [ \left ( \phi(X)^{\intercal}\theta \right ) \left ( \phi(X)^{\intercal}\theta' \right ) <0  \right ] \leq \nu \lVert \theta - \theta' \rVert\]
	for any $\theta$ and $\theta'\in\mathbb{R}^{q}$ such that $\lVert \theta \rVert =\lVert \theta' \rVert =1 $.
\end{assumption}

Assumption \ref{Assumption: technical condition for comparisons to EWM} or a direct analog is applied in several classification and bipartite ranking applications utilizing PAC-Bayesian approaches.  For examples, see  \cite{ridgway2014pac}, \cite{alquier2016properties}, and \cite{guedj2018pac}.  It is a fairly mild requirement and, as is shown in \cite{alquier2016properties} (c.f. p. 10 there), it is satisfied whenever $\phi(X)/\lVert \phi(X) \rVert$ has a bounded density on the unit sphere.  

We have the following result.

\begin{theorem}
	\label{Theorem: Oracle inequality w/ normal prior}
	Let Assumptions  \ref{Assumption: treatment identification and boundedness},  \ref{Assumption: Model family and prior in normal prior setting},  \ref{Assumption: existence of regret minimizers}, and \ref{Assumption: technical condition for comparisons to EWM} hold.  Let $\sigma_{\pi}=1/\sqrt{q}$.  Then we have the following properties for any $\epsilon\in(0,1]$.
	
	(a)  Let Assumption \ref{Assumption: CQ} (i) hold for a given $B\in\mathbb{R}\cup\{\infty\}$.  Let  $\lambda =\kappa  \sqrt{n q}/M_{y}$, $\hat{u}=\hat{u}(B,\lambda)$ and $u^{*}=u^{*}(B, \lambda/2)$.  With probability at least $1-\epsilon$, it holds that
	\[R\left ( f_{G,\hat{\rho}_{\lambda,\hat{u}}} \right ) \leq R \left ( \overline{\theta} \right ) + \sqrt{\frac{q}{n}}\log\left ( 4n \right ) \frac{M_{y}}{\kappa} + \frac{2M_{y}\log\frac{3}{\epsilon}}{\kappa \sqrt{nq}}+\hat{u}\sqrt{\frac{M_{c}^{2}\log\frac{3}{\epsilon}}{2n\kappa^{2}}}+  \frac{u^{*} \nu M_{c}}{\sqrt{n}}  + \overline{U}_{1}(n;q), \]
	where $\overline{U}_{1}(n;q)=\mathcal{O}(n^{-1/2})$ with the explicit formulation given in the proof.
	
	(b)  Fix $u\geq 0$ and set $\lambda = \kappa \sqrt{nq}/(M_{y}+uM_{c})$.  Let Assumption Assumption \ref{Assumption: CQ} (ii) hold.   With probability at least $1-\epsilon$,
	\begin{align*}
	R\left ( f_{G,\hat{\rho}_{\lambda,u}} \right ) \leq R\left ( \overline{\theta}_{u} \right )+ \frac{M_{y}+uM_{c}}{\kappa} \left [ \overline{U}_{2}(n;q,u,\epsilon) + \overline{U}_{3}(n;q,u,\epsilon)+\overline{U}_{4}(n;q,u) \right ],
	\end{align*}
	where $\overline{U}_{2}(n;q,u,\epsilon)=\mathcal{O}(\log(n)n^{-1/2})$,  $\overline{U}_{3}(n;q,u,\epsilon)=\mathcal{O}(n^{-1/2})$, and $\overline{U}_{4}(n;q,u)=\mathcal{O}(n^{-1/2})$, with the explicit forms given in the proof.
\end{theorem}

Note that the values for $\lambda$ in parts (a) and (b) are chosen to produce the nearly optimal rate of convergence in part (b).  In practice there may be better choices and we will typically choose $\lambda$ via cross-validation.  As noted in the discussion following Theorem \ref{Theorem: PAC-Bayesian Generalization bounds},   we may choose $\lambda$, $u$, or pairs $(\lambda, u)$ from a finite set of values $\mathcal{W}$.  In this case the theorem above can be adjusted to hold simultaneously for all elements of $\mathcal{W}$ by replacing the terms $\log ( \epsilon/3)$ on the right-hand side of the inequality in (a) by $\log(\epsilon/3)+\log|\mathcal{W}|$ and the terms on the right-hand side of (b) that involve $\log ( \epsilon/4)$, which appear in the $\overline{U}_{j}$ terms defined in the proof, are replaced by $\log(\epsilon/4)+\log|\mathcal{W}|$.  For example, for fixed $\epsilon\in (0,1]$ and $u\geq 0$, this adds a term that is $\mathcal{O}(\log  |\mathcal{W}| n^{-1/2})$ to the right hand side of (b).

In Theorem \ref{Theorem: Oracle inequality w/ normal prior},  $f_{G,\hat{\rho}_{\lambda, \hat{u}(B,\lambda) }}$ and $f_{G,\hat{\rho}_{\lambda, u }}$ are compared to the best (non-stochastic) models in $\mathcal{F}_{\Theta}$ with an expected cost no greater than $B$ or $B ( \hat{\rho}_{\lambda, u}  )$, respectively.  Additionally, the absence of KL terms in the inequalities allows for a more salient comparison to relevant alternatives.  In part (b), for any $u\geq 0$ and $\epsilon\in(0,1]$, the terms beside $R(\overline{\theta}_{u})$ on the right hand-side are collectively $\mathcal{O}(\log(n)n^{-1/2})$.  With high probability, the regret of $f_{G,\hat{\rho}_{\lambda,u}}$ gets close to the regret an oracle would obtain choosing the best rule from the subset of $\mathcal{F}_{\Theta}$ with a target population budget no greater than that of $f_{G,\hat{\rho}_{\lambda,u}}$.   The rate $\log(n)n^{-1/2}$ is nearly optimal.  For example, in the unconstrained case with $B=\infty$, which corresponds to $u=0$ or $\hat{u}=u^{*}=0$, \cite{KT2018} show that $n^{-1/2}$ is the optimal rate for bounds on the expected regret of the empirical welfare maximizer over $\mathcal{F}_{\Theta}$, provided $\mathcal{F}_{\Theta}$ has a finite VC-dimension (see the discussion there for more details).

Part (a) has the complication of involving $\hat{u}=\hat{u}(B,\lambda)$ and $u^{*} = u^{*}(B,\lambda/2)$ as $\lambda$ grows with $n$.  The effect of $\hat{u}$ is related to the marginal decrease in  $R_{n}(f_{G,\hat{\rho}_{\lambda, \hat{u}(B,\lambda)}}) +\lambda^{-1}D_{\mathrm{KL}}(\hat{\rho}_{\lambda, \hat{u}(B,\lambda)}, \pi )$ associated with marginal increases in $B$ as the penalty diminishes ($\lambda$ increases).  The behavior of $u^{*}$ is related to the marginal decrease in the penalized regret of $f_{G,\rho^{*}_{\lambda,u^{*}(B,\lambda/2)}}$ associated with marginal increases in $B$.  Suppose we are unlikely to have large marginal gains in empirical or theoretical penalized regret associated with a marginal increase in $B$ at all or small penalty levels (i.e. as $\lambda\rightarrow \infty$).  Then (a) implies that, with high probability and for large enough sample sizes, the regret of $f_{G,\hat{\rho}_{\lambda, \hat{u} }}$ is close to the regret that would be obtained by an oracle choosing the best policy from the subset of $\mathcal{F}_{\Theta}$ with an expected cost in the target population that is less than or equal to $B$.  For example, if $u^{*}=\mathcal{O}(1)$ and $\hat{u}=\mathcal{O}_{p}(1)$ as $n$ and $\lambda $ increase, then the terms on the right-hand side of the inequality in (a) other than $R(\overline{\theta})$ are $\mathcal{O}_{p}(\log(n)n^{-1/2})$.

We conclude this subsection with remarks regarding implications for the proposed treatment assignment rules.  One drawback of starting from a fixed $B$ and utilizing $\hat{u}=\hat{u}(B,\lambda)$ is the absence of a counterpart to Theorem \ref{Theorem: PAC-Bayesian Generalization bounds} (b) for the cost $K(f_{G,\hat{\rho}_{\lambda, \hat{u}}})$ when $\hat{u}$ is random.   Even when $\hat{u}$ and $u^{*}$ are well behaved so that Theorem \ref{Theorem: Oracle inequality w/ normal prior} (a) implies it is likely that $f_{G,\hat{\rho}_{\lambda, \hat{u}}}$ will have regret comparable to the best rules in $\mathcal{F}_{\Theta}$ with expected cost less than $B$, this may be achieved with an expected cost greater than $B$.  On the other hand Theorem \ref{Theorem: PAC-Bayesian Generalization bounds} (a) with $\rho = \hat{\rho}_{\lambda, \hat{u}}$ yields that with probability at least $1-\epsilon$,
\[K\left ( f_{G,\hat{\rho}_{\lambda, \hat{u}}} \right ) \leq B + \frac{1}{\lambda}D_{\mathrm{KL}}\left ( \hat{\rho}_{\lambda, \hat{u}}, \pi \right ) + \frac{1}{\lambda}\left [ \frac{\lambda^{2}M^{2}_{c}}{8n\kappa^{2}}+\log\frac{1}{\epsilon} \right ], \]
where we have used the fact that $K_{n}(f_{G,\hat{\rho}_{\lambda, \hat{u}}})\leq B$ a.s. under the assumptions of the theorem.  When $\lambda =\mathcal{O}(n^{1/2})$, for example, whether or not we have an upper bound that approaches $B$ depends on the behavior of this KL term.  Unfortunately, $\hat{u}$ and the KL term above are difficult to analyze in this scenario as $\hat{u}$ is essentially defined implicitly to ensure $K_{n}(f_{G,\hat{\rho}_{\lambda, \hat{u}}})\leq B$.   It is possible to cross-validate $B$, for example examining values less than $B$ to try and ensure the expected budget is not violated.  The comments regarding extending the high probability bounds to apply simultaneously for multiple values of $u$ can be applied to choices for $B$ as well. 

On the whole, the procedure starting with a set of values for $u$ may be more compelling.  By Theorem \ref{Theorem: PAC-Bayesian Generalization bounds} (b) and the surrounding discussion, for values $u$ in a reasonably sized set $\mathcal{W}$, the values of $K_{n}(f_{G,\hat{\rho}_{\lambda, u}})$ provide reasonable estimates of $K(f_{G,\hat{\rho}_{\lambda, u}})$, the expected costs of these policies conditional on the rules estimated from the sample.  These can be utilized to select $u$.  Alternatively, $u$ can be chosen from $\mathcal{W}$ via cross-validation or by some other method.  For example, in the case of pure quantity constraints, it may be possible use data from the target population to select $u$ to achieve the correct (or nearly correct) proportion of treatments assigned in the target population.  Theorem \ref{Theorem: Oracle inequality w/ normal prior} (b) and its extension to hold for all $u\in\mathcal{W}$ simultaneously, then indicate it is likely $ R(f_{G,\hat{\rho}_{\lambda, u}})$ for the selected $u$ will be comparable to the best treatment rules in $\mathcal{F}_{\Theta}$ among those whose target population cost does not exceed that of $f_{G,\hat{\rho}_{\lambda, u}}$.  Hence by starting from a set of $u$ values, the policy maker can trace out reasonable estimates of the target population budget horizon.  At the same time, the policy selected according to these budget estimates is likely to be the best bang for the buck in that the associated regret gets close to that which an oracle would choose for the same target population cost. 

\subsection{The Majority Vote Treatment Rule \label{subsec: The Majority Vote TR}} 
Let $\rho\in\mathcal{P}_{\pi}(\Theta)$.  As mentioned in Section \ref{sec: PAC-Bayesian Setting},  the non-stochastic majority vote treatment rule $f_{\mathrm{mv},\rho}$ in \eqref{Definition: Majority vote treatment rule}  is a close relative of the Gibbs rule $f_{G,\rho}$ that can prove numerically more stable in practice.  In the classification literature, it is well known that the risk associated with the majority vote rule, where risk is defined for a zero-one loss function, is upper bounded by twice the risk associated with the Gibbs classification method (e.g., \cite{langford2003pac}, \cite{mcallester2003simplified}).  Hence analysis of the Gibbs treatment rule is often used to justify use of the majority vote.  Additionally, the ``$2 \times$"  upper bound can be loose and it is not uncommon for majority vote rules to outperform Gibbs rules.  We refer to \cite{germain15aJMLR} for further discussion regarding the majority vote versus the Gibbs method for classification settings.  Here, we show that, as in the classification setting, the majority vote treatment rule can inherit desirable qualities from the Gibbs treatment rule in the budget constrained treatment rule setting.

While the majority vote rule $f_{\mathrm{mv},\rho}$ is not guaranteed to satisfy the same budget as its Gibbs counterpart $f_{G,\rho}(x)$, we can still show that when $f_{G,\rho}(x)$ is close to $f^{*}_{B(\rho)}(x)$, the optimal rule for its budget,  
\[B(\rho) = K(f_{G,\rho}),\]
then $f_{\mathrm{mv},\rho}$ will also be close to $f^{*}_{B(\rho)}$.  The measurement of closeness, defined shortly, depends on both the welfare achieved and deviations from the budget $B(\rho)$.  We will suppose that 
\begin{equation}
\label{min solution condition cost of gibbs}
B(\rho)> E_{Q}[\delta_{c}(X) 1\{\delta_{c}(X) < 0\}].
\end{equation}
That is, $f_{G,\rho}$ does not achieve the exact cost of the cost-minimizing rule $1\{\delta_{c}(x)<0\}$ for $x\in\mathcal{X}$.  If \eqref{min solution condition cost of gibbs} were an equality, the budget of $f_{G,\rho}$ would be such that a policy maker faced with this budget would need to ignore welfare and seek the lowest cost rule.  Hence, when we are interested in maximizing welfare with a budget constraint, it is reasonable to rule out the case where the solution to the policy maker's problem is to ignore welfare and seek the lowest cost.  In addition to \eqref{min solution condition cost of gibbs}, we will assume that $\delta_{y}(X)$ and $\delta_{c}(X)$ have bounded densities so that optimal solution to the decision makers in Theorem \ref{Theorem: theoretically optimal policy choice} is deterministic.

Under \eqref{min solution condition cost of gibbs}, Assumption \ref{Assumption: treatment identification and boundedness}, and the condition that $\delta_{y}(X)$ and $\delta_{c}(X)$ have bounded densities, Theorem \ref{Theorem: theoretically optimal policy choice} yields that the optimal budget-constrained policy for the budget $B(\rho)$ of the Gibbs rule $f_{G,\rho}$ is of the form
\begin{equation}
\label{equation: optimal rule B(rho)}
f^{*}_{B(\rho)}(x)=1\{\delta_{y}(x)-\eta_{B(\rho)  }\delta_{c}(x) >0\}, \ \ x\in\mathcal{X},
\end{equation}
for a constant $\eta_{B(\rho)}$.  It also follows from Theorem \ref{Theorem: theoretically optimal policy choice} that either $\eta_{B(\rho)}=0$ and $K(f^{*}_{B(\rho)})<B(\rho)$ or else $\eta_{B(\rho)}>0$ and $K(f^{*}_{B(\rho)})=B(\rho).$    Recalling the definition of the welfare-regret under a budget constraint in \eqref{Definition: R_B(f)},
\[R_{B(\rho)}(f) \equiv W\left ( f^{*}_{B(\rho)} \right )- W\left ( f \right ),\]
it is clear that $R_{B(\rho)}(f_{G,\rho})$ is non-negative.  It is small only when $f_{G,\rho}$ attains a welfare that is close to the budget optimal rule in its own budget class.  We will show that when $R_{B(\rho)}(f_{G,\rho})$ is small, $f_{\mathrm{mv},\rho}$ has similar welfare to the optimal policy $f^{*}_{B(\rho)}$ and is unlikely to violate the budget $B(\rho)$ by a large amount.  

First note that if a decision maker faced a budget of $B(\rho)$, it would be reasonable to seek a rule $f:\mathcal{X}\rightarrow[0,1]$ that minimizes 
\[L_{B(\rho)}(f) \equiv  E_{Q}\left [ \left ( \delta_{y}(X)-\eta_{B(\rho)} \delta_{c}(X) \right ) \left (f^{*}_{B(\rho)}(X)-f(X) \right ) \right ],\]
with the associated loss function
\begin{align*}
\ell_{B(\rho)}(f,x)  &=  \left ( \delta_{y}(x)-\eta_{B(\rho)} \delta_{c}(x) \right ) \left (f^{*}_{B(\rho)}(x)-f(x) \right )
\\
&=\begin{cases}
0 &
\mathrm{if}\ f^{*}_{B(\rho)}(x)=f(x), \\
\left | \delta_{y}(x)-\eta_{B(\rho)}\delta_{c}(x) \right | & \mathrm{if} \ f^{*}_{B(\rho)}(x)\neq f(x).
\end{cases}
\end{align*}

By the form of $f^{*}_{B(\rho)}$ in \eqref{equation: optimal rule B(rho)},  $L_{B(\rho)}(f)$ is non-negative and attains the value zero only when $f(X)=f^{*}_{B(\rho)}(X)$ almost surely.  Of course, such a loss function cannot yield an estimation strategy directly because $\delta_{y}$, $\delta_{x}$, and $\eta_{B(\rho)}$ are unknown.  However,  when  $L_{B(\rho)}(f)$ is small, this means we are unlikely to encounter a set of co-variates $X$ for which $f$ assigns treatment and $\eta_{B(\rho)}\delta_{c}(X)$ exceeds $\delta_{y}(X)$ by a large amount.   We have the following result

\begin{theorem}
	\label{Theorem: mv loss bounded by twice the budget penalized welfare regret}
	Let $\rho\in\mathcal{P}_{\pi}(\Theta)$.  Let Assumptions \ref{Assumption: treatment identification and boundedness} and \ref{Assumption: measurability} hold and also assume that \eqref{min solution condition cost of gibbs} holds and $\delta_{c}(X)$ and $\delta_{y}(X)$ have bounded densities so that $E_{Q}[1\{\delta_{y}(X)=\eta_{B(\rho)} \delta_{c}(X)\}]=0$.  Then  
	\[L_{B(\rho)}\left ( f_{\mathrm{mv},\rho} \right ) \leq 2 R_{B(\rho)}\left ( f_{G,\rho} \right ).\]
\end{theorem}

We note that the expectation in the definition of $L_{B(\rho)}(f)$ is taken with respect to a draw from the target population.  When $\rho$ is dependent on the sample data, the result and proof still hold, conditional on the estimated rule or sample, provided that \eqref{min solution condition cost of gibbs} can be assumed to hold almost surely for $\rho$ or with high probability if considering probabilistic bounds such as those in Sections \ref{subsec: Regret Bounds and Oracle Inequalities} and \ref{subsec: Normal Prior}.  This is reasonable to assume for $\hat{\rho}_{\lambda, u}$, particularly when $u$ is not so large that no treatments will be assigned.  The notion that, for appropriately chosen values of $\lambda$, $R_{B(\hat{\rho}_{\lambda, u})}(f_{G,\hat{\rho}_{\lambda, u}})$ is small is exactly the implication of Theorems \ref{Theorem: Main Oracle Inequality Constrained Case} (b) and \ref{Theorem: Oracle inequality w/ normal prior} (b).  

For example, assume that the conditions of Theorem \ref{Theorem: Oracle inequality w/ normal prior} hold, take $\lambda = \kappa \sqrt{nq}/(M_{y}+uM_{c})$ (although we continue to write $\lambda$ to reduce clutter in the notation), and suppose that \eqref{min solution condition cost of gibbs} holds almost surely for $\rho = \hat{\rho}_{\lambda, u}$ and that $\delta_{c}(X)$ and $\delta_{y}(X)$ have bounded densities.  Then by Theorem \ref{Theorem: Oracle inequality w/ normal prior} (b), with  probability at least $1-\epsilon$ it holds that
\begin{align*}
& R_{B ( f_{G,\hat{\rho}_{\lambda, u}} )}\left ( f_{G,\hat{\rho}_{\lambda, u}} \right )
\\
& \leq \underset{\theta\in \Theta_{B(\hat{\rho}_{\lambda, u})}}{\arg\min}  \left [ R_{ B ( f_{G,\hat{\rho}_{\lambda, u}} )}(\theta)  \right ] + \frac{M_{y}+uM_{c}}{\kappa} \left [ \overline{U}_{2}(n;q,u,\epsilon) + \overline{U}_{3}(n;q,u,\epsilon)+\overline{U}_{4}(n;q,u) \right ]
\end{align*}
where we have done some simple algebra on the inequality of part (b) of Theorem \ref{Theorem: Oracle inequality w/ normal prior} utilizing the definitions of regret and regret under a budget constraint. The above also uses the notation
\[ R_{ B ( f_{G,\hat{\rho}_{\lambda, u}} )}(\theta) = W\left ( f^{*}_{B ( f_{G,\hat{\rho}_{\lambda, u}} )} \right ) - W \left ( f_{\theta} \right ). \]
Recall that the terms outside of the $\arg\min$ on the right-hand side of the above inequality are at most $\mathcal{O}(\log(n)n^{-1/2})$ for fixed $u\geq 0 $, $q\in\mathbb{N}$ and $\epsilon\in (0,1]$.  If, for example, 
\[\delta_{y}(X) = \phi(X)^{\intercal}\theta_{y}, \ \ \mathrm{and} \ \ \delta_{c}(X) = \phi(X)^{\intercal}\theta_{c}, \]
for some $\theta_{y},\theta_{c}\in\mathbb{R}^{q}$, then we would have 
\[\underset{\theta\in \Theta_{B(\hat{\rho}_{\lambda, u})}}{\arg\min}  \left [ R_{ B ( f_{G,\hat{\rho}_{\lambda, u}} )}(\theta)  \right ]=0.\]
In this case, the above combined with Theorem \ref{Theorem: mv loss bounded by twice the budget penalized welfare regret} produce that, with probability at least $1-\epsilon$, $L_{B(\hat{\rho}_{\lambda, u})}(f_{\mathrm{mv},\hat{\rho}_{\lambda, u}})$ is bounded above by terms that are $\mathcal{O}(\log(n)n^{-1/2})$.

\section{Simulation Study and Implementation Details \label{sec: Simulation and Implementation}}
In this section we evaluate the proposed treatment assignment methodology in a simulation environment.  We also discuss model estimation and implementation.  Section \ref{subsec: Simulation Setup} describes the simulation environment and findings. Section \ref{subsec: SMC implementation} describes a model estimation strategy using the Sequential Monte Carlo (SMC) approach and discusses the implementation choices utilized in the simulation. 

\subsection{Simulation Study \label{subsec: Simulation Setup}}
We assign treatments utilizing $\hat{\rho}_{\lambda, u}$ in the following simulation environments.  We take $X= ( X_{1}, X_{2}, X_{3} )$ where $X_{j}\sim \mathrm{Unif}(-1,1)$ for $j=1,2,3$ are i.i.d. uniform random variables.  Letting $\Lambda(v)=(1+\exp(-v))^{-1}$ denote the logistic function, potential outcomes are determined via
\[Y_{d}=\max\{X_{1}+X_{2}, 0\}+\max\{ X_{3},0 \} + 4 d \Lambda\left ( 2 \left ( X_{1}+X_{1}X_{2}+X_{2} \right ) \right )+\epsilon, \ \ d\in\{0,1\}, \]
where $\epsilon$ is taken to be a standard normal random variable that is truncated to take values in $[-2,2]$ and is independent of all other variables considered.  Potential costs are determined via 
\[C_{0}=0, \ \ C_{1}\sim \mathrm{Binom}\left ( 6, \frac{ 4\Lambda\left (a(3 X_{2} + 1.5 X_{3} ) \right )}{6} \right ),\]
where $a$ is a constant.  Lastly, $e(x)=1/2$ for all $x\in\mathcal{X}$ so that $D\sim \mathrm{Bern}(1/2)$ and is independent of the other variables.  We consider $a\in\{1,2,4\}$.  

Each choice of $a$ corresponds to a different data generating process (DGP) and for each we perform the following simulation study separately.  We simulate training sets each with sample size $n=1,000$.  A testing sample of size $n_{\mathrm{test}}=10,000$, which is re-used across training sample iterations, yields approximately the true costs and benefits from of any considered treatment rule. We consider 100 training simulation replicates.  Using knowledge of the DGP, we can  calculate $E_{Q}[Y_{1,i}-Y_{0,i}|X_{i}]$ and $E_{Q}[C_{1,i}|X_{i}]$ for each testing set observation.  Then, for a rule $f(x):\mathcal{X}\rightarrow [0,1]$, we use the testing set to obtain the (approximate) gain and cost associated with $f$,
\[\mathrm{Gain \ of \ } f = E_{Q}\left [ \left ( Y_{1}-  Y_{0} \right ) f(X) \right ],\]
and 
\[\mathrm{Cost \ of \ } f = E_{Q}\left [ C_{1} f(X) \right ].\]
The Gain of $f$ is the expected increase in welfare, relative to no treatments, associated with the treatment rule policy while the Cost of $f$ is its cost.

Section \ref{subsec: SMC implementation} describes the Sequential Monte Carlo procedure used to sample from $\hat{\rho}_{\lambda, u}$ to implement the associated Gibbs or majority vote rule.  We consider values of $u$ increasing from $0$ to $2$ in increments of $0.05$.  For each choice of $u$, $\lambda$ is chosen by $4$-fold cross-validation to maximize $W_{n}(f)-uK_{n}(f)$ across hold-out folds, where $f=f_{G,\hat{\rho}_{\lambda, u}}$ for the Gibbs rules and $f=f_{\mathrm{mv},\hat{\rho}_{\lambda, u}}$ for the majority vote rules.  We thus obtain treatment rules with varying gain-cost pairs for different choices of $u$ and can obtain cross-validation-based estimates of these pairs during the estimation stage.  

To make estimation from a training sample operational, we must specify a treatment rule space $\mathcal{F}_{\Theta}$ and prior $\pi$.  With $d_{x}$ denoting the dimension of $\mathcal{X}$  ($d_{x}=3$ in the simulation setting), for $k\in\mathbb{N}$ and $q_{k}={d_{x} + k \choose k}$, the polynomial transformation on $\mathcal{X}$ of order at most $k$ is defined as
\begin{equation}
\label{Polynomial transformations Def}
\mathcal{F}_{\Theta}^{\mathrm{poly}}(k) = \left \{ m(x): m(x) = \sum_{j=1}^{q_{k}}\theta_{j}\phi_{j}(x), \theta\in\mathbb{R}^{q_{k}} \right \} ,
\end{equation}
where the summation is over all monomials $\phi_{j}(x)=\prod^{d}_{\ell=1}x_{\ell}^{p_{j\ell}}$ with $\sum_{\ell=1}^{d}p_{j\ell}\leq q$, $p_{j\ell}\in \mathbb{N}\cup\{0\}$.  We take $\mathcal{F}_{\Theta}$ to be the family of rules described in \eqref{Definition: class of models} and \eqref{Example treatment model class} where the transformations $\phi_{j}(x)$ are the monomials used in the construction of the polynomial transformations on $\mathbb{R}^{3}$ of order at most $2$ with the monomials normalized by their sample mean and standard deviation calculated from training data.  We  set $\pi $ to be the standard multivariate normal distribution over $\mathbb{R}^{10}$. 

As an alternative treatment rule, we consider the approach of \cite{sun2021treatment}.  Under our simulation setting, where for example  $C_{0}\leq C_{1}$ almost surely, $f^{*}_{B}$ in Theorem \ref{Theorem: theoretically optimal policy choice} takes the form
\[f^{*}_{B}(x) = 1\left \{ \frac{\delta_{y}(x)}{\delta_{c}(x)} > \eta_{B}  \right \},\]
for some constant $\eta_{B}$.  \cite{sun2021treatment} show that $\delta_{y}(x)/\delta_{c}(x)$ can be estimated nonparametrically by re-purposing the so-called generalized random forest methodology of \cite{athey2019generalized}.  When costs can be observed at the time of treatment assignment, their approach first estimates $\delta_{y}(x)/\delta_{c}(x)$ for $x\in\mathcal{X}$.  This produces an estimate of the conditional welfare to conditional cost ratio $\delta_{y}(X_{i})/\delta_{c}(X_{i})$ for each observation in the target group\footnote{By target group we mean individuals or units for whom treatment assignment must be determined, typically this is the wider population from which the sample comes from that consists of individuals or units not used in fitting treatment rules.}.  These ratio estimates are ranked according in descending order and treatments are allotted according to this order until the budget is exhausted.   We call such a method of  assignment, where a ranking is derived for members of the target group who are then treated in that order until the budget is reached, a ``batch implementation" method.  Additionally, as a baseline rule, we estimate the CATE  $\delta_{y}(x)=E_{Q}[Y_{1}-Y_{0}|X=x]$ using the generalized random forest of \cite{athey2019generalized} and then use the resulting scores in the target group for a batch implementation.  This baseline approach does not factor costs into the treatment decisions.  In our simulations, these methods are implemented using R 4.2.2 (\cite{R_key}) with the \textit{grf} package (\cite{grf_r_key}) following the described adaptation in \cite{sun2021treatment} for their approach.  The default package settings were except that the known treatment probabilities supplied to the algorithm.  

The approach of \cite{sun2021treatment} and the baseline that ignores cost utilize batch implementations while the Gibbs and majority vote methods do not.  To compare like-for-like, the majority vote models associated with $\hat{\rho}_{\lambda, u}$ for a range of $u$ values (with $\lambda$ chosen via cross-validation for each $u$) are amenable to a batch implementation method.  An algorithm for implementing a batch treatment rule utilizing the majority vote rules is described below.

\noindent\hrulefill

\noindent\textbf{Batch treatment implementation utilizing majority vote rules}

\vspace{-9pt} \noindent\hrulefill

\vspace{0.1in} \noindent\textbf{Input} Target group observations indexed by $\mathcal{I}_{\mathrm{target}} = \{ 1, 2,\dots, n_{\mathrm{target}} \} $ with $n_{\mathrm{target}}$  total observations and covariates $\{X_{j}: j\in\mathcal{I}_{\mathrm{target}}\}$, minimum cost $B_{\mathrm{min}}$, number of bins used denoted $n_{\mathrm{bin}}$, budget $B$, set of $u$ values denoted $\mathcal{W}_{u}$, majority vote rules $f_{\mathrm{mv}, \hat{\rho}_{\tilde{\lambda}_{u},u} }$ for each $u\in\mathcal{W}_{u}$ along with cost estimates $\hat{\mathrm{cost}}(f_{\mathrm{mv},\hat{\rho}_{\tilde{\lambda}_{u}, u}})$.  If treatment is assigned to $X_{j}$, we then observe the cost of treating individual $j$, $C_{1,j}.$

\vspace{0.1in} \noindent\textbf{Output} $\mathcal{I}_{\mathrm{treat}} \subseteq \mathcal{I}_{\mathrm{target}}$, a subset of individuals in the target group assigned treatment.

\bigskip 

\vspace{0.1in} \noindent Step 1: Initialization

\bigskip 

\noindent \hspace{0.05in} Set $B_{0} \leftarrow B_{\mathrm{min}}$ and  $\mathcal{I}_{\mathrm{treat}} \leftarrow \emptyset$.

\bigskip 

\noindent Step 2: Treatment determinations

\bigskip 

\noindent \hspace{0.05in} \textbf{For} $i=1:n_{\mathrm{bin}}$
\begin{itemize}
\item  Set $u_{i} \leftarrow  \underset{u\in\mathcal{W}_{u}}{\arg\min } \  \left | \hat{\mathrm{cost}} \left (f_{\mathrm{mv},\hat{\rho}_{\tilde{\lambda}_{u}, u}} \right ) - \left ( \frac{i (B-B_{\mathrm{min}})}{n_{\mathrm{bin}}} \right ) \right |.$

\item Let $\mathcal{I}_{i} = \{\alpha_{i}(1),\alpha_{i}(2),\dots \}$ denote the ordered ranking of target group observations not currently in the set $\mathcal{I}_{\mathrm{treat}}$ in decreasing order of the majority vote scores.  That is, in decreasing order of $\int_{\Theta}f_{\theta}(X_{j})d\hat{\rho}_{\tilde{\lambda}_{u_{i}},u_{i}}(\theta)$ for $j\in\mathcal{I}_{\mathrm{target}} \cap \mathcal{I}_{\mathrm{treat}}^{\mathrm{\mathbf{c}}} $.  For example, $\alpha_{i}(1)$ gives the index of the individual with the largest such majority vote score that is in $\mathcal{I}_{\mathrm{target}}$ but not currently in $\mathcal{I}_{\mathrm{treat}}$, provided that $\mathcal{I}_{\mathrm{target}} \cap \mathcal{I}_{\mathrm{treat}}^{\mathrm{\mathbf{c}}} \neq \emptyset $.  In the latter case, $\mathcal{I}_{i}=\emptyset$.

\item Set $k \leftarrow 1$.

\item  \textbf{While} $B_{0} < B \times  n_{\mathrm{target}}$ \textbf{and} $k \leq  \left | \mathcal{I}_{i} \right |$ \textbf{do} $\mathcal{I}_{\mathrm{treat}} \leftarrow \mathcal{I}_{\mathrm{treat}} \cup \alpha_{i}(k)$,  $B_{0} \leftarrow B_{0}+C_{1, \alpha_{i}(k)}$, and then $k \leftarrow k+1$.

\end{itemize}
\noindent \hspace{0.05in} \textbf{End For}

\vspace{-8pt} \noindent\hrulefill

\vspace{0.15in}

\newpage 

The algorithm above divides the cost space below the budget into bins and then performs a batch implementation at each bin using the majority vote scores of the model with an estimated cost nearest to that bin's endpoint.  Note that we are using the notation $\tilde{\lambda}_{u}$ in $f_{\mathrm{mv}, \hat{\rho}_{\tilde{\lambda}_{u},u} }$ to reflect that $\lambda$ varies with $u$ and is data dependent.  For $\hat{\mathrm{cost}}(f_{\mathrm{mv}, \hat{\rho}_{\tilde{\lambda}_{u},u} }),$ one could use $K_{n}(f_{\mathrm{mv}, \hat{\rho}_{\tilde{\lambda}_{u},u} })$, an estimate of the cost arising during the cross-validation of $\lambda$, or some other estimate such as one arising from an auxiliary testing dataset if one is available.  Minor modifications may improve the performance, for example dropping any values of $u$ from consideration in Step $2$ if there exists another $u'$ with a corresponding estimated majority vote model that has lower estimated cost but higher estimated welfare.  However, in our simulations we use the simpler version presented above.  We take $\hat{\mathrm{cost}}(f_{\mathrm{mv}, \hat{\rho}_{\tilde{\lambda}_{u},u} })$ to be the average cost associated with $f_{\mathrm{mv}, \hat{\rho}_{\tilde{\lambda}_{u},u} }$ across the hold-out fold samples during the cross-validation of $\lambda$ when estimating $\hat{\rho}_{\lambda, u}$ for the majority vote model.

The batch implementation utilizing the majority vote rules is noteworthy because, when batch implementation is feasible, it controls costs accurately.  In our simulations, we created $20$ equally spaced cost bins, starting at $0$ and with endpoints increasing from $0.1$ to $2$ by increments of $0.1$.  We treated each end point as a desired budget level and applied the batch implementation that utilizes the majority vote models.  For example, the first desired budget level is $B=0.1$ and utilizes $n_{\mathrm{bin}}=1$ in the algorithm above, while the last desired budget is $B=2$ and we set $n_{\mathrm{bin}}=20$.  Throughout, we take $B_{\mathrm{min}}=0$.   For each budget level we also applied the alternative batch implementation methods.  The gains associated with models fit to each training sample iteration were calculated using the test set.  Then these gains were averaged over all training sample iterations to produce Figures \ref{fig:BatchComparison}, \ref{fig:BatchComparison_v2} and \ref{fig:BatchComparison_v4} for $a=1, 2, 4$, respectively.   We denote the batch implementation method utilizing the majority vote models by ``PB-B", we denote the non-parametric method of \cite{sun2021treatment} centered around the conditional welfare to conditional cost ratio by ``R-NP", and we denote the baseline that ignores cost by ``Ignore Cost" or IC in subsequent discussion.  

We will refer to the non-batch-implemented stochastic Gibbs and non-stochastic majority vote methods by ``PB-G" and ``PB-MV", respectively.  To assess these methods, we utilize ``cost curves" to compare the gain-cost trade-off of the considered rules at different budget levels.  These are constructed as follows.  For a single training sample iteration, for each $u$ we estimate a Gibbs rule and a majority vote rule. We then evaluate the true cost and gain associated with these treatment rules (for different choices of $u$) using the test data.  Once we have the true costs associated with these rules, we estimate the R-NP ratios and IC CATE scores from the training sample and implement these rules via batch implementation in the testing data. For each $u$ choice and for each PB-MV and PB-G rule, the R-NP and IC rules are implemented to stop assigning treatment when they reach the same cost as the PB-MV or PB-G rule of interest.  In this way we are comparing models with the same true costs.  

For each training sample, the various (approximately) true gain-cost points associated with different $u$ choices for the PB-MV and PB-G methods are plotted in gain-cost space along with the associated points for the R-NP and IC models. The gain-cost curve for the iteration is then estimated by interpolating between these points.  For a single training sample iteration, this process is illustrated in Figures \ref{fig:SingleIteration_two} and \ref{fig:SingleIterationMJ_three} for the DGP with $a=1$.  Then, the gain-cost curves for all training sample iterations are averaged (vertically) to produce the final (approximately) true gain-cost curves.  This procedure for the DGP with $a=1$ then produces Figure   \ref{fig:GibbsAndMVCCs}.   The black lines in these figures give the gain-cost pairs that would result from randomly assigning treatment in the target population until the particular cost level is achieved.  The cost curves for the DGPs with $a=2$ and $a=4$ are presented in Figures \ref{fig:GibbsAndMVCCs_v2} and \ref{fig:GibbsAndMVCCs_v4}, respectively.  

We can now discuss the main takeaways and results from the simulation study. Figures \ref{fig:BatchComparison}-\ref{fig:SingleIterationMJ_three} present the cost curves from the simulation study while Table \ref{Tab: Simulation} collects select data points from these graphs for a more precise snapshot.  For $a=1$, the PAC-Bayesian methods PB-G, PB-MV, and PB-B perform quite closely to the R-NP method.  In this setting, the R-NP slightly outperforms the PB-G and PB-MV methods across most cost levels, with the gap in out-performance slightly increasing at greater cost levels.  The PB-B models, on the other hand, perform quite similarly to the R-NP method across the cost levels in this setting, with $\pm0.01$ differences in welfare at a few cost levels.  

As $a$ increases to $2$ and $4$, all of the PAC-Bayesian-based rules improve their performance relative to the R-NP approach.  PB-B rules yield higher welfare gains than the R-NP rules at lower to middle cost levels while slightly lagging the welfare of the R-NP models at higher cost levels for $a=2$ and slightly out-performing them at $a=4$.  The out-performance of PB-B models increases slightly at lower cost levels as $a$ increases from $2$ from $4$.  For $a\in\{2,4\}$, relative to the R-NP models, the PB-MV and PB-G models now yield higher welfare gains at lower cost levels, perform similarly at middling cost levels, and are slightly beaten at the highest budgets.  As the cost/budget level increases, the optimal rules in these simulation environments involve treating a higher proportion of the target population, eventually treating everyone as cost levels are allowed to rise enough.  The optimization problem that the Gibbs posterior solves is penalized towards allowing a degree of randomness in the resulting Gibbs rule (see, for example, the $D_{\mathrm{KL}}(\rho,\pi)$ term in \eqref{Optimization problem: sample, with budget}).  This could help to explain why the PB-G and closely related PB-MV models lag slightly at the highest cost levels whereas the PB-B implementation that treats until the budget is met performs well at these levels.  

Note that 
\[\delta_{y}(x) = 4 \Lambda\left ( 2 \left ( x_{1}+x_{1}x_{2}+x_{2} \right ) \right ), \ \delta_{c}(x) = 4\Lambda\left (a(3 x_{2} + 1.5 x_{3} ) \right ).\]
For values of $v$ near zero, $\Lambda(v)$ is approximately linear in $v$ and so the above compositions are also approximately linear in $x_{1}, x_{1}x_{2},$ $x_{2}$ and $x_{3}$ near the origin.  For values further from the origin, which are encountered with increasing probability as $a$ increases, this linear approximation worsens.  When we are likely to observe combinations of $X_{1}$, $X_{2}$, $X_{1}X_{2}$ and $X_{3}$ that are further from the origin, the conditional welfare to cost ratio in the optimal rules is a more complex object in these regions that is less well approximated by individual rules in $\mathcal{F}_{\Theta}$ and has increasing variance as $a$ increases.  This simulation study could suggest the PAC-Bayesian approaches may have benefits over the R-NP method when conditional expected costs are noisier.  

In practice it is desirable to compare alternative methods prior to implementation. For example, via evaluation using an auxialry testing data set separate from that which the models are trained on.  One drawback of the approach of \cite{sun2021treatment} and other batch implementation methods is they cannot be evaluated in a traditional way using test data withheld from model estimation.  For example, test sample data points that a batch implementation method may rank highly for treatment may not have received the treatment and thus we do not observe the costs accruing properly to know when a batch implementation method would stop assigning treatments.  We note that \cite{sun2021treatment} is a working paper and since this paper was started the authors have added material aimed at addressing this issue.  

It also is important to note that there are a number of settings where the forest-based R-NP method is not viable whereas the PAC-Bayesian approaches considered here remain applicable.  Batch implementations are not always viable.  The cost of a treatment may not be realized until sometime after treatment assignment and one may not always have the full target group available when the rule must be set.  Batch implementations, where treatment is assigned until the budget is hit, could also be unacceptable to policy makers in settings where the ``budget" is something with a negative connotation like a complication rate in a medical setting.  

Additionally, the R-NP rule can only be applied when $C_{0}\leq  C_{1}$ a.s., which rules out certain circumstances relevant to policy makers.  For example, as noted in \cite{sun2021empirical}, \cite{hendren2020unified} identify fourteen welfare programs out of 133 considered that are estimated to have negative or zero net cost to the government.  The EWM based approach of  \cite{sun2021empirical} can accommodate the setting where $C_{0}> C_{1}$ with positive probability, as can the PB-G and PB-MV methods considered here.  However, the approach of \cite{sun2021empirical} may be difficult to implement when allowing for more flexible decision rule classes (she considers threshold rules that vary with a covariate in her application) and lacks the budget efficiency properties derived here.  An additional benefit of the PAC-Bayesian approaches here is their ability to utilize estimation tools from the Bayesian literature as demonstrated in Section \ref{subsec: SMC implementation} below.  Lastly, while one could estimate  $\delta_{y}(x)$ and $\delta_{c}(x)$ separately and try to build a workaround via Theorem \ref{Theorem: theoretically optimal policy choice} when $C_{0}>C_{1}$ is possible, the resulting ratio estimates may have increased variance and will again require batch implementation, which adds a complication in this setting. 

\subsection{Implementation and Estimation via Sequential Monte Carlo \label{subsec: SMC implementation}}
To implement treatment rules associated with $\hat{\rho}_{\lambda, u}(\theta)$, we must evaluate the treatment assignment probabilities or majority vote scores of the form
\begin{equation}
\label{eqn: SMC integral to approximate}
\int_{\Theta}f_{\theta}(x)d\hat{\rho}_{\lambda, u}(\theta), \ \ x\in\mathcal{X}.
\end{equation}
To do so, we utilize the Sequential Monte Carlo (SMC) procedure considered, for example, in \cite{del2006sequential}.  While a Markov Chain Monte Carlo (MCMC) approach also could be derived, recently \cite{ridgway2014pac} and \cite{alquier2016properties} have highlighted the usefulness of the SMC procedure in PAC-Bayesian applications.  One benefit is the ability to sample from a sequence of Gibbs posterior distributions for a range of $\lambda$ values.  This can ease the computational burden for cross-validation.  Here we discuss key elements of the approach, provide an estimation algorithm for our setting, and discuss implementation.  We also discuss the choices utilized in implementing the procedure for Section \ref{subsec: Simulation Setup}.    

Throughout, we make the following computational adjustment to the definition of $\hat{\rho}_{\lambda,u}$ in order to make the implementation choices for Section \ref{subsec: Simulation Setup} applicable to more general settings.  We define $\hat{\rho}_{\lambda, u}$ to be the distribution over $\Theta$ with RN derivative with respect to $\pi$ given by 
\begin{equation}
\label{normalized rho hat}
\frac{d\hat{\rho}_{\lambda, u}}{d \pi} \left ( \theta \right ) = \frac{\exp \left [ -\lambda \left ( u\overline{K}_{n}(\theta) - \overline{W}_{n}(\theta)\right ) \right ] }{Z(\lambda, u)},
\end{equation}
where 
\[Z(\lambda, u)=\int_{\Theta} \exp \left [ -\lambda \left ( u\overline{K}_{n}(\theta) - \overline{W}_{n}(\theta)\right ) \right ] d\pi(\theta),\]
and
\[\overline{W}_{n}(\theta)= \frac{W_{n}(\theta)}{\frac{1}{n}\sum_{i=1}^{n}\delta_{y,i}} \ \mathrm{and} \ \overline{K}_{n}(\theta) = \frac{K_{n}(\theta)}{\frac{1}{n}\sum_{i=1}^{n}\delta_{y,i}}.  \]
This adjustment is relevant when the average treatment effect is expected to be is positive.  Clearly, if we denote $\hat{\delta}_{y}=n^{-1}\sum_{i=1}^{n}\delta_{y,i}$, the distribution $\hat{\rho}_{\lambda, u}$ in \eqref{normalized rho hat} is equivalent to $\hat{\rho}_{(\hat{\delta}_{y}\lambda),u}$ in Definition \ref{Def: u hat and u star}.  In practice we choose $\lambda$ via cross-validation from a wide range of values.   

For given choices of $\lambda>0$ and $u\geq 0$, the SMC algorithm we adopt samples from $\hat{\rho}_{\lambda, u}$ to evaluate \eqref{eqn: SMC integral to approximate} by simulating a set of parameter draws from each of a sequence of distributions $\{\hat{\rho}_{\lambda_{t}, u} \}_{t=0}^{T}$.  Here,
\[ 0 =\lambda_{0} < \lambda_{1} < \cdots < \lambda_{T} = \lambda \]
is an increasing temperature ladder that must be specified.  Note that $\hat{\rho}_{\lambda_{0},u} = \pi$, which the user may specify and we assume can be sampled from.  The temperature ladder $\{\lambda_{t}\}_{t=0}^{T}$ is intended to be such that the corresponding distributions $\hat{\rho}_{\lambda_{t}, u}$ progress gradually from $\pi$ to the target distribution $\hat{\rho}_{\lambda, u}$.  

For each $t=0,\dots, T$, the SMC algorithm produces a set of $N$ weighted samples $\{\Psi_{t}^{(i)}, \theta_{t}^{(i)}\}_{i=1}^{N}$ with $\Psi_{t}^{(i)}>0$ and $\sum_{i=1}^{N} \Psi_{t}^{(i)} =1$ where $\theta_{t}^{(i)}\in\Theta$ for all $t$ and $i$ in our setting.  The set of parameter draws $\{\theta_{t}^{(i)}\}_{i=1}^{N}$ are referred to as particles (there are $N$ weighted particles for each $t$).  SMC combines MCMC moves with sequential importance sampling; we refer to \cite{del2006sequential} for additional details and discussion.  This produces weighted particles that emulate, in terms of computing expectations, samples from the distributions $\hat{\rho}_{\lambda_{t},u}$  associated with 
\begin{equation*}
\frac{d\hat{\rho}_{\lambda_{t}, u}}{d \pi} \left ( \theta \right ) = \frac{\exp \left [ -\lambda_{t} \left ( u\overline{K}_{n}(\theta) - \overline{W}_{n}(\theta)\right ) \right ] }{Z_{t}}, \ \ Z_{t}=\int_{\Theta} \exp \left [ -\lambda_{t} \left ( u\overline{K}_{n}(\theta) - \overline{W}_{n}(\theta)\right ) \right ] d\pi(\theta).
\end{equation*}
Conditional on $\hat{\rho}_{\lambda_{T},u}$, under general conditions, for a $\hat{\rho}_{\lambda_{T},u}$-integrable function $\varphi:\Theta\rightarrow \mathbb{R}$, 
\[\sum_{i=1}^{N}\Psi_{T}^{(i)}\varphi\left ( \theta_{T}^{(i)} \right )\overset{a.s.}{\rightarrow } \int_{\Theta} \varphi \left (\theta \right ) d\hat{\rho}_{\lambda_{T}, u}(\theta) \ \ \mathrm{as} \ N\rightarrow\infty.\]

In our setting, we are interested in $\varphi(\theta)=f_{\theta}(x)$ to approximate \eqref{eqn: SMC integral to approximate}  via
\[\sum_{i=1}^{N}\Psi_{T}^{(i)}f_{\theta_{T}^{(i)}}(x), \ \ x\in\mathcal{X}.\]
Once we have run the SMC algorithm to yield $\{\Psi_{T}^{(i)}, \theta_{T}^{(i)}\}_{i=1}^{N}$ for a given pair $(\lambda, u)=(\lambda_{T}, u)$, the treatment probability or majority vote score for any value $x$ in the covariate space can be computed as above.  Alternatively, for example, we may be interested in $\varphi(\theta) = K_{n}(\theta)$, to approximate $K_{n}(f_{G,\hat{\rho}_{\lambda, u}})$.  

The SMC algorithm utilized to estimate the treatment rules in Section \ref{subsec: Simulation Setup} is detailed in the algorithm tables below.  We set the input parameters $\tau_{\mathrm{ESS}}$ and $N$ there equal to $1/2$ and $1,000$, respectively.  $\tau_{\mathrm{ESS}}$ is an Effective Sample Size threshold criterion.  When the variance of the weights at a given step $t$ is too high, the SMC procedure utilizes a re-sampling step.  This is referred to in Step 2 of the algorithm below.  In our application we utilize systematic resampling, which is also outlined below.  The choice of temperature ladder, additional algorithm details, and cross-validation points are detailed below the algorithm descriptions.  

\bigskip  

\noindent\hrulefill

\noindent\textbf{Tempering SMC Algorithm}

\vspace{-9pt} \noindent\hrulefill

\vspace{0.1in} \noindent\textbf{Input} $N$ (number of particles),
$\tau_{\mathrm{ESS}}\in(0,1)$ (ESS threshold), $\{\lambda_{t}\}_{t=1}^{T}$ (temperature ladder).

\vspace{0.1in} \noindent\textbf{Output} $\{\Psi_{t}^{(i)}, \theta_{t}%
^{(i)}\}_{i=1}^{N}$ for $t=0,\dots, T$.

\vspace{0.1in} \noindent Step 1: initialization

\begin{itemize}
	\item Set $t\leftarrow0$. For $i=1,\dots, N$, draw
	$\theta_{0}^{(i)}\sim\pi$ and set $\Psi_{0}^{(i)}\leftarrow1/N$.
\end{itemize}

\noindent Iterate steps 2 and 3

\noindent Step 2: Resampling

\begin{itemize}
	\item If
	\[
	\left\{  \sum_{i=1}^{N}\left(  \Psi_{t}^{(i)}\right)  ^{2}\right\}  ^{-1}%
	<\tau_{\mathrm{ESS}}N,
	\]
	resample $\left\{  \Psi_{t}^{(i)},\theta_{t}^{(i)}\right\}  _{i=1}^{N}$ yielding
	equally weighted resampled particles $\left\{  \frac{1}{N},\overline{\theta
	}_{t}^{(i)}\right\}  _{i=1}^{N}$ and set $\left\{  \Psi_{t}^{(i)},\theta
	_{t}^{(i)}\right\}  _{i=1}^{N}\leftarrow\left\{  \frac{1}{N},\overline{\theta
	}_{t}^{(i)}\right\}  _{i=1}^{N}$. Otherwise, leave $\left\{  \Psi_{t}%
	^{(i)},\theta_{t}^{(i)}\right\}  _{i=1}^{N}$ unaltered.
\end{itemize}

\noindent Step 3: Sampling

\begin{itemize}
	\item Set $t\leftarrow t+1$; if $t=T+1$, stop.
	
	\item For $i=1,\dots,N$, draw $\theta_{t}^{(i)}\sim K_{t}(\theta_{t-1}%
	^{(i)},\cdot)$, where $K_{t}$ is an MCMC kernel with invariant distribution
	$\hat{\rho}_{\lambda_{t}, u}$, and evaluate the unnormalized importance weights
	\[
	\omega_{t}^{(i)}\left(  \theta_{t-1}^{(i)}\right)  =\exp\left[ \lambda_{t-1}\left ( u\overline{K}_{n} \left(  \theta_{t-1}^{(i)}\right) -\overline{W}_{n}\left(  \theta_{t-1}^{(i)}\right) \right )  -\lambda_{t}\left ( u\overline{K}_{n} \left(  \theta_{t-1}^{(i)}\right) -\overline{W}_{n}\left(  \theta_{t-1}^{(i)}\right) \right )    \right]  .
	\]

	\item For $i=1,\dots, N$, set
	\[
	\Psi_{t}^{(i)} \leftarrow\frac{\Psi_{t-1}^{(i)} \omega_{t} \left(  \theta
		_{t-1}^{(i)} \right)  }{\sum_{j=1}^{N} \Psi_{t-1}^{(j)} \omega_{t} \left(
		\theta_{t-1}^{(j)} \right)  }.
	\]
	
\end{itemize}
\vspace{-8pt} \noindent\hrulefill

\vspace{0.15in}

\noindent\hrulefill

\noindent\textbf{Resampling Algorithm (systematic resampling): }

\vspace{-9pt} \noindent\hrulefill

\noindent\textbf{Input} A set of (normalized) weights and associated
particles, $\left\{  \Psi_{t}^{(i)}, \theta_{t}^{(i)} \right\}  _{i=1}^{N}$ for
some $t\in\{0,\dots, T\}$.

\noindent\textbf{Output} Resampled particles for equal weighting, $\left\{
\overline{\theta}_{t}^{(i)} \right\}  _{i=1}^{N}$

\begin{itemize}
	\item Draw $u\sim U\left[  0,\frac{1}{N}\right]  $.
	
	\item Compute cumulative weights $C^{\left(  i\right)  }=\sum_{m=1}^{i}%
	\Psi_{t}^{(m)}$ for $i=1,...,N$.
	
	\item Set $m\leftarrow1$.
	
	\item \textbf{For} $i=1:N$
	
	\hspace{0.25in} \textbf{While} $u < C^{(i)}$ \textbf{do} $\overline{\theta
	}_{t}^{(m)}\leftarrow\theta_{t}^{(i)}$.
	
	\hspace{0.25in} $m\leftarrow m+1$, and $u\leftarrow u+1/N$.
	
	\noindent\textbf{End For}
\end{itemize}

\vspace{-8pt} \noindent\hrulefill

\bigskip 

Some additional implementation details utilized in Section \ref{subsec: Simulation Setup} are as follows.  For the MCMC kernel in the sampling step of the
SMC algorithm, we use a Gaussian random-walk Metropolis kernel with covariance
matrix proportional to the empirical covariance matrix of the current set of
particles. We scale the empirical covariance of the
step $t$ particles by $t^{-0.9}$.  In practice, the scaling factor can be adjusted, possibly dynamically rather than with a general rule like $t^{-0.9}$, to ensure reasonable acceptance rates in the MCMC steps of the SMC algorithm. 

For a temperature ladder, we set $T=800$ and utilize a piece-wise linear structure as follows.   $\{\lambda_{t}\}_{t=0}^{200}$ increases from $0$ to $4/(|1+u|)$ in equally spaced steps. Then,  $\{\lambda_{t}\}_{t=201}^{320}$ increases from $4/(|1+u|)$ to $32/(|1+u|)=2^5/(|1+u|)$ in equally spaced increments,  $\{\lambda_{t}\}_{t=321}^{470}$ increases from  $2^5/(|1+u|)$ to $2^8/(|1+u|)$ in equally spaced increments, and $\{\lambda_{t}\}_{t=471}^{800}$ increases from $2^8/(|1+u|)$ to $2^{10}/(|1+u|)$ in equally spaced increments. For $\lambda$ values of interest less than $2^{10}$, the temperature ladder is cut short (at fewer than 800 steps) to end once $\lambda_{t}$ reaches the desired value.  Let $\mathcal{W}_{\lambda}$ denote the elements of  $\{\lambda_{t}\}_{t=0}^{800}$ that are closest to elements in $\{ 2^{2}, (2^{2}+2^{3})/2, 2^{3}, (2^{3}+2^{4})/2,2^{4},\dots, 2^{10} \}/(|1+u|)$.  Rather than considering each value in $\{\lambda_{t}\}_{t=0}^{800}$ as a potential choice for a rule, during cross-validation $\lambda$ chosen from values in $\mathcal{W}_{\lambda}$.

\section{Empirical Illustration}
\label{Sec: Empirical Illustration Gibbs}

Here, we illustrate the procedure for Gibbs treatment rules centered around $\hat{\rho}_{\lambda,u}$ using data from the National Job Training Partnership Act (JTPA) Study. This study has been a popular choice for illustrating individualized treatment rule estimators and is utilized, for example, in \cite{KT2018}, \cite{mbakop2021model}, and \cite{kitagawa2023stochastic}. Detailed descriptions of the study can be found in \cite{orr1994national} and \cite{bloom1997benefits}.

The JTPA study was a randomized controlled trial aimed at assessing the costs and benefits of the training and employment assistance programs of the JTPA. The study randomly assigned each participant to one of two groups. In the first group (the treatment group), participants had access to JTPA services, whereas in the second group (the control group), access to JTPA services was restricted. For example, access was limited to certain services, and a period of time was imposed when a control individual would be ineligible for services. Note that the treatment was ease of access to services, rather than participation in JTPA programs or any other type of compliance. The study collected background information on participants and tracked their earnings in the 30-month period following treatment assignments.

As in \cite{KT2018}, we use an individual's total earnings in the 30 months following treatment as the outcome variable of interest ($Y$). Our Gibbs treatment rules, like those proposed in the referenced above, are based on two variables that policymakers might consider in designing access policies: an individual's years of education and their earnings in the year prior to treatment assignment. JTPA personnel assessed all participants prior to treatment assignments and provided service type recommendations, which were categorized by \cite{orr1994national} into three types: classroom training, on-the-job training/job search assistance, and other services. To construct our cost variable, we use averages related to these categories, as described below.

Although treatment costs varied between individuals in the study, we don't have exact costs per person. Instead, we define the cost variable $C$ as follows: if an individual received $0$ hours of JTPA services, we set their cost to $0$. Otherwise, we take their cost to be the average cost of services for individuals with the same gender, treatment assignment, and service recommendation category. These averages are reported in Exhibit 5.3 of \cite{orr1994national} and were adjusted for our purposes to reflect the averages among individuals who received services, using the proportion of individuals in each subcategory who received more than $0$ hours of services.  As noted in \cite{orr1994national}, it is relevant that services utilized correlated with service recommendations and individuals in the control group that did access JTPA services at some point in the study tended to incur similar costs as individuals in the treatment group that received services.  The probability of using any services, on the other hand, was greatly impacted by treatment assignment and is a key driver in cost differences.  

Our sample consists of 7,675 adults (22 years and older) for whom data on years of education, pre-program earnings, and service hours received are available.  As in \cite{KT2018}, we only consider individuals from the original program evaluation and studies around it (e.g. \cite{bloom1997benefits} and other references provided in \cite{KT2018}).  The probability of being assigned treatment is $2/3$ in this sample.  To estimate potential models of interest, we utilize the SMC procedure described in Section \ref{subsec: SMC implementation}.  We consider $u\in\{0.05, 0.1,\dots, 3\}$ and cross-validate $\tilde{\lambda}_{u}\in \{2, 2^{2},\dots, 2^{10}\}/(1+u)$ for each.   During the cross-validation step, for each $u$ we obtain an estimated cost and welfare, namely the averages of $\hat{B}(\hat{\rho}_{\tilde{\lambda}_{u},u})$ and $W_{n}(f_{G,\hat{\rho}_{\tilde{\lambda},u}})$ across hold out folds.  Note we are abusing notation here because, during cross-validation, objects involving  $\hat{\rho}_{\tilde{\lambda},u}$ are calculated from the k-fold training sample rather than the entire sample while the cost and welfare estimates are averages of objects calculated using hold out samples.  

We use 2-fold cross-validation because the 30-month post-treatment earnings data is highly variable, and there is a potential to overfit noise in smaller cross-validation samples. We take $\mathcal{F}_{\Theta}$ to be the family of rules described in \eqref{Definition: class of models} and \eqref{Example treatment model class}, where the transformations $\phi_{j}(x)$ are the monomials used in the construction of polynomial transformations on $\mathbb{R}^{2}$ of order at most 3. As in the simulation section, we normalize the monomial transformations by subtracting their sample means and scaling by the sample standard deviations since there is a considerable degree of variation in scale among the transformations, with education taking values from 7 to 18 years and pre-program earnings ranging from \$$0$ to \$$63{,}000$.

The cross-validation-based estimates of the welfare and cost pairs for different values of $u$ are plotted below in Figure \ref{fig:empiricalexgibbsgaincostcvgibbsonlyestimates}. We dropped points corresponding to models with dominated welfare-cost pairs, i.e., points where there existed an alternative model for a different choice of $u$ with a higher estimated welfare for the same or lower estimated cost, and these are not displayed. Following an approach where we consider multiple values of $u$ and examine feasible budget estimates of the policies, we can see that there are models with average costs per person ranging from roughly \$$100$ to about \$$650$. That is, we estimate that when these treatment rules are applied to a wider population similar to the sample, we can achieve average costs per person within these ranges.

We examine three of the estimated models corresponding to the circled points in Figure \ref{fig:empiricalexgibbsgaincostcvgibbsonlyestimates} in greater detail. Starting from the left, the first circle corresponds to the model with an estimated cost closest to \$$200$ and with $u=2.45$.  We refer to this model as the ``low-budget model." The second circled point corresponds to the model with an estimated cost nearest to \$$400$, with $u=2.1$, and we call this model the ``medium-budget model." The right circle, which we refer to as the "no budget model," has the highest estimated welfare and corresponds to $u=0.15$. Figure \ref{fig:allheatmapsgibbscv} presents treatment probabilities for different values of education and pre-program earnings associated with the high, medium, and no budget models. The population densities linked to the covariate space for these models are shown in Figure \ref{fig:populationdensityheatmap}.

\begin{figure}[H]
	\centering
	\includegraphics[width=0.6\linewidth]{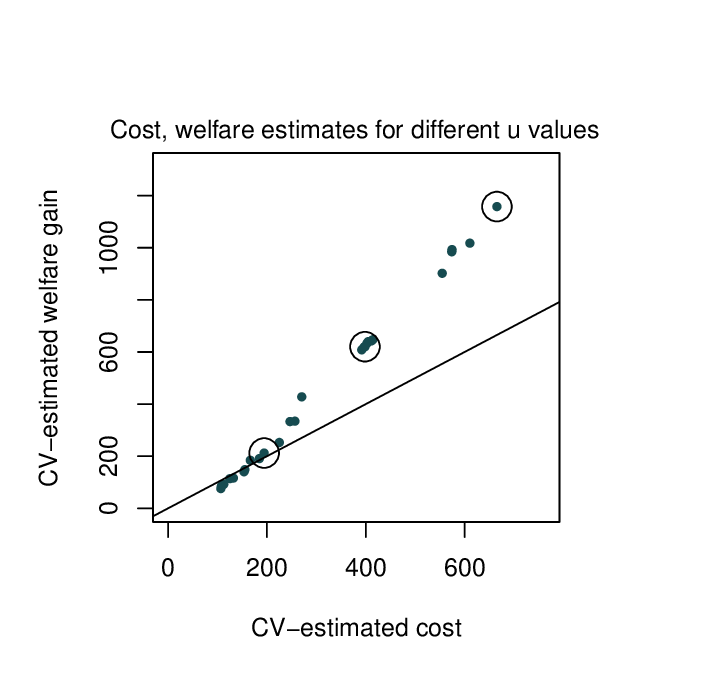}
	\caption{Estimated welfare and cost pairs associated with $\hat{\rho}_{\tilde{\lambda}_{u},u}$ for different values of $u$.  The straight line represents points that would have equal cost and welfare.  Models with cost estimates closes to \$$200$ and \$$400$ are circled, as is the model with the highest estimated welfare gain.}
	\label{fig:empiricalexgibbsgaincostcvgibbsonlyestimates}
\end{figure}

\vspace{0.1in}
\begin{figure}[H]
	\centering
	\includegraphics[width=0.9\linewidth]{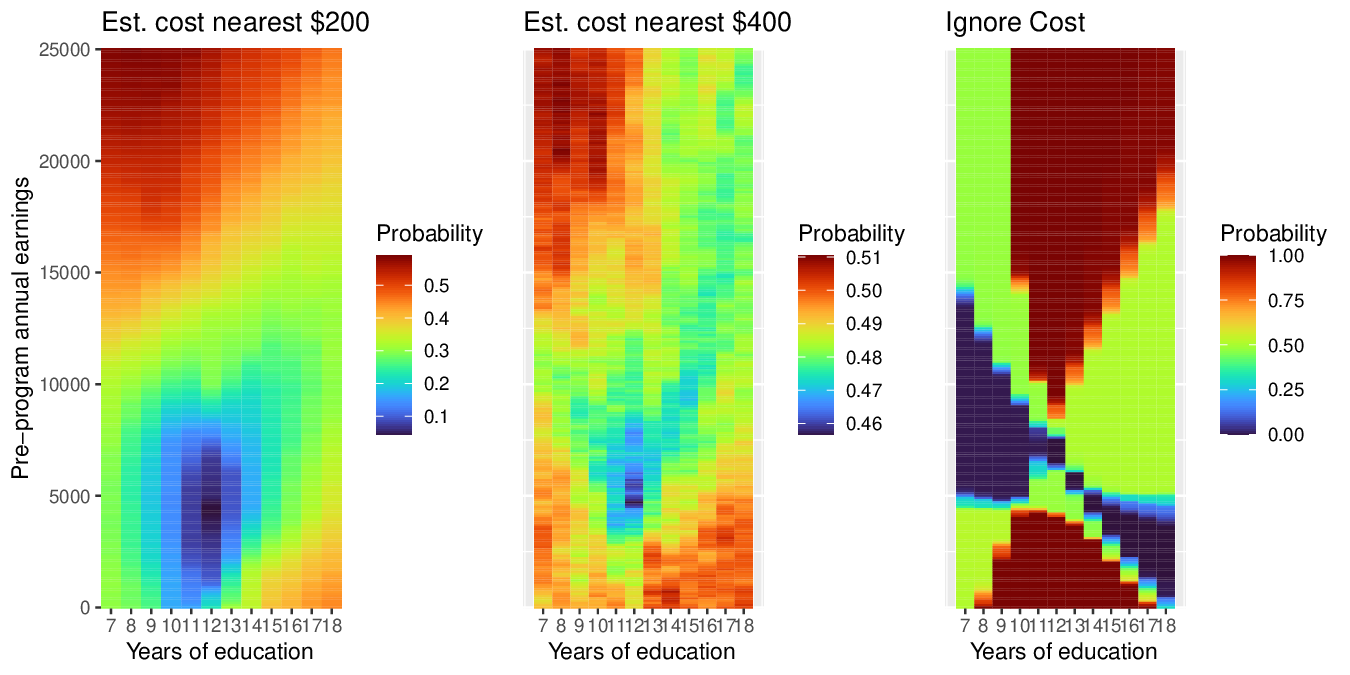}
	\caption{Treatment probabilities associated with the low-budget, medium-budget, and no-budget models over different regions of the covariate space.}
	\label{fig:allheatmapsgibbscv}
\end{figure}

\begin{figure}[H]
	\centering
	\includegraphics[width=0.6\linewidth]{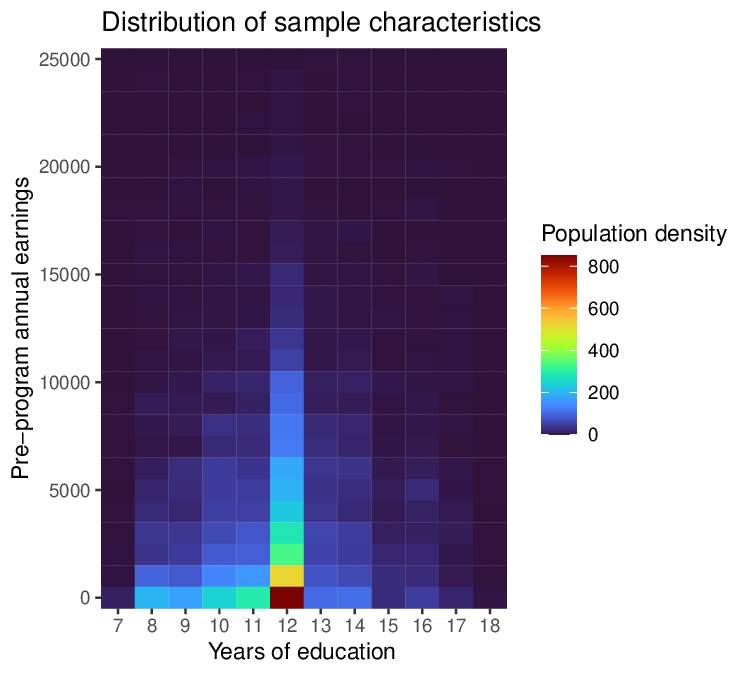}
	\caption{Population densities for different regions of the covariate space.  The color at each block represents the number of individuals in the sample that fall into the associated education and pre-program earnings level.}
	\label{fig:populationdensityheatmap}
\end{figure}

We observe that treatment probabilities transition fairly smoothly across the covariate space in all three models but less so in the no budget model. Treatment assignment probabilities are not uniform across the covariate space in any of the estimated models. Furthermore, as we consider models with lower costs, we do not simply obtain uniform reductions in treatment probabilities associated with higher-cost models. In the no budget model, treatment probabilities are either $1$ or $0$ for sizable portions of the covariate space. The no-budget model treats $81\%$ of individuals in the sample on average. The medium-budget model treats $49\%$ of individuals on average, and the treatment probabilities among individuals in the sample range from $46\%$ to $51\%$. The low-budget model treats $20\%$ of individuals on average, and treatment probabilities among covariate pairs encountered in the sample range from $4\%$ to $66\%$.  

One common feature among most of the (non-stochastic) treatment rules estimated in \cite{KT2018} is the avoidance of treating individuals with the highest levels of education. The most comparable model to those considered here, the no budget model, differs from these rules in that it treats individuals with higher levels of education above a certain income level and randomizes treatment in regions of the covariate space with lower education. It also avoids treating individuals with middling pre-program earnings and low education, although these individuals are relatively less common. As we move to models with lower costs, the medium-budget model becomes more uniform in its treatment probabilities than the other Gibbs rules.  It also features increased treatment probabilities among lower and higher education levels compared to the other models, particularly at lower income levels. The low-budget model reduces cost by lowering treatment probabilities among individuals with lower to middle education levels, especially at the most commonly encountered non-zero education levels.

\section{Conclusion \label{sec: Conclusion}} 
In this paper, we proposed a new approach to estimating treatment rules in a budget constrained setting.  Utilizing the PAC-Bayesian framework, theoretical properties of interest were derived, including generalization bounds and oracle-type inequalities demonstrating a type of budget efficiency for a proposed class of stochastic treatment rules. The treatment rule estimators can accommodate a variety of budget constraints of interest including settings with uncertain and or heterogeneous costs, quantity constraints, and settings where costs are not realized at the time of treatment.   Another benefit is that the proposed rules can take advantage of well developed Bayesian estimation machinery.  Lastly, the models were shown to be competitive against state-of-the-art alternatives in a simulation study and an empirical illustration was examined. 

There are a number of considerations for future work.  It would be of interest to determine if different prior choices, for example a sparsity-inducing prior or a normal prior with a different form for the covariance matrix than that considered here, would facilitate bounds of the type in Section \ref{subsec: Normal Prior} or yield modeling suggestions for higher dimensional feature spaces.  Rather than using the Gibbs posterior to form treatment rules, it could also prove fruitful to approximate the Gibbs posterior with alternative distribution such as a normal distribution.  So-called variational approximations of Gibbs posteriors for general PAC-Bayesian approaches are considered in \cite{alquier2016properties}.  This could yield greater flexibility in terms of functional form constraints, beyond control over the variables included in the treatment rules featured here.    It would also be of interest to incorporate estimated propensity scores into the PAC-Bayesian framework here and to explore how this impacts rates of convergence.  Lastly, the analysis for balancing the primary welfare or regret objective against that of a secondary cost objective can be generalized to settings beyond the welfare-based potential outcomes framework.  Balancing a secondary objective of concern could also be of interest in classification or regression settings.

\newpage 

\vspace{1in}

\begin{table}[ht]
	\vspace{1.5in}
\caption{Simulation welfare gains for models at different cost levels}
\centering 
	\begin{tabular}{cccccc}
		\hline 
		Cost & PB-G & PB-MV & PB-B & R-NP & Ignore Cost\tabularnewline
		& \multicolumn{5}{c}{$a=1$}\tabularnewline
		\cline{2-6} \cline{3-6} \cline{4-6} \cline{5-6} \cline{6-6} 
		$0.1$ & $0.27$ & $0.28$ & $0.28$ & $0.28$ & $0.11$\tabularnewline
		$0.3$ & $0.56$ & $0.57$ & $0.58$ & $0.57$ & $0.34$\tabularnewline
		$0.6$ & $0.90$ & $0.91$ & $0.92$ & $0.92$ & $0.70$\tabularnewline
		$0.9$ & $1.20$ & $1.22$ & $1.23$ & $1.23$ & $1.04$\tabularnewline
		$1.2$ & $1.46$ & $1.48$ & $1.49$ & $1.49$ & $1.34$\tabularnewline
		$1.5$ & $1.66$ & $1.68$ & $1.69$ & $1.70$ & $1.57$\tabularnewline
		$1.8$ & $1.79$ & $1.81$ & $1.82$ & $1.82$ & $1.74$\tabularnewline
		& \multicolumn{5}{c}{$a=2$}\tabularnewline
		\cline{2-6} \cline{3-6} \cline{4-6} \cline{5-6} \cline{6-6} 
		$0.1$ & $0.46$ & $0.47$ & $0.47$ & $0.45$ & $0.10$\tabularnewline
		$0.3$ & $0.71$ & $0.72$ & $0.73$ & $0.70$ & $0.31$\tabularnewline
		$0.6$ & $1.01$ & $1.01$ & $1.03$ & $1.01$ & $0.63$\tabularnewline
		$0.9$ & $1.28$ & $1.29$ & $1.30$ & $1.29$ & $0.96$\tabularnewline
		$1.2$ & $1.51$ & $1.53$ & $1.54$ & $1.54$ & $1.26$\tabularnewline
		$1.5$ & $1.68$ & $1.70$ & $1.71$ & $1.72$ & $1.52$\tabularnewline
		$1.8$ & $1.79$ & $1.81$ & $1.82$ & $1.83$ & $1.71$\tabularnewline
		& \multicolumn{5}{c}{$a=4$}\tabularnewline
		\cline{2-6} \cline{3-6} \cline{4-6} \cline{5-6} \cline{6-6} 
		$0.1$ & $0.60$ & $0.61$ & $0.61$ & $0.57$ & $0.10$\tabularnewline
		$0.3$ & $0.80$ & $0.80$ & $0.82$ & $0.78$ & $0.30$\tabularnewline
		$0.6$ & $1.07$ & $1.08$ & $1.09$ & $1.07$ & $0.61$\tabularnewline
		$0.9$ & $1.33$ & $1.34$ & $1.35$ & $1.34$ & $0.92$\tabularnewline
		$1.2$ & $1.55$ & $1.56$ & $1.58$ & $1.57$ & $1.24$\tabularnewline
		$1.5$ & $1.70$ & $1.72$ & $1.74$ & $1.74$ & $1.51$\tabularnewline
		$1.8$ & $1.80$ & $1.82$ & $1.84$ & $1.83$ & $1.70$\tabularnewline
		\hline 
	\end{tabular}
\label{Tab: Simulation}
\end{table}

\newpage 

\begin{center}
	Figures for $a=1$
\end{center}

\begin{figure}[H]
	\centering
	\includegraphics[width=0.5\linewidth]{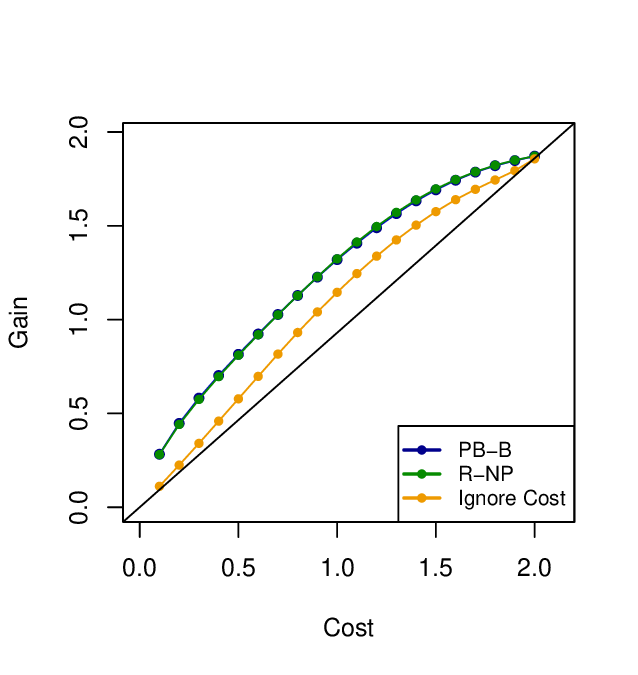}
	\caption[Simulation cost curves for rules featuring batch implementation when $a=1$]{Cost curves when all methods utilize batch implementation for the DGP featuring $a=1$.}
	\label{fig:BatchComparison}
\end{figure}

\begin{figure}[H]
	\centering
	\includegraphics[width=1\linewidth]{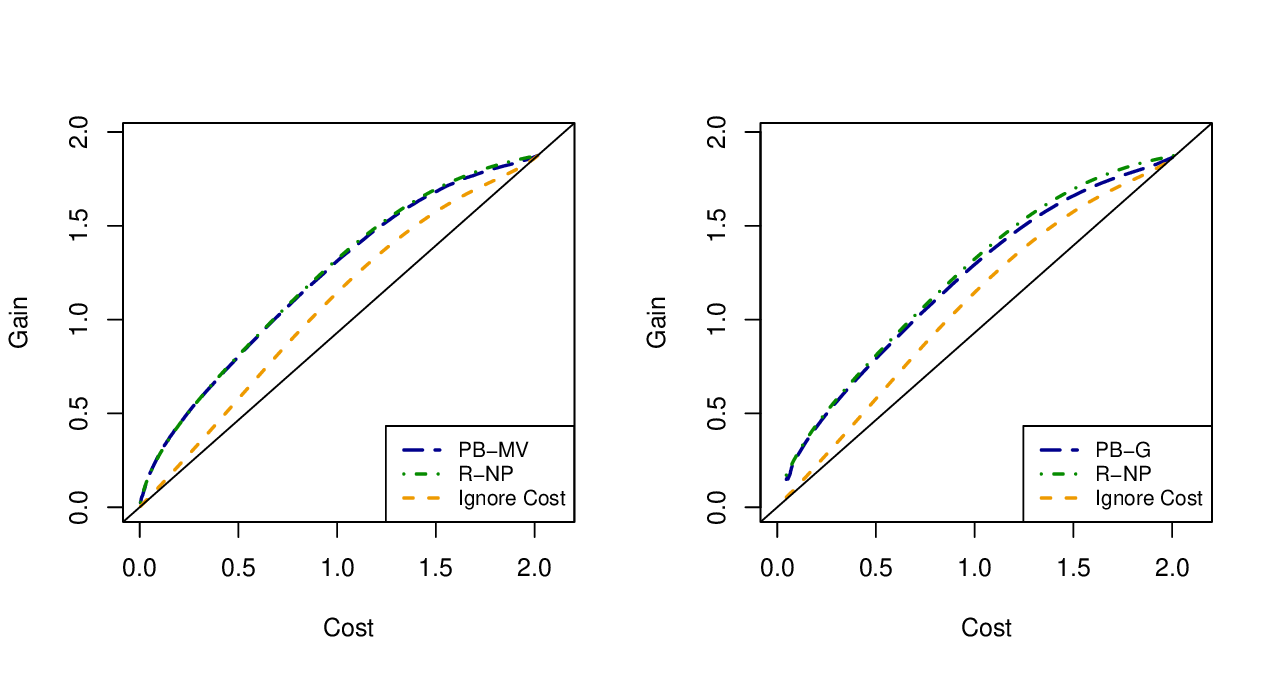}
	\caption[Simulation cost curves for Gibbs and majority vote methods when $a=1$]{With $a=1$ in the DGP, cost curves for the PB-MV and PB-G methods, which do not feature batch implementation, compared with the batch-implemented R-NP and IC methods.}
	\label{fig:GibbsAndMVCCs}
\end{figure}
\newpage 

\begin{center}
	Figures for $a=2$
\end{center}

\begin{figure}[H]
	\centering
	\includegraphics[width=0.5\linewidth]{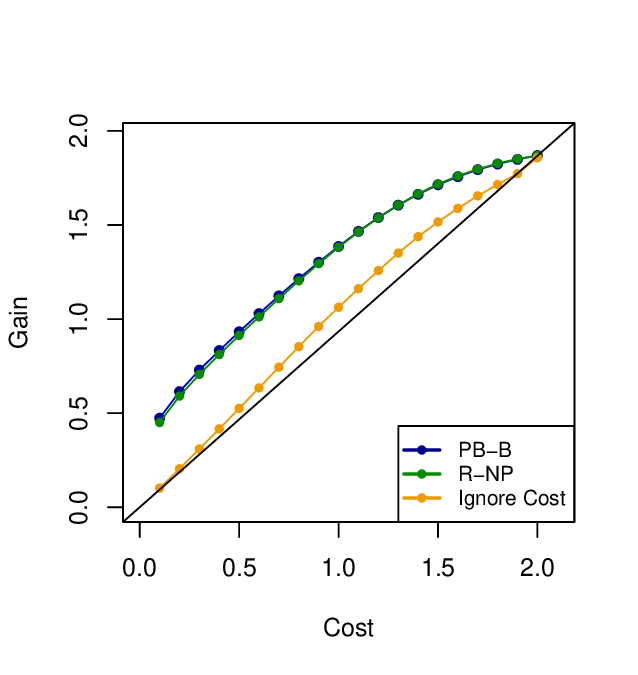}
	\caption[Simulation cost curves for rules featuring batch implementation when $a=2$]{Cost curves when all methods utilize batch implementation for the DGP featuring $a=2$.}
	\label{fig:BatchComparison_v2}
\end{figure}

\begin{figure}[H]
	\centering
	\includegraphics[width=1\linewidth]{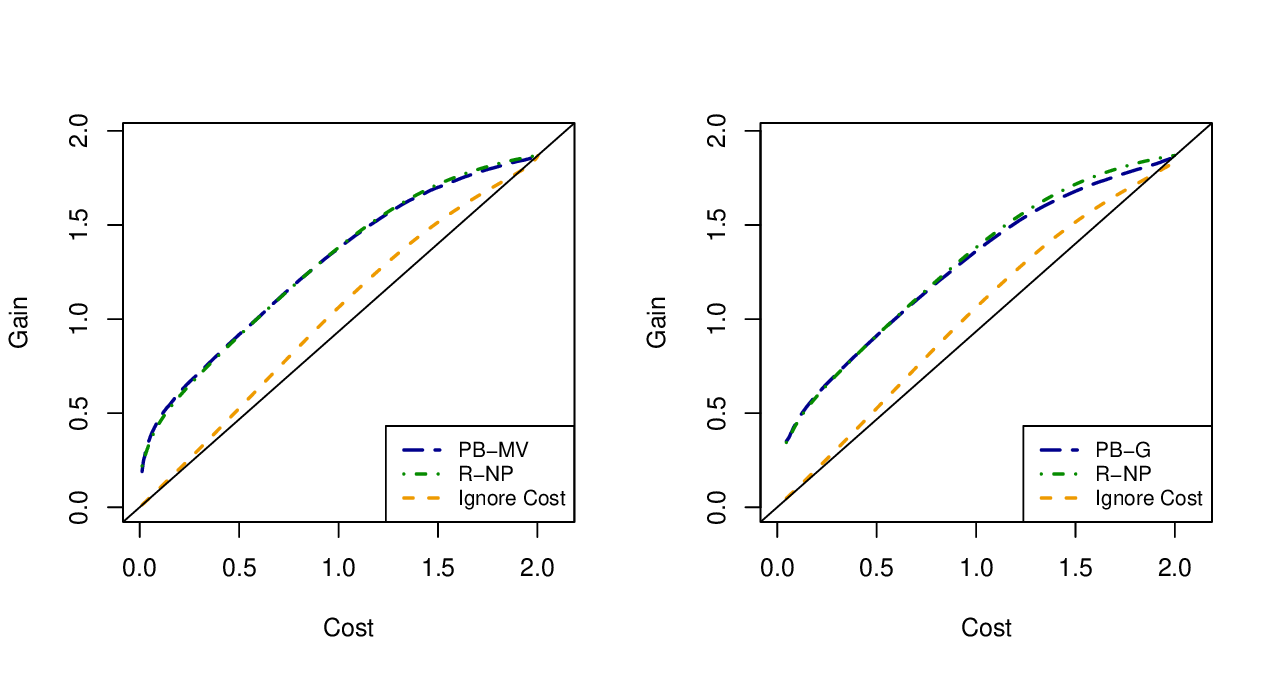}
	\caption{With $a=2$ in the DGP, cost curves for the PB-MV and PB-G methods, which do not feature batch implementation, compared with the batch-implemented R-NP and IC methods.}
	\label{fig:GibbsAndMVCCs_v2}
\end{figure}
\newpage 

\begin{center}
	Figures for $a=4$
\end{center}

\begin{figure}[H]
	\centering
	\includegraphics[width=0.5\linewidth]{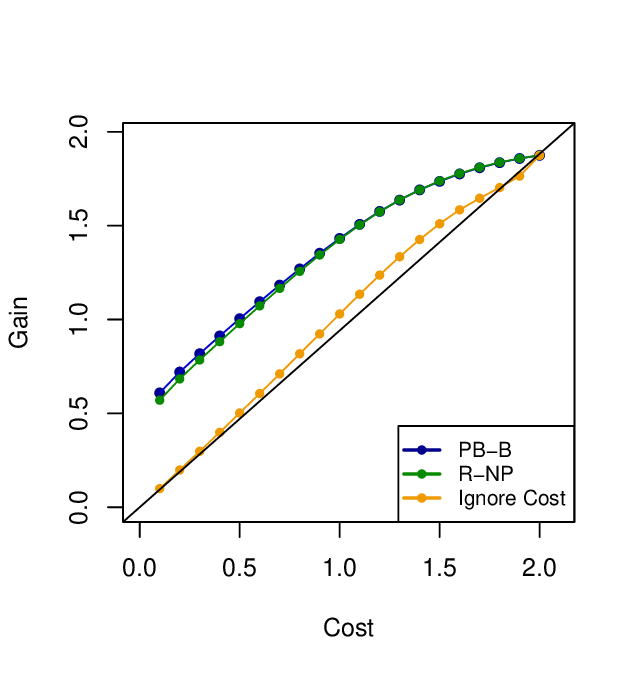}
	\caption[Simulation cost curves for rules featuring batch implementation when $a=4$]{Cost curves when all methods utilize batch implementation for the DGP featuring $a=4$.}
	\label{fig:BatchComparison_v4}
\end{figure}

\begin{figure}[H]
	\centering
	\includegraphics[width=1\linewidth]{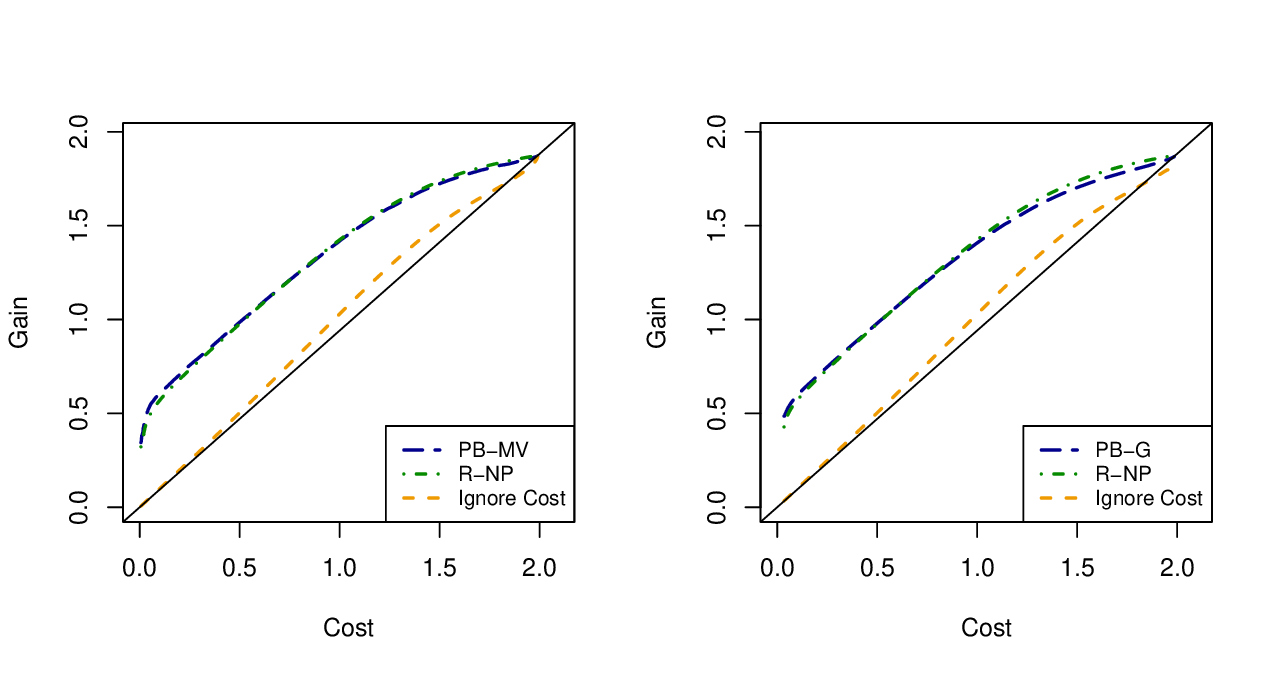}
	\caption{With $a=4$ in the DGP, cost curves for the PB-MV and PB-G methods, which do not feature batch implementation, compared with the batch-implemented R-NP and IC methods.}
	\label{fig:GibbsAndMVCCs_v4}
\end{figure}
\newpage 

\newpage 
\begin{figure}[H]
	\centering
	\includegraphics[width=.9\linewidth]{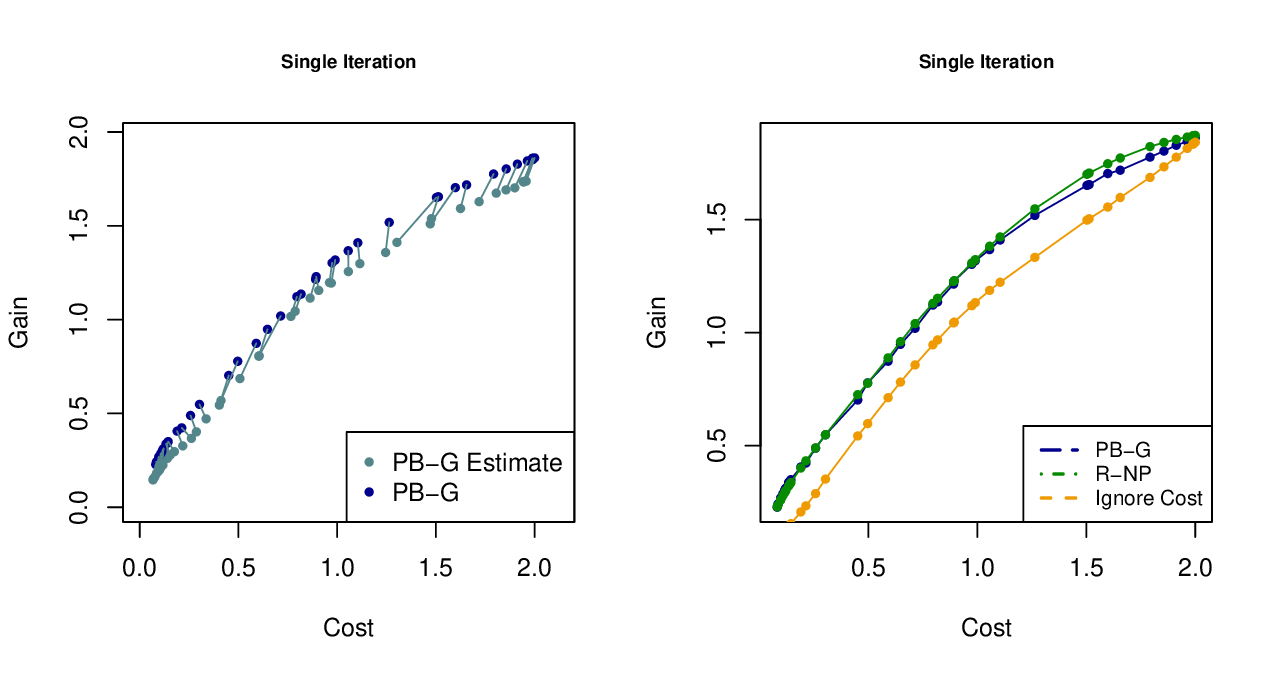}
	\caption[A single training sample iteration utilizing the Gibbs treatment method]{For the DGP with $a=1$, the left-hand side plots the estimated and actual cost-gain pairs (one point for each $u$) for a single training sample iteration for the stochastic treatment assignment Gibbs method. On the right-hand side the actual cost-gain pairs for PB-G for various $u$ values are plotted again, now compared with the actual cost-gain pairs associated with the R-NP and IC rules that produce the same target population cost.  The points on the right are then interpolated to produce cost-gain curve estimates for a single iteration.  These curves are averaged (vertically) over all simulation iterations to produce the right-hand side of Figure \ref{fig:GibbsAndMVCCs}.}
	\label{fig:SingleIteration_two}
\end{figure}

\begin{figure}[H]
	\centering
	\includegraphics[width=0.9\linewidth]{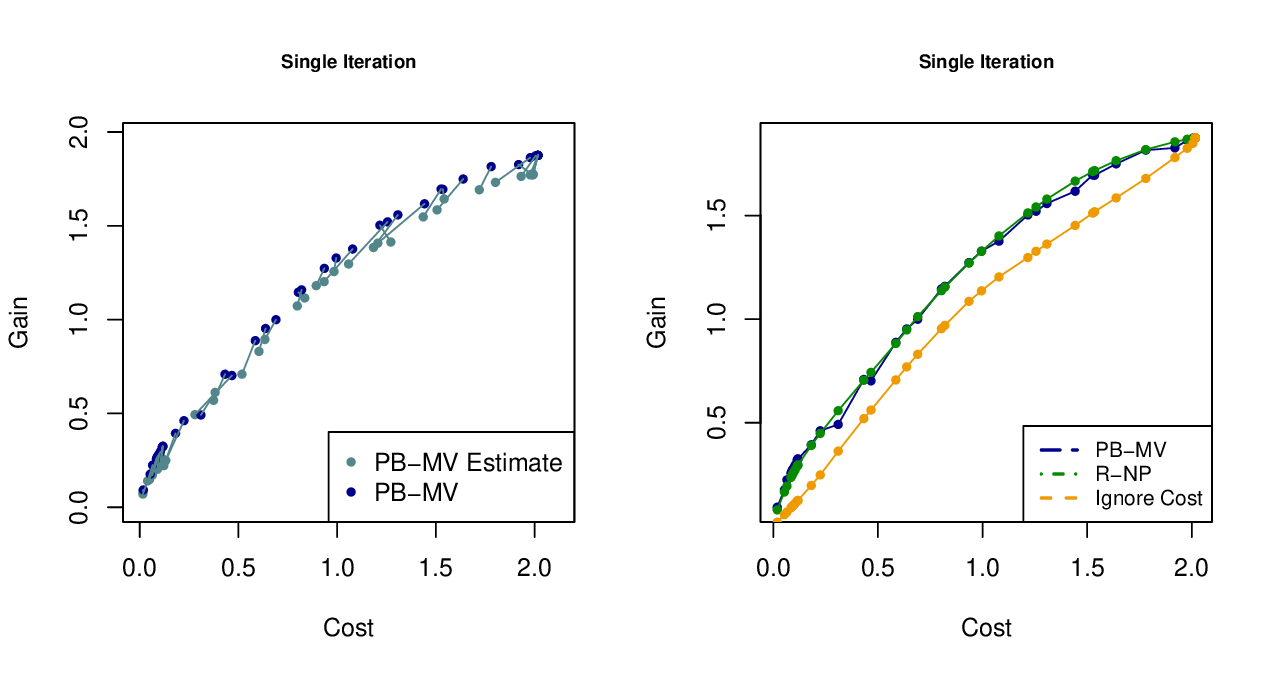}
	\caption[A single training sample iteration utilizing the majority vote treatment method]{Illustrates a single training sample iteration for DGP1 when considering the PB-MV treatment model.}
	\label{fig:SingleIterationMJ_three}
\end{figure}

\newpage

\setcounter{section}{0} 
\renewcommand{\theassumption}{A.\arabic{assumption}}
\renewcommand{\thetheorem}{A.\arabic{theorem}} 
\renewcommand{\theproposition}{A.\arabic{proposition}}
\renewcommand{\thesubsection}{A.\arabic{subsection}}
\renewcommand{\thefigure}{A.\arabic{figure}}
\renewcommand{\thelemma}{A.\arabic{lemma}}
\renewcommand{\thecorollary}{A.\arabic{corollary}}

\section*{Appendix A: Proofs\label{sec: Appendix}}

\subsection{Preliminaries and Adaptations From Earlier Literature to Our Setting \label{Subsec: Appendix Prelinimaries and Adaptations}}
Here we consider preliminary properties to be utilized in subsequent analysis and recall results from the PAC-Bayesian literature that are also needed, sometimes with minor modifications.  For the most part, proofs (and citations) are included for completeness even when a result is a fairly straightforward adaption.

Let $\mathcal{M}\left(  \Theta\right)  $ be the set of measurable functions on
$(\Theta,\mathcal{B}_{\theta})$ and let
\[
\mathcal{M}_{b}^{\pi}\left(  \Theta\right)  =\left\{  A:A\in\mathcal{M}\left(
\Theta\right)  \text{ and }\int_{\Theta}\exp\left(  A(\theta)\right)
d\pi\left(  \theta\right)  <\infty\right\}  ,
\]
which is a subset of $\mathcal{M}\left(  \Theta\right)  $ that has a finite
exponential moment under $\pi.$   We have the following lemma and corollary that will be utilized repeatedly in subsequent analysis.  In particular they serve as a base in deriving Lemma \ref{Lemma constrained KL} in Section \ref{sec: Gibbs Posterior Section}.

\begin{lemma}
	\label{Lemma KL}For $\pi\in\mathcal{P}(\Theta)$ and $A\in\mathcal{M}\left(
	\Theta\right)  $ such that $-A\in\mathcal{M}_{b}^{\pi}\left(  \Theta\right)
	$, let $\rho_{A,\pi}\in\mathcal{P}_{\pi}(\Theta)$ be the probability measure
	on $\Theta$ with the Radon--Nikodym (RN) derivative with respect to $\pi$
	given by
	\[
	\frac{d\rho_{A,\pi}  }{d\pi  }(  \theta )%
	=\frac{\exp\left(  -A\left(  \theta\right)  \right)  }{\int_{\Theta}%
		\exp\left(  -A\left(  \tilde{\theta}\right)  \right)  d\pi\left(
		\tilde{\theta}\right)  }.
	\]
	Then for any probability measure $\rho\in\mathcal{P}_{\pi}\left(
	\Theta\right)  $ we have
	\begin{equation}
	\log\left[  \int_{\Theta}\exp\left(  -A\left(  \theta\right)  \right)
	d\pi\left(  \theta\right)  \right]  =-\left[  \int_{\Theta}A\left(
	\theta\right)  d\rho\left(  \theta\right)  +D_{\mathrm{KL}}\left(  \rho
	,\pi\right)  \right]  +D_{\mathrm{KL}}\left(  \rho,\rho_{A,\pi}\right)
	.\label{KL1}%
	\end{equation}
\end{lemma}

\begin{proof}
	[Proof of Lemma \ref{Lemma KL}]By definition,%
	\begin{align*}
	&  D_{\mathrm{KL}}\left(  \rho,\rho_{A,\pi}\right) \\
	&  =\int_{\Theta}\log\left[  \frac{d\rho}{d\rho_{A,\pi}}(\theta)\right]
	d\rho(\theta)\\
	&  =\int_{\Theta}\log\left\{  \frac{d\rho}{d\pi}(\theta)\left[  \frac
	{d\rho_{A,\pi}}{d\pi}(\theta)\right]  ^{-1}\right\}  d\rho(\theta)\\
	&  =\int_{\Theta}\left[  \log\frac{d\rho}{d\pi}\left(  \theta\right)
	-\log\frac{\exp\left(  -A\left(  \theta\right)  \right)  }{\int_{\Theta}%
		\exp\left(  -A\left(  \tilde{\theta}\right)  \right)  d\pi\left(
		\tilde{\theta}\right)  }\right]  d\rho\left(  \theta\right) \\
	&  =\int_{\Theta}A\left(  \theta\right)  d\rho\left(  \theta\right)
	+\int_{\Theta}\log\left[  \int_{\Theta}\exp\left(  -A\left(  \tilde{\theta
	}\right)  \right)  d\pi\left(  \tilde{\theta}\right)  \right]  d\rho\left(
	\theta\right)  +\int_{\Theta}\left[  \log\frac{d\rho}{d\pi}\left(
	\theta\right)  \right]  d\rho\left(  \theta\right) \\
	&  =\int_{\Theta}A\left(  \theta\right)  d\rho\left(  \theta\right)
	+\log\left[  \int_{\Theta}\exp\left(  -A\left(  \theta\right)  \right)
	d\pi\left(  \theta\right)  \right]  +\int_{\Theta}\left[  \log\frac{d\rho
	}{d\pi}\left(  \theta\right)  \right]  d\rho\left(  \theta\right) \\
	&  =\int_{\Theta}A\left(  \theta\right)  d\rho\left(  \theta\right)
	+\log\left[  \int_{\Theta}\exp\left(  -A\left(  \theta\right)  \right)
	d\pi\left(  \theta\right)  \right]  +D_{\mathrm{KL}}\left(  \rho,\pi\right)  .
	\end{align*}
	Hence,
	\[
	\log\left[  \int_{\Theta}\exp\left(  -A\left(  \theta\right)  \right)
	d\pi\left(  \theta\right)  \right]  =-\left[  \int_{\Theta}A\left(
	\theta\right)  d\rho\left(  \theta\right)  +D_{\mathrm{KL}}\left(  \rho
	,\pi\right)  \right]  +D_{\mathrm{KL}}\left(  \rho,\rho_{A,\pi}\right)  .
	\]	
\end{proof}

\bigskip 

\begin{corollary} 
	\label{Corollary KL}
	
	(a)  Let $\lambda>0$, $\pi\in\mathcal{P}(\Theta)$, and let $A\in \mathcal{M}(\Theta)$ be such that $-\lambda A\in \mathcal{M}^{\pi}_{b}(\Theta)$.  Then  
	\[\rho_{\lambda A, \pi} = \underset{\rho\in\mathcal{P}_{\pi }(\Theta)}{ \arg\min} \left [ \int_{\Theta} A\left ( \theta \right ) d\rho\left ( \theta \right ) +\frac{1}{\lambda} D_{\mathrm{KL}}\left ( \rho, \pi  \right ) \right ],\]
	and
	\[\min_{\rho\in\mathcal{P}_{\pi}(\Theta)}\left [ \int_{\Theta} A\left ( \theta \right ) d\rho\left ( \theta \right ) +\frac{1}{\lambda} D_{\mathrm{KL}}\left ( \rho, \pi \right ) \right ] = -\frac{1}{\lambda}\log \left [ \int_{\Theta} \exp \left ( -\lambda A (\theta) \right ) d\pi  \left ( \theta \right )  \right ]. \]
	
	(b)  For any $\mathcal{A}\left ( \cdot \right ) \in \mathcal{M}_{b}^{\pi}\left ( \Theta \right )$, $\pi\in\mathcal{P}\left ( \Theta \right )$, $\rho\in\mathcal{P}_{\pi}\left ( \Theta \right )$,
	\[
	\int_{\Theta}\mathcal{A}\left(  \theta\right)  d\rho\left(  \theta\right)
	\leq\log\left[  \int_{\Theta}\exp\left(  \mathcal{A}\left(  \theta\right)
	\right)  d\pi\left(  \theta\right)  \right]  +D_{\mathrm{KL}}\left(  \rho
	,\pi\right)  .
	\]
\end{corollary}

\begin{proof}[Proof of Corollary \ref{Corollary KL}]
	Part (a).  Note $\rho_{\lambda A,\pi}=\arg \min_{\rho\in\mathcal{P}_{\pi}\left ( \Theta \right )} D_{\mathrm{KL}} (\rho, \rho_{\lambda A,\pi})$ as $D_{\mathrm{KL}}(\rho, \pi )\geq 0$ with equality if and only if $\rho=\pi$ $\pi$-almost surely.  Replacing $A$ with $\lambda A$ in Lemma \ref{Lemma KL} and noting that the left-hand-side of \eqref{KL1} does not vary with $\rho$ we have
	\begin{align*}
	\rho_{\lambda A, \pi } &= \underset{ \rho\in \mathcal{P}_{\pi}\left(
		\Theta\right)}{\arg \min} \left [  D_{\mathrm{KL}}\left ( \rho, \rho_{\lambda A, \pi} \right ) \right ]
	\\
	&=\underset{ \rho\in \mathcal{P}_{\pi}\left(
		\Theta\right)}{\arg \min} \left [ \int_{\Theta} \lambda A(\theta) d\rho(\theta) + D_{\mathrm{KL}}(\rho,\pi)  \right ] 
	\\
	&= \underset{ \rho\in \mathcal{P}_{\pi}\left(
		\Theta\right)}{\arg \min} \left [ \int_{\Theta} A(\theta) d\rho(\theta) + \frac{1}{\lambda} D_{\mathrm{KL}}(\rho,\pi)  \right ] .
	\end{align*}
	By equation \eqref{KL1} we then have
	\begin{align*}
	\min_{\rho\in\mathcal{P}(\Theta)}\left [ \int_{\Theta} \lambda A\left ( \theta \right ) d\rho \left (\theta \right ) + D_{\mathrm{KL}}\left ( \rho, \pi  \right ) \right ] &=  \int_{\Theta} \lambda A\left ( \theta \right ) d\rho_{\lambda A, \pi }\left ( \theta \right ) + D_{\mathrm{KL}}\left ( \rho_{\lambda A, \pi}, \pi  \right )
	\\
	&=- \log \left [ \int_{\Theta} \exp \left ( -\lambda A \left (\theta \right ) \right ) d\pi \left (\theta \right ) \right ].
	\end{align*}
	This is equivalent to the second statement in part (a).
	
	Part (b) Taking $A=-\mathcal{A}$ in Lemma \ref{Lemma KL}, we obtain that for
	any probability measure $\rho\in\mathcal{P}_{\pi}\left(  \Theta\right)  $,
	\begin{equation}
	\log\left[  \int_{\Theta}\exp\left(  \mathcal{A}\left(  \theta\right)
	\right)  d\pi\left(  \theta\right)  \right]  =\left[  \int_{\Theta}%
	\mathcal{A}\left(  \theta\right)  d\rho\left(  \theta\right)  -D_{\mathrm{KL}%
	}\left(  \rho,\pi\right)  \right]  +D_{\mathrm{KL}}\left(  \rho,\rho_{-A,\pi
	}\right)  .\label{KL sign reversed}%
	\end{equation}
	Note that $D_{\mathrm{KL}}\left(  \rho,\rho_{-A,\pi}\right)  \geq0.$ It
	follows that
	\begin{align*}
	\log\left[  \int_{\Theta}\exp\left(  \mathcal{A}\left(  \theta\right)
	\right)  d\pi\left(  \theta\right)  \right]   &  =\left[  \int_{\Theta
	}\mathcal{A}\left(  \theta\right)  d\rho\left(  \theta\right)  -D_{\mathrm{KL}%
	}\left(  \rho,\pi\right)  \right]  +D_{\mathrm{KL}}\left(  \rho,\rho_{-A,\pi
	}\right)  \\
	&  \geq\left[  \int_{\Theta}\mathcal{A}\left(  \theta\right)  d\rho\left(
	\theta\right)  -D_{\mathrm{KL}}\left(  \rho,\pi\right)  \right]  .
	\end{align*}
	This implies that
	\[
	\int_{\Theta}\mathcal{A}\left(  \theta\right)  d\rho\left(  \theta\right)
	\leq D_{\mathrm{KL}}\left(  \rho,\pi\right)  +\log\left[  \int_{\Theta}%
	\exp\left(  \mathcal{A}\left(  \theta\right)  \right)  d\pi\left(
	\theta\right)  \right].
	\]
\end{proof}

\bigskip

The following Theorem helps to produce PAC-Bayesian generalization bounds in our setting similar to counterparts in the classification literature.  In particular, it essentially the same as Theorem 18 in \cite{germain15aJMLR} with the loss function altered to the structure our setting; it is also similar to Theorem 4.1 in \cite{alquier2016properties}.   The proof follows similar steps to those in \cite{germain15aJMLR} and \cite{alquier2016properties}.  We note that the proof applies to more general sample spaces, not just those following Assumption \ref{Assumption: treatment identification and boundedness}.  We follow the current formulation to avoid additional exposition/notation. 

\begin{theorem}
	\label{Theorem: General PAC-Bayesian Generalization Bound}
	Let Assumptions \ref{Assumption: treatment identification and boundedness} and \ref{Assumption: measurability} (i) hold and let $\pi\in\mathcal{P}(\Theta)$.  Let $\ell(Z, \theta):\mathcal{Z}\times \Theta \rightarrow \mathcal{R}$ denote a measurable loss function with range $\mathcal{R}\subseteq\mathbb{R}$.  Define
	\[L(\theta)=E_{P}\left [ \ell(Z,\theta) \right ], \ \ L_{n}(\theta) = \frac{1}{n}\sum_{i=1}^{n}\ell(Z_{i},\theta),\]
	and, for $\rho\in\mathcal{P}_{\pi }(\Theta)$,
	\[L \left ( f_{G,\rho} \right )=  \int_{\Theta}L(\theta) d\rho(\theta), \ \  L_{n} \left ( f_{G,\rho} \right ) = \int_{\Theta}L_{n}(\theta)d\rho(\theta).\] 
	Let $D: \mathcal{R}\times\mathcal{R}\rightarrow \mathbb{R}$ be any convex function and let $\lambda >0$.  Suppose
	\begin{equation}
	\label{Equation: condition in General PAC-Bayesian Generalization Bound}
	E_{P^{n}} \left [  \int_{\Theta}   \exp\left ( \lambda D\left [ L_{n}(\theta), L(\theta)  \right ] \right )  d\pi(\theta) \right ]  \leq \exp \left (  f(\lambda,n) \right ),
	\end{equation}
	where $f(\lambda,n)<\infty$ and may depend on $\lambda$ and $n$. Then for any $\epsilon\in(0,1]$ it holds with probability at least $1-\epsilon$  that, simultaneously for all $\rho\in\mathcal{P}_{\pi}(\Theta)$,
	\[D\left [ L_{n} \left (  f_{G,\rho} \right ),  L \left ( f_{G,\rho} \right ) \right ]  \leq \frac{f(\lambda, n)+\log\left ( \frac{1}{\epsilon} \right ) +D_{\mathrm{KL}}(\rho,\pi )}{\lambda }.\]
\end{theorem}

\begin{proof}[Proof of Theorem \ref{Theorem: General PAC-Bayesian Generalization Bound}]
	\eqref{Equation: condition in General PAC-Bayesian Generalization Bound} implies that
	\[\int_{\Theta} \exp \left ( \lambda D\left [L_{n}(\theta), L(\theta) \right ] \right ) d\pi(\theta) <\infty,\]
	holds almost surely.  Therefore, applying Corollary \ref{Corollary KL} (b) with $\mathcal{A}(\theta)=\lambda D[L_{n}(\theta),L(\theta)]$, the event  
	\begin{align*}
	\notag 
	&\left \{  \int_{\Theta}\lambda D[L_{n}(\theta),L(\theta)] d\rho(\theta)  \right.
	\\
	\notag 
	&\left. \leq \log \left [ \int_{\Theta} \exp\left ( \lambda D\left [ L_{n}(\theta), L(\theta) \right ] \right ) d\pi(\theta) \right ] + D_{\mathrm{KL}}(\rho,\pi) \ \ \mathrm{for \ all } \ \rho\in\mathcal{P}_{\pi}\left (\Theta\right ) \ \mathrm{simultaneously} \right \},
	\end{align*}
	occurs with probability one.  Applying Jensen's inequality to the object on the left-hand-side of the inequality in this event, we have
	\begin{align}
	\notag 
	&P^{n}  \left \{  \vphantom{\int_{\Theta}}  \lambda D[L_{n}\left (f_{G,\rho} \right ),L\left (f_{G,\rho} \right )]   \right.
	\\
	\notag 
	&\left. \hspace{0.3in} \leq \log \left [ \int_{\Theta} \exp\left ( \lambda D\left [ L_{n}(\theta), L(\theta) \right ] \right ) d\pi(\theta) \right ] + D_{\mathrm{KL}}(\rho,\pi) \ \ \mathrm{for \ all } \ \rho\in\mathcal{P}_{\pi}\left (\Theta\right ) \ \mathrm{simultaneously} \right \},
	\\
	\label{eqn: Gen pac-bayesian theorem probability one event}
	&=1
	\end{align}
	By Markov's inequality and then applying \eqref{Equation: condition in General PAC-Bayesian Generalization Bound},
	\begin{align*}
	&P^{n} \left \{ \int_{\Theta} \exp\left ( \lambda D\left [ L_{n}(\theta), L(\theta) \right ] \right ) d\pi(\theta)  >  \exp \left [ f(\lambda, n) + \log \left ( \frac{1}{\epsilon } \right ) \right ] \right \} 
	\\
	&\leq \frac{E_{P^{n}}\left [\int_{\Theta} \exp\left ( \lambda D\left [ L_{n}(\theta), L(\theta) \right ] \right ) d\pi(\theta) \right ]}{\exp \left [ f(\lambda, n) + \log \left ( \frac{1}{\epsilon } \right ) \right ]}
	\\
	&\leq \epsilon .
	\end{align*}
	Therefore, 
	\[ P^{n} \left \{ \log \left [  \int_{\Theta} \exp\left ( \lambda D\left [ L_{n}(\theta), L(\theta) \right ] \right ) d\pi(\theta) \right ] \leq f(\lambda, n)+ \log\left ( \frac{1}{\epsilon } \right ) \right \} \geq 1-\epsilon \]
	Note that this high probability bound does not involve $\rho$.  Combining it with \eqref{eqn: Gen pac-bayesian theorem probability one event}, we have 
	\begin{align*}
	& P^{n} \left \{  D[L_{n}\left (f_{G,\rho} \right ),L\left (f_{G,\rho} \right )] \leq \frac{ f(\lambda, n) + \log\left (  \frac{1}{\epsilon }\right ) + D_{\mathrm{KL}}(\rho,\pi ) }{\lambda} \  \mathrm{for \ all } \ \rho\in\mathcal{P}_{\pi}\left (\Theta\right ) \ \mathrm{simultaneously} \right \}
	\\
	&\geq 1-\epsilon
	\end{align*}
\end{proof}

\bigskip 

The following lemma will be combined with Theorem \ref{Theorem: General PAC-Bayesian Generalization Bound} to produce Theorem \ref{Theorem: adaptation of Seeger's bound} below.  The lemma yields a key step in adapting PAC-Bayesian bounds from the 0/1-loss classification literature to more general settings, a procedure utilized in \cite{maurer2004note}  and \cite{germain15aJMLR}.   For us it will allow us to follow those author's adaption of a well known PAC-Bayesian bound, appearing, for example, in \cite{seeger2002pac}, to more general settings.  This then serves as a key input for producing Lemma \ref{Lemma: high prob bound for KL(rho hat, rho star)} following the analysis of \cite{lever2010distribution}.

\begin{lemma}
	\label{Lemma: Maurer} Let $X$ be any random variable taking values in $[0,1]$
	with $EX=\mu$. Denote $\mathbf{X}=(X_{1},\dots,X_{n})$ where $X_{1}%
	,\dots,X_{n}$ are iid realizations of $X$. Let $\mathbf{X}^{\prime}%
	=(X_{1}^{\prime},\dots,X_{n}^{\prime})$ where $X_{1}^{\prime},\dots
	,X_{n}^{\prime}$ are iid realizations of a Bernoulli random variable
	$X^{\prime}$ with probability of success $\mu$. If $f:[0,1]^{n}\rightarrow
	\mathbb{R}$ is convex, then
	\[
	E\left[  f\left(  \mathbf{X}\right)  \right]  \leq E\left[  f\left(
	\mathbf{X}^{\prime}\right)  \right]
	\]
	
\end{lemma}

\begin{proof}
	[Proof of Lemma \ref{Lemma: Maurer}]This lemma is due to \cite{maurer2004note}%
	. Another proof with more details is given in \cite{germain15aJMLR}; see Lemmas 51
	and 52 there. For intuition, we can regard $\mathbf{X}^{\prime}$ as a
	mean-preserving spread of $\mathbf{X}$ and $-f$ as the utility function. Then
	the lemma says that $\mathbf{X}$ is preferred by an expected utility maximizer
	having concave utility $-f\left(  \cdot\right)  .$
\end{proof}

\bigskip 

Now we use Lemma \ref{Lemma: Maurer} combined with Theorem \ref{Theorem: General PAC-Bayesian Generalization Bound} to produce Theorem \ref{Theorem: adaptation of Seeger's bound} below, which is a variant of a well known bound appearing in \cite{seeger2002pac}.  To do this, we follow the analysis in \cite{germain15aJMLR} to verify the bound for our setting.  The proof closely follows that in \cite{germain15aJMLR}.  Theorem 20 in \cite{germain15aJMLR}, for example, is a very similar and can apply to a variety of settings.  The only difference here is that the structure of what plays the role of a loss function is stated differently in Theorem \ref{Theorem: General PAC-Bayesian Generalization Bound}. 

The following notation is used in the next theorem.  We let 
\begin{equation}
\label{Definition: small kl()}
\mathrm{kl}(a,b)= a\log\frac{a}{b}+(1-a)\log\frac{1-a}{1-b},
\end{equation}
and adopt the convention that $0\log 0 = 0$, $a\log\frac{a}{0}=\infty$ if $a>0$ and $0\log \frac{0}{0}=0$.  Note that $\mathrm{kl}(a,b)$ is the KL-divergence between two Bernoulli random variables with success probabilities $a$ and $b$.  

\begin{theorem}
	\label{Theorem: adaptation of Seeger's bound}
	Set any prior $\pi \in\mathcal{P}(\theta)$ and $\epsilon\in (0,1]$.  Let Assumption \ref{Assumption: treatment identification and boundedness}, \ref{Assumption: measurability}, and \ref{Assumption: prior indep of data} hold.  Let $\ell(Z, \theta):\mathcal{Z}\times \Theta \rightarrow [0,1]$ denote a measurable loss function with range $[0,1]$ (equipped with the standard Borel sigma field).  Define
	\[L(\theta)=E_{P}\left [ \ell(Z,\theta) \right ], \ \ L_{n}(\theta) = \frac{1}{n}\sum_{i=1}^{n}\ell(Z_{i},\theta),\]
	and, for $\rho\in\mathcal{P}_{\pi }(\Theta)$,
	\[L \left ( f_{G,\rho} \right )=  \int_{\Theta}L(\theta) d\rho(\theta), \ \  L_{n} \left ( f_{G,\rho} \right ) = \int_{\Theta}L_{n}(\theta)d\rho(\theta).\]

	(a).  With probability at least $1-\epsilon$, for all posteriors $\rho\in\mathcal{P}_{\pi}(\Theta)$ simultaneously it holds that
	\[\mathrm{kl}\left ( L_{n}\left ( f_{G,\rho} \right ), L\left ( f_{G,\rho} \right ) \right ) \leq \frac{1}{n}\left [ D_{\mathrm{KL}}(\rho,\pi) + \log\left (2\sqrt{n} \right )+\log \frac{1}{\epsilon} \right ]. \]
	
	\medskip 
	
	(b).  With probability at least $1-\epsilon$, for all posteriors $\rho\in\mathcal{P}_{\pi}(\Theta)$ simultaneously it holds that
	\[ \left (  L_{n}\left ( f_{G,\rho} \right ) - L\left ( f_{G,\rho} \right ) \right )^{2}  \leq \frac{1}{2n}\left [ D_{\mathrm{KL}}(\rho,\pi) + \log\left (2\sqrt{n} \right )+\log \frac{1}{\epsilon} \right ]. \]
\end{theorem}

\begin{proof}
	[Proof of Theorem \ref{Theorem: adaptation of Seeger's bound}]
	Part (a)  Given the adaptation of Theorem \ref{Theorem: General PAC-Bayesian Generalization Bound} to our setting, the proof follows that of Lemma 19 in \cite{germain15aJMLR} or Theorem 1 in \cite{maurer2004note}.  We will apply Theorem \ref{Theorem: General PAC-Bayesian Generalization Bound} with
	\[D\left ( a, b \right ) = \frac{n}{\lambda }\mathrm{kl}\left ( a, b \right ).  \]
	That $\mathrm{kl}(\cdot, \cdot)$ is convex follows from Theorem 2.7.2 of \cite{cover2006elements}.
	We must verify that the condition in \eqref{Equation: condition in General PAC-Bayesian Generalization Bound} holds with $f(\lambda, n)=\log(2\sqrt{n})$.  We will show that for any $\theta\in \Theta$,
	\begin{equation}
	\label{Equation for f(lambda, n ) bound in proof of Seeger adaptation}
	E_{P^{n}} \left \{  \exp \left [ n \mathrm{kl} \left ( L_{n}(\theta), L(\theta) \right ) \right ]  \right \} \leq \sum_{k=0}^{n} {n \choose k} \left ( \frac{k}{n} \right )^{k} \left ( 1-\frac{k}{n} \right )^{n-k}\equiv \xi(n).
	\end{equation}
	It can be shown (c.f. Lemma 19 in \cite{germain15aJMLR} and the references therein) that $\sqrt{n} \leq \xi(n) \leq 2\sqrt{n}$.  Then, by Assumption \ref{Assumption: prior indep of data}, we can reverse the order of integration on the object on the left hand side of condition \ref{Equation: condition in General PAC-Bayesian Generalization Bound}, so that \eqref{Equation for f(lambda, n ) bound in proof of Seeger adaptation} yields that \eqref{Equation: condition in General PAC-Bayesian Generalization Bound} holds with $f(\lambda, n)=\log(2\sqrt{n})$.   All that remains is to prove \eqref{Equation for f(lambda, n ) bound in proof of Seeger adaptation}.   
	
	Let $\theta\in \Theta$.  First note that in edge cases where $L(\theta)=0$ or $L(\theta)=1$, we then have with probability one that $L_{n}(\theta)=0$ or $L_{n}(\theta)=1$, respectively, in which case $\mathrm{kl}(L_{n}(\theta),L(\theta))=0$ and \eqref{Equation for f(lambda, n ) bound in proof of Seeger adaptation} holds.  Now consider any $\theta$ such that $L(\theta)\in (0,1)$.  Note that
	\[\exp \left \{ \lambda D\left ( L_{n}(\theta), L(\theta) \right ) \right \} = \exp \left \{ n \cdot  \mathrm{kl}\left ( \frac{1}{n}\sum_{i=1}^{n}\ell(Z_{i}, \theta) , L(\theta) \right ) \right \} \]
	is a convex function of $\mathbf{X}=(\ell(Z_{1},\theta),\dots,\ell(Z_{n},\theta))$.  Then, by Lemma \ref{Lemma: Maurer},
	\begin{align}
	\label{Maurer equation for KL distance}
	E_{P^{n}} \left \{  \exp \left \{ \lambda D\left ( L_{n}(\theta), L(\theta) \right ) \right \} \right \} \leq E \exp \left \{ n \cdot  \mathrm{kl}\left ( \frac{1}{n}\sum_{i=1}^{n}X'_{i} , L(\theta) \right ) \right \},
	\end{align}
	where $X'_{1},\dots, X'_{n}$ are iid Bernoulli random variables with success probability $L(\theta)$ and the expectation on the right is taken with respect to their joint distribution.  Denoting $X'=\sum_{i=1}^{n}X_{i}'$, we have
	\begin{align}
	\notag 
	E   &\exp \left \{ n \cdot  \mathrm{kl}\left (  \frac{1}{n}X' , L(\theta)\right ) \right \} 
	\\
	\notag 
	& \hspace{0.5in}= E \left ( \frac{\frac{1}{n}X'}{L(\theta)} \right )^{X'} \left ( \frac{1-\frac{1}{n}X'}{1-L(\theta)} \right )^{n-X'} 
	\\
	\notag 
	&\hspace{0.5in}= \sum_{k=0}^{n} \Pr \left ( X'=k \right ) \left ( \frac{\frac{k}{n}}{L(\theta)} \right )^{k} \left ( \frac{1-\frac{k}{n}}{1-L(\theta)} \right )^{n-k}
	\\
	\notag 
	&\hspace{0.5in}= \sum_{k=0}^{n}  {n \choose k} \left (L(\theta)\right )^{k} \left (1-L(\theta) \right )^{n-k} \left ( \frac{\frac{k}{n}}{L(\theta)} \right )^{k} \left ( \frac{1-\frac{k}{n}}{1-L(\theta)} \right )^{n-k}
	\\
	\label{second eqn kl lemma}
	&\hspace{0.5in}= \sum_{k=0}^{n}  {n \choose k} \left (\frac{k}{n} \right )^{k} \left ( 1-\frac{k}{n} \right )^{n-k}=\xi(n)
	\end{align}
	Therefore \eqref{Equation for f(lambda, n ) bound in proof of Seeger adaptation} holds for any $\theta\in\Theta$, completing the proof.
	
	Part (b).  Part (b) follows from part (a) with an application of Pinsker's inequality,
	\begin{equation}
	\label{Pinsker's Inequality}
	2(a-b)^{2} \leq \mathrm{kl}(a,b)
	\end{equation}
\end{proof}

\bigskip 

The following lemma adapts Lemma 2 of \cite{lever2010distribution} to our setting, it will aid in removing a $D_{\mathrm{KL}}$ term from several bounds in Section \ref{sec: PAC Analysis}. 

\begin{lemma}
	\label{Lemma: high prob bound for KL(rho hat, rho star)}
	Let $\hat{\rho}_{\lambda, u}$ and $\rho^{*}_{\lambda, u}$ be as in Definition \ref{Definition: optimal rho hat under a budget constraint} with $\pi\in\mathcal{P}(\Theta)$, $\lambda >0$, and $u\geq 0$.  Let Assumptions \ref{Assumption: treatment identification and boundedness}, \ref{Assumption: measurability}, and \ref{Assumption: prior indep of data} hold and let $\epsilon\in (0,1]$.  With probability at least $1-\epsilon$ it holds that
	\[D_{\mathrm{KL}}\left ( \hat{\rho}_{\lambda, u}, \rho^{*}_{\lambda, u} \right ) \leq \frac{\lambda \sqrt{2} \left ( M_{y} +uM_{c} \right ) }{\kappa \sqrt{n}}\sqrt{  \log \left ( \frac{2\sqrt{n}}{\epsilon} \right )  } + \frac{\lambda^{2}\left ( M_{y}+uM_{c} \right )^{2}}{2 n \kappa^{2} }. \]
\end{lemma} 

\bigskip 
\begin{proof}[Proof of Lemma \ref{Lemma: high prob bound for KL(rho hat, rho star)}]
	The proof follows that of Lemma 2 in \cite{lever2010distribution}, with some minor adjustments, which are straightforward with Theorem \ref{Theorem: adaptation of Seeger's bound} taking the place of Seeger's (c.f. \cite{seeger2002pac})  bound in the setting of \cite{lever2010distribution}.  To lighten the exposition, we will write 
	\[ M(\theta; u)= R(\theta)+uK(\theta) \  \mathrm{and} \  M_{n}(\theta; u)= R_{n}(\theta)+uK_{n}(\theta)\] 
	when writing the RN deriviatives of $\hat{\rho}_{\lambda, u}$ and $\rho^{*}_{\lambda, u}$ with respect to $\pi$ and related objects.  Note that for any $\theta\in \Theta$ we have $M_{n}(\theta;u)\in [-(M_{y}+uM_{c})/2\kappa, (M_{y}+uM_{c})/2\kappa]$ by Assumption \ref{Assumption: treatment identification and boundedness} (iii) and (iv).

	First, observe that
	\begin{align}
	\notag 
	&D_{\mathrm{KL}} \left ( \hat{\rho}_{\lambda, u}, \rho^{*}_{\lambda, u} \right ) 
	\\
	\notag 
	&=  \int_{\Theta}  \log \left [  \left ( \frac{d\hat{\rho}_{\lambda, u}}{d\pi}(\theta) \right ) \left ( \frac{d\pi}{d\rho^{*}_{\lambda, u}}(\theta) \right ) \right ] d\hat{\rho}_{\lambda, u} (\theta)
	\\
	\notag 
	&=\int_{\Theta} \left (  \log \left [  \frac{\exp \left ( -\lambda M_{n}(\theta;u) \right )}{\exp\left ( -\lambda M(\theta; u) \right )} \right ] - \log \left [ \frac{\int_{\Theta} \exp \left (  -\lambda M_{n}(\theta, u) \right )  d\pi(\theta)}{\int_{\Theta} \exp \left (  -\lambda M(\theta; u) \right )  d\pi(\theta)} \right ] \right )  d\hat{\rho}_{\lambda, u }(\theta)
	\\
	\notag 
	&=\int_{\Theta}   \log \left [  \frac{\exp \left ( -\lambda M_{n}(\theta;u) \right )}{\exp\left ( -\lambda M(\theta; u) \right )} \right ]  d\hat{\rho}_{\lambda, u }(\theta) 
	\\
	\notag 
	& \hspace{0.5in} - \log \left [ \frac{\int_{\Theta} \exp \left (  -\lambda \left [ M_{n}(\theta; u) +M(\theta;u)-M(\theta; u)  \right ] \right )  d\pi(\theta)}{\int_{\Theta} \exp \left (  -\lambda M(\theta; u) \right )   d\pi(\theta)} \right ]   
	\\
	\notag 
	&=\lambda \int_{\Theta}  M(\theta; u) - M_{n}(\theta; u) d\hat{\rho}_{\lambda, u}(\theta)  - \log \left [  \int_{\Theta} \exp \left ( \lambda \left [  M(\theta; u )- M_{n}(\theta; u) \right ]  \right ) d\rho^{*}_{\lambda, u} \right ] 
	\\
	\label{Inequality: in Lever Lemma adaptation}
	&\leq  \lambda \left [ \int_{\Theta} M(\theta;u )- M_{n}(\theta; u)d\hat{\rho}_{\lambda, u}(\theta ) -\int_{\Theta} M(\theta; u)- M_{n}(\theta; u) d\rho^{*}_{\lambda, u} \right ],
	\end{align}
	where the last inequality follows from Jensen's inequality.
	
	Next we utilize an Theorem \ref{Theorem: adaptation of Seeger's bound} (b).  For the setting there, let 
	\[\ell(Z,\theta) =  \left (   \ell_{y}(Z,\theta)+ u\ell_{c}(Z,\theta) + \frac{M_{y}+ u M_{c}}{2\kappa}     \right )    \left ( \frac{\kappa }{M_{y}+uM_{c}} \right ) \]
	where
	\begin{align}
	\label{Loss function regret}
	\ell_{y}(Z,\theta) &= \left ( \frac{YD}{e(X)} - \frac{Y(1-D)}{1-e(X)} \right )\left ( f^{*}(X)-  f_{\theta}(X) \right ),
	\\
	\label{Loss function cost}
	\ell_{c}(Z,\theta) &= \left ( \frac{CD}{e(X)} - \frac{C(1-D)}{1-e(X)} \right ) f_{\theta}(X),
	\end{align}
	and $f^{*}$ is as in \eqref{Optimal treatment rule when no budget constraint}.
	
	Note then that, by Assumption \eqref{Assumption: treatment identification and boundedness} (iii) and (iv), for all $\theta\in\Theta$, we have $\ell(Z,\theta)\in [0,1]$ almost surely.  Additionally, we have
	\begin{align*}
	L(\theta) = E_{P}[\ell(Z,\theta)] &= \left (  R(\theta) + uK(\theta) + \frac{M_{y}+uM_{c}}{2\kappa} \right ) \left ( \frac{\kappa}{M_{y}+uM_{c}} \right ) 
	\\
	&= \left ( M(\theta; u) + \frac{M_{y}+uM_{c}}{2\kappa } \right )  \left ( \frac{\kappa}{M_{y}+uM_{c}} \right ) 
	\end{align*}
	and 
	\begin{align*}
	L_{n}(\theta)  &= \left (  R_{n}(\theta) + uK_{n}(\theta) + \frac{M_{y}+uM_{c}}{2\kappa} \right ) \left ( \frac{\kappa}{M_{y}+uM_{c}} \right ) 
	\\
	&= \left ( M_{n}(\theta; u) + \frac{M_{y}+uM_{c}}{2\kappa } \right ) \left ( \frac{\kappa}{M_{y}+uM_{c}} \right ).
	\end{align*}
	  
	Given the above setting, we will apply Theorem \ref{Theorem: adaptation of Seeger's bound} (b).  Note that in Theorem \ref{Theorem: adaptation of Seeger's bound}, the prior $\pi$ does not have to be the same  as that used in the definition of $\hat{\rho}_{\lambda, u}$ and $\rho^{*}_{\lambda, u}$, provided that the posteriors of interest are still absolutely continuous with respect to the prior.  Rather that utilizing the theorem with the $\pi$ associated with $\hat{\rho}_{\lambda, u}$ and $\rho^{*}_{\lambda, u}$, we instead use $\rho^{*}_{\lambda, u}$ as the prior.  Note this prior choice satisfies Assumption \ref{Assumption: prior indep of data}, i.e. it does not depend on the sample.   Applying Theorem \ref{Theorem: adaptation of Seeger's bound} (b) and taking the square root of each side in the high probability bound there, utilizing posteriors $\rho = \hat{\rho}_{\lambda, u}$ and $\rho = \rho^{*}_{\lambda, u}$,  with probability at least $1-\epsilon$ it holds simultaneously that
	\begin{align*}
	\int_{\Theta} L(\theta)-L_{n}(\theta) d\hat{\rho}_{\lambda, u}(\theta) &\leq \frac{1}{\sqrt{2n}} \sqrt{ D_{\mathrm{KL}}\left ( \hat{\rho}_{\lambda, u}, \rho^{*}_{\lambda, u} \right ) + \log \left ( \frac{2\sqrt{n}}{\epsilon} \right )  },
	\\
	- \left ( \int_{\Theta} L(\theta)-L_{n}(\theta) d\rho^{*}_{\lambda, u}(\theta) \right ) &\leq \frac{1}{\sqrt{2n}} \sqrt{  \log \left ( \frac{2\sqrt{n}}{\epsilon} \right )  }.
	\end{align*}
	In terms of $M(\theta; u)$ and $M_{n}(\theta; u)$, this reads: with probability at least $1-\epsilon$ , the following events holds simultaneously
	\begin{align*}
	\int_{\Theta} M(\theta)-M_{n}(\theta) d\hat{\rho}_{\lambda, u}(\theta) &\leq \frac{M_{y}+uM_{c}}{ \kappa\sqrt{2n }} \sqrt{ D_{\mathrm{KL}}\left ( \hat{\rho}_{\lambda, u}, \rho^{*}_{\lambda, u} \right ) + \log \left ( \frac{2\sqrt{n}}{\epsilon} \right )  },
	\\
	- \left ( \int_{\Theta} M(\theta)-M_{n}(\theta) d\rho^{*}_{\lambda, u}(\theta) \right ) &\leq \frac{M_{y}+uM_{c}}{ \kappa\sqrt{2n}} \sqrt{  \log \left ( \frac{2\sqrt{n}}{\epsilon} \right )  }.
	\end{align*}
	Applying the above two inequalities to \eqref{Inequality: in Lever Lemma adaptation}, we obtain
	\begin{align}
	\notag 
	&D_{\mathrm{KL}} \left ( \hat{\rho}_{\lambda, u}, \rho^{*}_{\lambda, u} \right ) 
	\\
	\notag 
	&\leq   \frac{\lambda \left ( M_{y}+uM_{c} \right ) }{ \kappa\sqrt{2n }} \sqrt{ D_{\mathrm{KL}}\left ( \hat{\rho}_{\lambda, u}, \rho^{*}_{\lambda, u} \right ) + \log \left ( \frac{2\sqrt{n}}{\epsilon} \right )  } + \frac{\lambda \left ( M_{y}+uM_{c} \right ) }{ \kappa\sqrt{2n}} \sqrt{  \log \left ( \frac{2\sqrt{n}}{\epsilon} \right )  } 
	\end{align}
	Straightforward algebraic manipulations of the above produce that 
	\small
	\begin{align}
	\notag 
	&\left ( D_{\mathrm{KL}} \left ( \hat{\rho}_{\lambda, u}, \rho^{*}_{\lambda, u} \right ) \right )^{2} - \frac{2\lambda \left ( M_{y} +uM_{c} \right ) }{\kappa \sqrt{2n}}\sqrt{  \log \left ( \frac{2\sqrt{n}}{\epsilon} \right )  } D_{\mathrm{KL}} \left ( \hat{\rho}_{\lambda, u}, \rho^{*}_{\lambda, u} \right ) + \frac{\lambda^{2}\left ( M_{y}+uM_{c} \right )^{2}}{2 n \kappa^{2} } \log\left ( \frac{2\sqrt{n}}{\epsilon} \right ) 
	\\
	\label{Last labeled eqn Lever adaptation}
	&\leq \frac{\lambda^{2}\left ( M_{y}+uM_{c} \right )^{2}}{2 n \kappa^{2} } D_{\mathrm{KL}} \left ( \hat{\rho}_{\lambda, u}, \rho^{*}_{\lambda, u} \right )  +  \frac{\lambda^{2}\left ( M_{y}+uM_{c} \right )^{2}}{2 n \kappa^{2} }\log \left ( \frac{2\sqrt{n}}{\epsilon } \right ).
	\end{align}
	\normalsize
	If 
	\[D_{\mathrm{KL}} \left ( \hat{\rho}_{\lambda, u}, \rho^{*}_{\lambda, u} \right ) \leq \frac{2\lambda \left ( M_{y} +uM_{c} \right ) }{\kappa \sqrt{2n}}\sqrt{  \log \left ( \frac{2\sqrt{n}}{\epsilon} \right )  }, \]
	the statement of the lemma holds.  Otherwise, this and the fact that $D_{\mathrm{KL}}(\hat{\rho}_{\lambda, u}, \rho^{*}_{\lambda, u}) \geq 0$ imply that $D_{\mathrm{KL}}(\hat{\rho}_{\lambda, u}, \rho^{*}_{\lambda, u}) > 0$.  Then, canceling out terms on either side of the inequality in \eqref{Last labeled eqn Lever adaptation} and dividing each side by $D_{\mathrm{KL}}(\hat{\rho}_{\lambda, u}, \rho^{*}_{\lambda, u})$ produces the statement of the lemma.
\end{proof}

\bigskip 

The remainder of the section contains straightforward lemmas that will be utilized in proofs for results in Section \ref{sec: PAC Analysis} and one more substantial result adapted from \cite{FreundEtAl2004} that will conclude this subsection.  

\begin{lemma}
	\label{Lemma: McDiarmid inequality applied to sample cost with non-random probability measure}
	Let Assumptions \ref{Assumption: treatment identification and boundedness} and \ref{Assumption: measurability} hold. Let $\rho'\in\mathcal{P}(\Theta)$ be a (deterministic) probability that does not depend on the sample.  Then 
	\[P^{n} \left ( \int_{\Theta}K_{n}(\theta)d\rho'(\theta) \leq  \int_{\Theta}K(\theta) d\rho'(\theta) + \sqrt{\frac{M_{c}^{2}\log \left (1/\epsilon  \right )}{2n\kappa^{2}} } \right ) \geq 1-\epsilon . \]
\end{lemma}

\begin{proof}[Proof of Lemma \ref{Lemma: McDiarmid inequality applied to sample cost with non-random probability measure}] 
	Define the mapping
	\[K(Z_{1},\dots, Z_{n})=\int_{\Theta}K_{n}(\theta) d\rho'(\theta).\]
	It is straightforward to check that, under Assumption \ref{Assumption: treatment identification and boundedness} (iii),  $K$ satisfies the bounded differences property in Section 6.1 of \cite{boucheron2013concentration} with (in their notation) $c_{i}=M_{c}/(n\kappa )$ for $i=1,\dots, n$.  It follows by McDiarmid's inequality (c.f. \cite{mcdiarmid1989method}) that, for any $t\geq 0$,
	\begin{align*}
	&P^{n}\left ( \int_{\Theta}K_{n}(\theta) d\rho'(\theta)-E_{P^{n}}\left [  \int_{\Theta}K_{n}(\theta) d\rho'(\theta) \right ] > t \right ) 
	\\
	&=P^{n}\left ( \int_{\Theta}K_{n}(\theta) d\rho'(\theta)-  \int_{\Theta}K(\theta) d\rho'(\theta) > t \right )\leq \exp\left \{ -\frac{2 n\kappa^{2} t^{2}}{M_{c}^{2}} \right \}.
	\end{align*}
	Substituting $t = \sqrt{M^{2}_{c}\log(1/\epsilon)/(2 n \kappa^{2})}$, for any $\epsilon\in(0,1]$, this says
	\[P^{n}\left ( \int_{\Theta}K_{n}(\theta) d\rho'(\theta)-  \int_{\Theta}K(\theta) d\rho'(\theta) > \sqrt{\frac{M_{c}^{2}\log\left (1/\epsilon \right )}{2n\kappa^{2}} }  \right )\leq \epsilon . \]
	The result follows by taking the compliment and rearranging terms.
\end{proof}

\begin{lemma}
	\label{Lemma: Normal KL} The KL divergence between $\rho:N(\mu_{\rho}%
	,\Sigma_{\rho})$ and $\pi:N(\mu_{\pi},\Sigma_{\pi})$ on $\mathbb{R}^{q}$, where $\mu_{\theta}$ and $\mu_{\rho}$ are mean vectors and $\Sigma_{\pi}$ and $\Sigma_{\rho}$ are covariance matrices, is
	\[
	D_{\mathrm{KL}}\left(  \rho,\pi\right)  =\frac{1}{2}\left(  \mu_{\rho}%
	-\mu_{\pi}\right)  ^{\prime}\Sigma_{\pi}^{-1}\left(  \mu_{\rho}-\mu_{\pi
	}\right)  +\frac{1}{2}\left[  \mathrm{tr}\left(  \Sigma_{\rho}\Sigma_{\pi
	}^{-1}\right)  -q \right]  -\frac{1}{2}\log\frac{\det\left(  \Sigma_{\rho
		}\right)  }{\det\left(  \Sigma_{\pi}\right)  }.
	\]
\end{lemma}

\begin{proof}
	[Proof of Lemma \ref{Lemma: Normal KL}]By definition and via simple
	calculations, we have
	\begin{align*}
	&  D_{\mathrm{KL}}\left(  \rho,\pi\right) \\
	&  =-\frac{1}{2}E_{\theta\thicksim\rho}\left[  \log\frac{\det\left(
		\Sigma_{\rho}\right)  }{\det\left(  \Sigma_{\pi}\right)  }+\left(  \theta
	-\mu_{\rho}\right)  ^{\prime}\Sigma_{\rho}^{-1}\left(  \theta-\mu_{\rho
	}\right)  -\left(  \theta-\mu_{\pi}\right)  ^{\prime}\Sigma_{\pi}^{-1}\left(
	\theta-\mu_{\pi}\right)  \right] \\
	&  =-\frac{1}{2}\log\frac{\det\left(  \Sigma_{\rho}\right)  }{\det\left(
		\Sigma_{\pi}\right)  }-\frac{1}{2}\left[  q-E_{\theta\thicksim\rho}\left(
	\theta-\mu_{\rho}+\mu_{\rho}-\mu_{\pi}\right)  ^{\prime}\Sigma_{\pi}%
	^{-1}\left(  \theta-\mu_{\rho}+\mu_{\rho}-\mu_{\pi}\right)  \right] \\
	&  =-\frac{1}{2}\log\frac{\det\left(  \Sigma_{\rho}\right)  }{\det\left(
		\Sigma_{\pi}\right)  }-\frac{1}{2}\left[  q-tr\left(  \Sigma_{\rho}\Sigma
	_{\pi}^{-1}\right)  -\left(  \mu_{\rho}-\mu_{\pi}\right)  ^{\prime}\Sigma
	_{\pi}^{-1}\left(  \mu_{\rho}-\mu_{\pi}\right)  \right] \\
	&  =\frac{1}{2}\left(  \mu_{\rho}-\mu_{\pi}\right)  ^{\prime}\Sigma_{\pi}%
	^{-1}\left(  \mu_{\rho}-\mu_{\pi}\right)  +\frac{1}{2}\left[  tr\left(
	\Sigma_{\rho}\Sigma_{\pi}^{-1}\right)  -q \right]  -\frac{1}{2}\log\frac
	{\det\left(  \Sigma_{\rho}\right)  }{\det\left(  \Sigma_{\pi}\right)  }.
	\end{align*}
\end{proof}

The last results needed for our analysis are stated in the two lemmas below.  The first is a more elementary property used in proving the second, which is utilized during a step in the proof of Theorem \ref{Theorem: Main Oracle Inequality Constrained Case} in Section \ref{sec: PAC Analysis}.  Both are close adaptions of analysis in \cite{FreundEtAl2004}.   After a translation of the problem via Corollary \ref{Corollary KL}, we follow the method of proof there, adapting the analysis there in the 0/1 loss setting to ours with fairly straightforward modifications. 

\begin{lemma}
	\label{Lemma: convex mapping used in a proof}
	For $x=(x_{1},\dots, x_{m})\in \mathbb{R}^{m}$, and with $\{a_{i}\}_{i=1}^{m}$ such that $a_{i}\geq 0$ for all $i=1,\dots, m$, the function 
	\[x\mapsto - \log \left [ \sum_{i=1}^{m} a_{i} \exp \left [ x_{i} \right ] \right ] \]
	is concave.  
\end{lemma}

\begin{proof}[Proof of Lemma \ref{Lemma: convex mapping used in a proof}]
	Let $\alpha\in (0,1)$ and $x,y\in\mathbb{R}^{m}$.  We will show that 
	\[K(x) =   \log \left [ \sum_{i=1}^{m} a_{i} \exp \left [ x_{i} \right ] \right ] \]
	is convex.  Let $p=1/\alpha$, $q=1/(1-\alpha)$ and  define $r_{i}=a_{i}^{1/p}\exp[\alpha x_{i}]$ and $s_{i}=a_{i}^{1/q}\exp[(1-\alpha)y_{i}]$.  As $1/p + 1/q=1$, by H\"older's inequality,
	\[\sum_{i=1}^{m}r_{i}s_{i} \leq \left ( \sum_{i=1}^{m}r_{i}^{p} \right )^{1/p} \left ( \sum_{i=1}^{m}s_{i}^{q} \right )^{1/q}.\]
	Taking the logarithm of each side and plugging in the definitions of $p,$ $q$, $r_{i}$ and $s_{i}$, this is equivalent to 
	\[K(\alpha x + (1-\alpha)y)\leq \alpha K(x)+(1-\alpha)K(y),\]
	completing the proof.
\end{proof}

The following lemma combines pieces of Lemmas 1 and 2 of \cite{FreundEtAl2004} and translates those results for the 0/1-loss setting to a useful ingredient for ours.

\begin{lemma}
	\label{Lemma: Upper Bound via Rho star elements}
	Let $\hat{\rho}_{\lambda, u}$ and $\rho^{*}_{\lambda, u}$ be as in Definition \ref{Definition: optimal rho hat under a budget constraint} with $\pi\in\mathcal{P}(\Theta)$, $\lambda>0$, and $u\geq 0$.  Let assumptions \ref{Assumption: treatment identification and boundedness},  \ref{Assumption: measurability}, and \ref{Assumption: prior indep of data} hold.  
	Then,  for any  $\epsilon\in(0,1]$, it holds that 
	\begin{align*}
	P^{n} &\left \{ \int_{\Theta}R_{n}(\theta) d\hat{\rho}_{\lambda, u}(\theta)+u\int_{\Theta}K_{n}(\theta)d\hat{\rho}_{\lambda,u}(\theta) +\frac{1}{\lambda} D_{\mathrm{KL}} \left ( \hat{\rho}_{\lambda,u},\pi  \right ) \leq  \right. 
	\\
	&  \left.  \int_{\Theta}R(\theta)d\rho^{*}_{\lambda,u}(\theta)+u\int_{\Theta}K(\theta)d\rho^{*}_{\lambda,u}(\theta)+\frac{1}{\lambda}D_{\mathrm{KL}}(\rho^{*}_{\lambda,u},\pi )  +\sqrt{\frac{(M_{y}+uM_{c})^{2}\log(1/\epsilon)}{2n\kappa^{2}}}\right \} 
	\\
	& \geq 1-\epsilon .
	\end{align*}
\end{lemma}

\begin{proof}[Proof of Lemma \ref{Lemma: Upper Bound via Rho star elements}]  
	Define the mapping 
	\begin{align*}
	K_{u}(Z_{1},\dots, Z_{n})=\int_{\Theta}R_{n}(\theta) d\hat{\rho}_{\lambda,u}(\theta)+u\int_{\Theta}K_{n}(\theta)d\hat{\rho}_{\lambda,u}(\theta) +\frac{1}{\lambda} D_{\mathrm{KL}} \left ( \hat{\rho}_{\lambda,u},\pi  \right ).
	\end{align*}
	Note that by Corollary \ref{Corollary KL} (a), replacing $A(\theta)$ in the Corollary with $R(\theta)+uK_{n}(\theta)$,
	\begin{equation}
	\label{eqn: 1st eqn in 2nd Lemma for main oracle ineq proof}
	K_{u}(Z_{1},\dots, Z_{n}) = -\frac{1}{\lambda}\log \left [ \int_{\Theta} \exp\left [-\lambda\left (R_{n}(\theta)+u K_{n}(\theta) \right) \right] d\pi(\theta) \right ]. 
	\end{equation}
	First we show that for any $\epsilon\in(0,1]$ it holds that
	\begin{equation}
	\label{eqn: 2nd lemma in main oracle proof, Mcdiarmid result}
	P^{n} \left ( K_{u}(Z_{1},\dots, Z_{n} ) > E_{P^{n}}\left [K_{u}(Z_{1},\dots, Z_{n}) \right ] +\sqrt{\frac{(M_{y}+uM_{c})^{2}\log(1/\epsilon)}{2n\kappa^{2}}}   \right ) \leq \epsilon.
	\end{equation}  
	To show this, for any $i\in\{1,\dots, n\}$, let $Z_{i}'\in \mathcal{Z}$  and let $(Z_{1},\dots, Z_{n})\in\mathcal{Z}^{ n}$.  Let $K_{n}(\theta)$ and $R_{n}(\theta)$ be computed utilizing $(Z_{1},\dots, Z_{i-1},Z_{i}, Z_{i+1},\dots,Z_{n})$ and  let $K_{n}'(\theta)$ and $R_{n}'(\theta)$ be computed as $K_{n}(\theta)$ and $R_{n}(\theta)$ are, respectively, except utilizing the sample $(Z_{1},\dots, Z_{i-1},Z_{i}',Z_{i+1},\dots,Z_{n})$ instead of $(Z_{1},\dots, Z_{i-1},Z_{i}, Z_{i+1},\dots,Z_{n})$.  Also, let $K_{n-i}(\theta)$ and $R_{n-i}$ denote the computation of $K_{n}(\theta)$ and $R_{n}(\theta)$, respectively, except with the sample of size $n-1$ that drops observation $Z_{i}$.  Then by construction $K_{n-i}(\theta)=K_{n-1}'(\theta)$ and $R_{n-i}(\theta)=R_{n-1}'(\theta)$.  Under Assumptions \ref{Assumption: treatment identification and boundedness} (iii) and (iv), 
	\[ -\frac{M_{y} + uM_{c}}{2\kappa } \leq \ell_{y}(Z_{i})+u \ell_{c}(Z_{i},\theta) \leq \frac{M_{y}+uM_{c}}{2\kappa }  \]
	almost surely where $\ell_{y}$ and $\ell_{c}$ are defined in \eqref{Loss function regret} and \eqref{Loss function cost} and are summed over $i$ in $R_{n}(\theta)$ and $K_{n}(\theta)$, respectively.  It follows from \eqref{eqn: 1st eqn in 2nd Lemma for main oracle ineq proof} that,
	\begin{align*}
	&\lvert K_{u}(Z_{1},\dots, Z_{i-1}, Z_{i}, Z_{i+1},\dots , Z_{n})-K_{u}(Z_{1},\dots, Z_{i-1}, Z_{i}', Z_{i+1}, \dots , Z_{n}) \rvert
	\\
	& = \left \lvert -\frac{1}{\lambda}\log \left [ \frac{\int_{\Theta}\exp\left [ -\lambda \left ( R_{n}(\theta)+u K_{n}(\theta) \right )  \right ] d\pi(\theta)}{ \int_{\Theta}\exp\left [ -\lambda \left ( R_{n}'(\theta)+u K_{n}'(\theta) \right )  \right ] d\pi(\theta) } \right ]  \right \rvert 
	\\
	&\leq   -\frac{1}{\lambda}\log \left [ \left (  \frac{\exp\left [ -\lambda(M_{y}+u M_{c})/(2 n \kappa)  \right ] }{\exp\left [ \lambda(M_{y}+u M_{c})/(2 n \kappa)    \right ]} \right ) \left (  \frac{\int_{\Theta}\exp\left [ -\lambda \left ( R_{n-i}(\theta)+u K_{n-i}(\theta) \right )  \right ] d\pi(\theta)}{ \int_{\Theta}\exp\left [ -\lambda \left ( R_{n-i}'(\theta)+u K_{n-i}'(\theta) \right )  \right ] d\pi(\theta) } \right ) \right ]  
	\\
	&= \frac{M_{y}+uM_{c}}{n\kappa },
	\end{align*}
	
	Thus, $K_{u}$ satisfies the bounded differences property in Section 6.1 of \cite{boucheron2013concentration} with (in their notation) $c_{i}=(M_{y}+uM_{c})/(n\kappa)$. By McDiarmid's inequality, (see \cite{mcdiarmid1989method}), it holds that for any $t\geq 0$,
	\[P^{n} \left ( K_{u}(Z_{1},\dots, Z_{n}) -E_{P^{n}} \left [ K_{u}(Z_{1},\dots, Z_{n}) \right ]  > t \right ) \leq \exp \left ( - \frac{2n t^{2}\kappa^{2} }{(M_{y}+uM_{c})^{2}} \right ).\]
	Substituting $t=\sqrt{(M_{y}+u M_{c})^{2}\log(1/ \epsilon )/ (2 n\kappa^{2})}$, we obtain that for for any $\epsilon\in(0,1]$,
	\[P^{n} \left ( K_{u}(Z_{1},\dots, Z_{n} ) > E_{P^{n}}\left [K_{u}(Z_{1},\dots, Z_{n}) \right ] +\sqrt{\frac{(M_{y}+uM_{c})^{2}\log(1/\epsilon)}{2n\kappa^{2}}}  \right ) \leq \epsilon. \]
	Therefore \eqref{eqn: 2nd lemma in main oracle proof, Mcdiarmid result} holds.
	
	Next will show that  
	\begin{equation}
	\label{eqn: 2nd lemma in main oracle proof, Expectation upper bound}
	E_{P^{n}}\left [K_{u}(Z_{1},\dots, Z_{n}) \right ] \leq \int_{\Theta}R(\theta)d\rho^{*}_{\lambda,u}(\theta)+u\int_{\Theta}K(\theta)d\rho^{*}_{\lambda,u}(\theta)+\frac{1}{\lambda}D_{\mathrm{KL}}(\rho^{*}_{\lambda,u},\pi ).
	\end{equation}
	To do so, we follow arguments in Section 7 of \cite{FreundEtAl2004} with adjustments to suit our setting.  
	
	First note that by Corollary \ref{Corollary KL} (a),
	\begin{align}
	\notag 
	\int_{\Theta}R(\theta)d\rho^{*}_{\lambda,u}(\theta)+&u\int_{\Theta}K(\theta)d\rho^{*}_{\lambda,u}(\theta)+\frac{1}{\lambda}D_{\mathrm{KL}}(\rho^{*}_{\lambda,u},\pi )
	\\
	\label{eqn: 2nd eqn in 2nd Lemma for main oracle ineq}
	&=-\frac{1}{\lambda} \log \left [ \int_{\Theta} \exp\left [ -\lambda \left ( R(\theta)+uK(\theta) \right ) \right ]  d\pi(\theta) \right ].
	\end{align}
	
	Next, by Assumption \ref{Assumption: treatment identification and boundedness} and the definitions of $R(\theta)$ and $K(\theta)$, it follows that 
	\[ -M_{y}-uM_{c}\leq R(\theta)+uK(\theta) \leq M_{y}+uM_{c},\] 
	for all $\theta\in \Theta $.    For any $\delta >0$, let
	\[\mathcal{B}_{i}= \left \{\theta \in \Theta: -\left ( M_{y}+uM_{c} \right ) + i\delta  \leq R(\theta)+uK(\theta) <  -\left ( M_{y}+uM_{c} \right )+ (i+1)\delta \right  \},\]
	Then $\mathcal{B}_{0},\dots, \mathcal{B}_{k}$ with $k=\lfloor 2(M_{y}+uM_{c})/\delta  \rfloor $, form a partition of $\Theta$.  For $i\in \{0,\dots, k\}$ such that $\pi(\mathcal{B}_{i})>0$, define 
	\[\tilde{\varepsilon}_{i}\equiv \frac{\int_{\mathcal{B}_{i}}R_{n}(\theta)+uK_{n}(\theta) d\pi(\theta)}{\pi(\mathcal{B}_{i})},\]
	Then, as $\pi$ is independent of the sample by Assumption \ref{Assumption: prior indep of data} and $E_{P^{n}}[R_{n}(\theta)+ uK_{n}(\theta)]=R(\theta)+uK(\theta)$,
	\[E_{P^{n}} \left [ \tilde{\varepsilon}_{i} \right ]= \frac{\int_{\mathcal{B}_{i}} R(\theta)+uK(\theta)d\pi(\theta) }{\pi(\mathcal{B}_{i})} \leq -\left ( M_{y}+uM_{c} \right )+ (i+1)\delta .\]
	Combining this with the fact that $R(\theta)+uK(\theta)>-\left ( M_{y}+uM_{c} \right )+ i\delta$ for $\theta\in\mathcal{B}_{i}$, 
	\begin{align*}
	\int_{\Theta} \exp \left [ -\lambda \left ( R(\theta) +u K(\theta) \right ) \right ] d\pi(\theta) & \leq \sum \pi(\mathcal{B}_{i}) \exp \left [ -\lambda \left ( -\left ( M_{y}+uM_{c} \right )+ i\delta \right ) \right ]
	\\
	&\leq \sum \pi(\mathcal{B}_{i}) \exp \left [ -\lambda \left ( E_{P^{n}}\left [ \tilde{\varepsilon}_{i} - \delta  \right ]  \right ) \right ]
	\\
	&=\exp \left [ \lambda \delta  \right ] \sum \pi(\mathcal{B}_{i}) \exp \left [ -\lambda \left ( E_{P^{n}}\left [ \tilde{\varepsilon}_{i}  \right ]  \right ) \right ],
	\end{align*}
	where the sums above are to be understood as summing over all $i\in\{0,\dots, k \}$ such that $\pi(\mathcal{B}_{i})>0$.  Taking the logarithm of each side of this inequality and multiplying by $-1/\lambda $, we have
	\begin{align}
	\notag 
	&-\frac{1}{\lambda} \log \left [ \int_{\Theta} \exp\left [ -\lambda \left ( R(\theta)+uK(\theta) \right ) \right ]  d\pi(\theta) \right ] 
	\\
	\notag 
	& \geq -\delta -\frac{1}{\lambda}  \log \left [ \sum \pi(\mathcal{B}_{i}) \exp \left [ -\lambda \left ( E_{P^{n}}\left [ \tilde{\varepsilon}_{i}  \right ]  \right ) \right ] \right ] 
	\\
	\label{eqn: lemma first app Jensen}
	& \geq -\delta -\frac{1}{\lambda}E_{P^{n}} \left [ \log \left (  \sum \pi(\mathcal{B}_{i})\exp \left [ -\lambda \tilde{\varepsilon}_{i} \right ] \right ) \right ]
	\\
	\notag 
	&=-\delta -  \frac{1}{\lambda}E_{P^{n}} \left [ \log \left (  \sum \pi(\mathcal{B}_{i})\exp \left [ -\lambda \frac{\int_{\mathcal{B}_{i}}R_{n}(\theta)+uK_{n}(\theta) d\pi(\theta)}{\pi(\mathcal{B}_{i})} \right ] \right ) \right ]
	\\
	\label{eqn: lemma 2nd app Jensen}
	&\geq -\delta - \frac{1}{\lambda }  E_{P^{n}} \left [ \log \left ( \sum \pi(\mathcal{B}_{i}) \frac{\int_{\mathcal{B}_{i}}\exp\left [ -\lambda \left ( R_{n}(\theta)+uK_{n}(\theta) \right ) \right ] d\pi(\theta)}{\pi(\mathcal{B}_{i})} \right ) \right ] 
	\\
	\notag 
	&= -\delta -\frac{1}{\lambda } E_{P^{n}} \left [ \log \left (\int_{\Theta} \exp \left [ -\lambda\left ( R_{n}(\theta)+u K_{n}(\theta )\right ) \right ] d\pi(\theta) \right ) \right ]
	\\
	\label{eqn: form for Ku in lemma proof}
	&=-\delta + E_{P^{n}}\left [ K_{u}(Z_{1},\dots, Z_{n}) \right ]
	\end{align}
	In the above, \eqref{eqn: lemma first app Jensen} follows from an application of Jensen's inequality applied to the concave function 
	\[x\mapsto -\log \left ( \sum_{i}\pi(\mathcal{B}_{i})\exp \left [ x_{i} \right ] \right ),\]
	where the concavity of this function follows from Lemma \ref{Lemma: convex mapping used in a proof}.  \eqref{eqn: lemma 2nd app Jensen} follows from another application of Jensen's inequality now applied to the convex function  $\exp(x)$.  \eqref{eqn: form for Ku in lemma proof} follows from \eqref{eqn: 1st eqn in 2nd Lemma for main oracle ineq proof}.  $\delta$ was arbitrary, so this produces 
	\[-\frac{1}{\lambda} \log \left [ \int_{\Theta} \exp\left [ -\lambda \left ( R(\theta)+uK(\theta) \right ) \right ]  d\pi(\theta) \right ]  \geq E_{P^{n}}\left [ K_{u}(Z_{1},\dots, Z_{n}) \right ], \]
	which, in light of \eqref{eqn: 2nd eqn in 2nd Lemma for main oracle ineq}, shows that \eqref{eqn: 2nd lemma in main oracle proof, Expectation upper bound} holds.  
	\eqref{eqn: 2nd lemma in main oracle proof, Mcdiarmid result} and \eqref{eqn: 2nd lemma in main oracle proof, Expectation upper bound} together yield that
	\small 
	\begin{align*}
	P^{n} &\left ( \vphantom{\int_{\Theta}} K_{u}(Z_{1},\dots, Z_{n}) \right.
	\\
	& \hspace{0.2in} \left. > \int_{\Theta}R(\theta)d\rho^{*}_{\lambda,u}(\theta) + u\int_{\Theta}K(\theta)d\rho^{*}_{\lambda,u}(\theta)+ \frac{1}{\lambda}D_{\mathrm{KL}}(\rho^{*}_{\lambda,u},\pi) + \sqrt{\frac{(M_{y}+uM_{c})^{2}\log(1/\epsilon)}{2n\kappa^{2}}} \right ) \leq \epsilon,
	\end{align*}
	\normalsize
	which produces the statement of the lemma upon taking the compliment. 
\end{proof}

\bigskip 

\subsection{Proofs for Section \ref{sec: Setup}}

\subsubsection*{Proofs for Subsection \ref{sec: stat setting and policy maker's problem}: Statistical Setting and Policy Maker's Problem} 

We will utilize the following lemma in the proof of Theorem \ref{Theorem: theoretically optimal policy choice}.
\begin{lemma}
	\label{Lemma: prior to deriving optimal treatment policy}
	Under the assumptions and setting of Theorem \ref{Theorem: theoretically optimal policy choice}, let $\delta_{c}^{+}(x )=\max(\delta_{c}(x),0)$ and $\delta_{c}^{-}(x)=\max(-\delta_{c}(x),0)$ denote the positive and negative parts of $\delta_{c}(x)$, respectively.  Define 
	\[\beta (b )= E_{Q}[\delta_{c}(X)1\{\delta_{y}(X)> b \delta_{c}(X)\}], \ b\in\mathbb{R},\]
	which is the expected budget of the non-stochastic treatment assigmnet rule $1\{\delta_{y}(x)>b\delta_{c}(x)\}$.
	
	\noindent 
	(i) Let $\eta_{B} = \inf \left \{ b \geq 0:  \beta(b) \leq B  \right \}.$ $\beta(b)$ is non-increasing in $b$ and $0 \leq \eta_{B} < \infty $. 
	
	\medskip
	\noindent 
	(ii) Let 
	\[a_{1} = \begin{cases}
	\frac{ B-\beta \left ( \eta_{B} \right )}{E_{Q}\left [\delta_{c}^{+}(X)1\{\delta_{y}(X)=\eta_{B} \delta_{c}(X)\} \right ]}  & \mathrm{if} \  \beta \left ( \eta_{B} \right  ) < B \ \mathrm{and} \ \eta_{B}>0,
	\\
	0 & \mathrm{else},
	\end{cases} \]
	and
	\[a_{2} = \begin{cases}
	\frac{\beta\left ( \eta_{B} \right ) -B}{E_{Q}\left [\delta_{c}^{-}(X)1\{\delta_{y}(X)=\eta_{B} \delta_{c}(X)\} \right ]}  & \mathrm{if }  \  \beta\left ( \eta_{B} \right ) > B, \ \phantom{\mathrm{and} \ \eta_{B}>0}
	\\
	0 & \mathrm{else}.
	\end{cases}.\]
	Then these are well defined probabilities in that $\beta(\eta_{B})<B$ and $\eta_{B}>0$ implies 
	\[E_{Q}[\delta^{+}_{c}(X)1\{\delta_{y}(X)=\eta_{B} \delta_{c}(X)\}]>0,\] 
	$\beta(\eta_{B})>B$ implies 
	\[E_{Q}[\delta^{-}_{c}(X)1\{\delta_{y}(X)=\eta_{B} \delta_{c}(X)\}]>0,\] 
	and $a_{1},a_{2}\in[0,1]$.  Furthermore, for $f^{*}$ defined as in Theorem \ref{Theorem: theoretically optimal policy choice} with $\eta_{B}, a_{1},$ and $a_{2}$ as above, when  $\beta(0)> B$ it holds that
	\[E_{Q} \left [ \delta_{c}(X)  f^{*}_{B}(x) \right ] = B. \]
\end{lemma}

\bigskip 

\begin{proof}
	[Proof of Lemma \ref{Lemma: prior to deriving optimal treatment policy}]
	
	Proof of (i): To show $\beta(b)$ is non-increasing in $b$, write
	\begin{equation}
	\label{eqn: second eqn in lemma prior to optimal policy rule}
	\beta(b) = E_{Q}\left [ \delta_{c}^{+}(X) 1\{\delta_{y}(X)> b \delta_{c}(X) \} \right ] - E_{Q}\left [ \delta_{c}^{-}(X) 1\{\delta_{y}(X)> b \delta_{c}(X) \} \right ],
	\end{equation}
	By definition of $\delta_{c}^{+}(x)$ and $\delta_{c}^{-}(x)$, $\delta^{+}_{c}(x)1\{\delta_{y}(x)-b\delta_{c}(x)\}$ is non-increasing in $b$ and $\delta^{-}_{c}(x)1\{\delta_{y}(x)-b\delta_{c}(x)\}$ is non-decreasing in $b$ for all $x\in\mathcal{X}$.  It follows that $\beta(b)$ is non-increasing in $b$.  
	
	Checking $0\leq \eta_{B} <\infty $ translates to verifying that our form of policy assignment rule can meet the budget requirement. Let $\{b_{n}\}$ be any non-negative sequence such that $b_{n}\rightarrow\infty$.  Then, $E_{Q}|\delta_{c}(X)|<\infty $ and $ E_{Q}|\delta_{y}(X)|<\infty$, equation \eqref{eqn: second eqn in lemma prior to optimal policy rule}, and an application of the dominated convergence theorem yield
	\begin{align*}
	\lim_{n \rightarrow \infty} \beta(b_{n})&=\lim_{n\rightarrow \infty }E_{Q}\left [ \delta_{c}^{+}(X) 1\{\delta_{y}(X)>b_{n} \delta_{c}(X) \} \right ] - \lim_{n\rightarrow\infty } E_{Q}\left [ \delta_{c}^{-}(X) 1\{\delta_{y}(X)> b_{n} \delta_{c}(X) \} \right ]
	\\
	&= 0 -E_{Q}[\delta_{c}^{-}(X)] < B.
	\end{align*}
	The inequality follows from the assumption that $B>E_{Q}[\delta_{c}(X)1\{\delta_{c}(X)<0\}]=-E_{Q}[\delta_{c}^{-1}(X)]$.  As $\beta(b)$ is non-increasing, we have either $\{b\geq 0: \beta(b)\leq B\}=[r,\infty)$ or $\{b\geq 0: \beta(b)\leq B\}=(r,\infty)$ for some $r\in\mathbb{R}_{\geq 0}=\{x\in\mathbb{R}: x\geq 0\}$.  It follows that $0\leq \eta_{B} <\infty$.
	
	\medskip 
	Proof of (ii): Let $b\in\mathbb{R}$.  For any sequence $b_{n}\uparrow b$, by the dominated convergence theorem we have
	\begin{align*}
	\lim_{b_{n}\uparrow b}\beta(b_{n}) = & \lim_{b_{n}\uparrow b}E_{Q}[\delta_{c}^{+}(X)1\{\delta_{y}(X)>b_{n}\delta_{c}(X)\}] -\lim_{b_{n}\uparrow b}E_{Q}[\delta_{c}^{-}(X)1\{\delta_{y}(X)>b_{n}\delta_{c}(X)\}]
	\\
	=& E_{Q}[\delta_{c}^{+}(X)1\{\delta_{y}(X) \geq b\delta_{c}(X)\}] - E_{Q}[\delta_{c}^{-}(X)1\{\delta_{y}(X)>b\delta_{c}(X)\}].
	\\
	= &E_{Q}[\delta_{c}^{+}(X)1\{\delta_{y}(X) > b\delta_{c}(X)\}] + E_{Q}[\delta_{c}^{+}(X)1\{\delta_{y}(X) = b\delta_{c}(X)\}] 
	\\
	&-E_{Q}[\delta_{c}^{-}(X)1\{\delta_{y}(X)>b\delta_{c}(X)\}].
	\\
	=& \beta(b)+ E_{Q}[\delta_{c}^{+}(X)1\{\delta_{y}(X)=b\delta_{c}(X)\}]
	\end{align*}
	This yields
	\begin{equation}
	\label{eqn: lim from left optimal sol lemma}
	\lim_{x\rightarrow b^{-}}\beta(x) = \beta(b)+ E_{Q} \left [\delta_{c}^{+} \left ( X \right )1 \left \{\delta_{y}(X)=b\delta_{c}(X) \right \} \right ].
	\end{equation}
	Similar steps now starting with any sequence $b_{n}\downarrow b$ produce that
	\begin{equation}
	\label{eqn: lim from right optimal sol lemma}
	\lim_{x\rightarrow b^{+}}\beta(x) = \beta(b)- E_{Q} \left [ \delta_{c}^{-}\left ( X \right ) 1 \left \{\delta_{y}(X)=b\delta_{c}(X) \right \} \right ].
	\end{equation}
	As $\beta(\cdot)$ is non-increasing, it has at most countably many discontinuities, which occur at values $b$ for which either $E_{Q}[\delta_{c}^{+}(X)1\{\delta_{y}(X)=b\delta_{c}(X)\}]>0$ or $E_{Q}[\delta_{c}^{-}(X)1\{\delta_{y}(X)=b\delta_{c}(X)\}]>0$ or both.  
	
	Now, if $B > \beta(\eta_{B})$ and $\eta_{B}>0$, by definition of $\eta_{B}$ we have that $\beta(\eta')>B$ for any $\eta' < \eta_{B}$.   Combined with \eqref{eqn: lim from left optimal sol lemma}, we obtain
	\[\beta \left ( \eta_{B} \right ) < B \leq \beta \left ( \eta_{B} \right ) + E_{Q}\left [ \delta_{c}^{+}1\{\delta_{y}(X)=\eta_{B} \delta_{c}(X)\} \right ], \]
	which implies that  $E_{Q}[\delta_{c}^{+}(X)1\{\delta_{y}(X)=\eta_{B} \delta_{c}(X)\}]>0$ and $a_{1}\in[0,1]$.
	
	Next, if $B < \beta(\eta_{B})$, by defnition of $\eta_{B}$ we have $\beta(\eta')\leq B$ for any $\eta'>\eta_{B}$.  Combining this  with  \eqref{eqn: lim from right optimal sol lemma}, we obtain
	\[ \beta(\eta_{B}) - E_{Q} \left [ \delta_{c}^{-}\left ( X \right ) 1 \left \{\delta_{y}(X)=\eta_{B} \delta_{c}(X) \right \} \right ] \leq  B < \beta \left ( \eta_{B} \right ). \]
	This implies  $E_{Q}[\delta_{c}^{-}(X)1\{\delta_{y}(X)=\eta_{B} \delta_{c}(X)\}]>0$ and $a_{2}\in[0,1]$.
	
	For the last claim of (ii), write
	\begin{align}
	\notag 
	E_{Q}\left [ \delta_{c}(X)f^{*}_{B}(x) \right ] = & E_{Q} \left [ \delta_{c}(X) 1 \left \{ \delta_{y}(X) > \eta_{B} \delta_{c}(X) \right \}  \right ]
	\\
	\notag 
	&+ a_{1} E_{Q} \left [ \delta_{c}(X) 1\left \{ \delta_{y}(X)=\eta_{B} \delta_{c}(X) \right \} 1\left \{ \delta_{c}(X)>0 \right \} \right ] 
	\\
	\label{eqn: optimal cost no slack}
	&+a_{2} E_{Q} \left [ \delta_{c}(X) 1\left \{ \delta_{y}(X)=\eta_{B} \delta_{c}(X) \right \} 1\left \{ \delta_{c}(X) < 0 \right \} \right ] .
	\end{align}
	When $\beta(0)>B$, there are $3$ scenarios for $\beta(\eta_{B})$: (i) $\beta(\eta_{B})=B$ and $\eta_{B} > 0$; (ii) $\beta(\eta_{B})<B$ and $\eta_{B}>0$; or (iii) $\beta(\eta_{B})>B$ and $\eta_{B}\geq 0$.  For scenario (i), we have $a_{1}=a_{2}=0$ and the result holds as $E_{Q}[\delta_{c}(X)1\{\delta_{y}(X)>\eta_{B} \delta_{c}(X)\}]=\beta(\eta_{B})=B $.  For scenario (ii), $a_{2}=0$ and \eqref{eqn: optimal cost no slack} becomes
	\begin{align*}
	E_{Q}\left [ \delta_{c}(X)f^{*}_{B}(x) \right ] = & E_{Q} \left [ \delta_{c}(X) 1 \left \{ \delta_{y}(X) > \eta_{B} \delta_{c}(X) \right \}  \right ]
	\\
	\notag 
	&+ \frac{B-\beta(\eta_{B})}{E_{Q}\left [ \delta_{c}^{+}( X) 1\left \{ \delta_{y}(X)=\eta_{B} \delta_{c}(X) \right \} \right ] } E_{Q} \left [ \delta_{c}^{+}(X) 1\left \{ \delta_{y}(X)=\eta_{B} \delta_{c}(X) \right \} \right ] = B.
	\end{align*}
	For scenario (iii), $a_{1}=0$ and \eqref{eqn: optimal cost no slack} becomes
	\begin{align*}
	E_{Q}\left [ \delta_{c}(X)f^{*}_{B}(x) \right ] = & E_{Q} \left [ \delta_{c}(X) 1 \left \{ \delta_{y}(X) > \eta_{B} \delta_{c}(X) \right \}  \right ]
	\\
	\notag 
	&-\frac{\beta(\eta_{B})-B}{E_{Q}\left [ \delta_{c}^{-}( X) 1\left \{ \delta_{y}(X)=\eta_{B} \delta_{c}(X) \right \} \right ] } E_{Q}\left [ \delta_{c}^{-}(X) 1\left \{ \delta_{y}(X)=\eta_{B} \delta_{c}(X) \right \} \right ] = B.
	\end{align*}
	This completes the proof of (ii).
\end{proof}

\bigskip 

\begin{proof}
	[Proof of Theorem \ref{Theorem: theoretically optimal policy choice}]
	
	The existence of $\eta_{B}\geq 0$, $a_{1},a_{2}\in [0,1]$ such that either $\eta_{B}=a_{1}=a_{2}=0$ (then $f^{*}$ simplifies to $f^{*}$) when $K(f^{*})\leq B$ or else $(\eta_{B},a_{1},a_{2})$ are such that $K(f^{*})=B$ when $K(f^{*})>B$ follows from Lemma \ref{Lemma: prior to deriving optimal treatment policy}.  To see this note $\beta(0)=K(f^{*})$, where $\beta(\cdot)$ is defined in Lemma \ref{Lemma: prior to deriving optimal treatment policy}. Thus, the statement about the budget being used entirely when $K(f^{*})>B$ is stated directly in Lemma \ref{Lemma: prior to deriving optimal treatment policy}.  When $\beta(0)=K(f^{*})\leq B$, $\eta_{B}$ as defined in Lemma \ref{Lemma: prior to deriving optimal treatment policy} is equal to zero and then both $a_{1}=a_{2}=0$ also from their definitions there. 
	
	Next we need to verify that $f^{*}$ satisfies \eqref{eqn: policy makers budget problem}, i.e. is an optimal budget-constrained treatment policy.  Let $r:\mathcal{X}\rightarrow [0,1]$ denote any other stochastic treatment assignment rule that satisfies the budget constraint $K(r)\leq B$.  As in \cite{sun2021treatment}, we proceed by verifying that 
	\[E_{Q}\left [ \delta_{y}(X)f^{*}_{B}(x) \right ] \geq E_{Q}\left [ \delta_{y}(X) r(X) \right ]. \]
	
	By the definition of $f^{*}$, when $\delta_{y}(x)>\eta_{B} \delta_{c}(x)$ we also have $f^{*}_{B}(x)-r(x) \geq 0$.  Hence $\delta_{y}(x) ( f^{*}_{B}(x)-r(x) ) \geq \eta_{B} \delta(x) ( f^{*}_{B}(x)-r(x)  )$ in this case.  When $\delta_{y}(x)< \eta_{B} \delta_{c}(x)$, we have $f^{*}_{B}(x)-r(x) \leq  0$ and hence $\delta_{y}(x) ( f^{*}_{B}(x)-r(x) ) \geq \eta_{B} \delta(x) ( f^{*}_{B}(x)-r(x)  )$ in this case as well.  It follows that
	\begin{align}
	\label{first eqn optimal treatment rule theorem}
	E_{Q} \left [ \delta_{y}(X) \left ( f^{*}_{B}(x)- r(X) \right ) \right ] \geq &  \eta_{B} E_{Q} \left [ \delta_{c}(X) \left ( f^{*}_{B}(x)-r(X) \right ) \right ].
	\end{align}
	There are two possible scenarios: $K(f^{*})\leq B$ or else $K(f^{*})>B$.  When $K(f^{*})\leq B$, we have $\eta_{B}=0$ and hence the right-hand-side of \eqref{first eqn optimal treatment rule theorem} is zero implying $f^{*}$ is optimal.  If, alternatively, $K(f^{*})>B$, then we know that $K(f^{*})=B$ and $K(r)\leq B$.  Thus
	\[E_{Q}[\delta_{c}(X)(f^{*}_{B}(x)-r(X))]=K(f^{*})-K(r)\geq 0.\]
	Now the right-hand-side of \eqref{first eqn optimal treatment rule theorem} is non-negative (as $\eta_{B}\geq 0$) and  $f^{*}$ is again optimal.
	
	Lastly we need to show that if 
	\begin{equation}
	\label{condition for uniqueness in optimal policy proof}
	E_{Q}[1\{\delta_{y}(X)=\eta_{B} \delta_{c}(X)\}]=0,
	\end{equation}
	then $f^{*}$ is deterministic and unique (in an almost sure sense).  It is clear from the form of $f^{*}$ that it is almost surely equivalent to $1\{\delta_{y}(x)>\eta_{B} \delta_{c}(x)\}$ in this setting.  Additionally, with the choices of $\eta_{B},a_{1}, a_{2}$ given in Lemma \ref{Lemma: prior to deriving optimal treatment policy} this will be true for all $x\in\mathcal{X}$.   To see that this follows from the proof of Lemma \ref{Lemma: prior to deriving optimal treatment policy}, by \eqref{eqn: lim from left optimal sol lemma} and \eqref{eqn: lim from right optimal sol lemma} there, $\beta(b)$ is continuous in this scenario so that $\beta(\eta_{B})=B$ when $\eta_{B}>0$; this implies $a_{1}=a_{2}=0$ when $\eta_{B}>0$.  When $\eta_{B}=0$, we must have $\beta(0)\leq B$ and then again $a_{1}=a_{2}=0$ as defined in Lemma \ref{Lemma: prior to deriving optimal treatment policy}.  
	
	To check uniqueness, let $r(x)$ be any other treatment assignment rule that satisfies the budget constraint $K(r)\leq B$ and is not a.s. equal to $f^{*}_{B}(x)$.  When $\eta_{B}=0$, $r$ must then assign treatment for a subset of $\mathcal{X}$ with positive probability that has negative CATE or else fail to assign treatment to some subset of $\mathcal{X}$ that has positive CATE with positive probability (or both).  This results in lower expected welfare than $f^{*}$, so $r$ cannot be optimal.  When $\eta_{B}>0$, the argument is similar to that showing $f^{*}$ is optimal.  When $\eta_{B}>0$, its definition in Lemma \ref{Lemma: prior to deriving optimal treatment policy} indicates that $K(f^{*})=\beta(0)>B$ (and from \eqref{eqn: optimal cost no slack} in Lemma \ref{Lemma: prior to deriving optimal treatment policy}, it follows that $\eta_{B}>0$ in this case must be the unique choice for which the expected budget of $f^{*}$ is $B$).  
	$P(f^{*}_{B}(x)\neq r(x))>0$ then implies that for some subset of $\mathcal{X}$ with positive probability we must have $f^{*}_{B}(x)-r(x)>0$ when $\delta_{y}(x)>\eta_{B} \delta_{c}(x)$ or else $f^{*}_{B}(x)-r(x)<0$ when $\delta_{y}(x)< \eta_{B} \delta_{c}(x)$ (or both).  This implies that the inequality in \eqref{first eqn optimal treatment rule theorem} is strict.  As $f^{*}$ uses up the entire budget \eqref{first eqn optimal treatment rule theorem} now implies the left-hand-side is strictly positive, which concludes the proof.
\end{proof}

\subsubsection*{Proofs for Subsection \ref{sec: Gibbs Posterior Section}: Initial Properties of the Gibbs Posterior}

\begin{proof}[Proof of Lemma \ref{Lemma constrained KL}]
	First we derive the result in \eqref{Lemma constrained KL: minimization problem}.  There are two possible scenarios.  First, if $\Lambda(0)\leq B$, i.e. the ``cost" at $u=0$ is within budget, then $\tilde{\rho}_{A,H,\lambda, 0}\in \mathcal{E}_{B}$ and $\tilde{\rho}_{A,H,\lambda, 0} = \rho_{\lambda A,\pi}$ in the notation of Corollary \ref{Corollary KL}.  Then the result follows from Corollary \ref{Corollary KL} (a).  Note that this scenario captures the case when $B=\infty$, i.e. when there is no budget constraint.
	
	In the second scenario, $\Lambda(0)>B$ (and $B<\infty $).  Assume this is case for the remainder of the proof of property \eqref{Lemma constrained KL: minimization problem}.  First, we will show that this implies $\Lambda(u)$ is (strictly) decreasing in $u$ and that there exists a unique $\overline{u}_{B}>0$ such that $\Lambda(\overline{u}_{B})=B$.  Note below that at any point $u\geq 0$, because the derivatives of the integrands are dominated by integrable functions on intervals of the form $(u-a,u+b)$, some $a,b>0$, and as $\Lambda(u)$ is easily extended in definition to negative values of $u$ in neighborhoods of $0$, we can exchange differentiation and integration.  We have
	\begin{align}
	\notag 
	&\frac{d}{d u} \Lambda (u) 
	\\
	\notag 
	&=\frac{d}{du} \left [ \left (  \int_{\Theta} H \left ( \theta \right )  \exp\left [ -\lambda \left (  A\left ( \theta \right )+ u H\left ( \theta\right)  \right ) \right ] d\pi\left (\theta \right ) \right ) \left ( \int_{\Theta} \exp\left [ -\lambda \left (  A\left ( \theta \right )+ u H\left ( \theta\right)  \right ) \right ] d\pi\left ( \theta\right ) \right )^{-1} \right ]
	\\
	\notag
	& = -\lambda \int_{\Theta} H^{2}\left ( \theta \right ) d\tilde{\rho}_{A,H,\lambda, u}\left ( \theta \right ) + \lambda \left (\int_{\Theta} H\left (\theta\right ) d\tilde{\rho}_{A,H,\lambda, u}\left(\theta\right )  \right )^{2}
	\\
	\notag
	&= -\lambda \mathbb{V}_{\theta\sim \tilde{\rho}_{A,H,\lambda, u}} \left [ H\left ( \theta \right ) \right ]
	\\
	\label{Lemma Constrained KL Proof: d/du V(u) < 0}
	&<0,
	\end{align}
	where $\mathbb{V}_{\theta\sim \tilde{\rho}_{A,H,\lambda, u}}[H(\theta)]$ denotes the variance of $H(\theta)$ when $\theta\sim \tilde{\rho}_{A,H,\lambda, u}$.  Note the strict inequality of the last line holds because the distribution of $H(\theta)$ induced by $\tilde{\rho}_{A,H,\lambda, u }$ is degenerate only when the distribution of $H(\theta)$ induced by $\pi$ is degenerate.  If this were the case, \eqref{eqn: condition for constrained KL solution to hold} would imply that $\Lambda(0)< B$.  Hence the strict inequality when $\Lambda(0)\geq B$, which includes our current $\Lambda(0)>B$ scenario.   
	
	Now, note that \eqref{eqn: condition for constrained KL solution to hold} implies there exist $\epsilon_{1},\eta >0$ such that $\pi(\{ \theta: H(\theta) \leq B-\epsilon_{1} \})=\eta >0$.  Let $\epsilon_{2}$ be such that $0 < \epsilon_{2} < \epsilon_{1}$.  Letting $M_{h}$ and $M_{a}$ be such that $|H(\theta)|\leq M_{h}$ and $|A(\theta)|\leq M_{a}$ for all $\theta$ (as these functions are assumed bounded), we have
	\begin{align*}
	&\int_{\Theta} H \left ( \theta \right ) d\tilde{\rho}_{A,H,\lambda, u}\left ( \theta \right )  
	\\
	&\leq (B-\epsilon_{2}) + M_{h} \int_{\Theta} 1 \left \{H\left ( \theta \right )> B-\epsilon_{2} \right \}  d\tilde{\rho}_{A,H,\lambda, u}\left ( \theta \right ) 
	\\
	&= (B-\epsilon_{2}) + M_{h} \frac{\int_{\Theta} 1 \left \{H\left ( \theta \right )> B-\epsilon_{2} \right \} \exp\left [ -\lambda \left ( A\left ( \theta \right ) +u H\left ( \theta \right ) \right ) \right ] d\pi\left ( \theta \right ) }{\int_{\Theta} \left (1 \left \{ H\left ( \theta \right )\leq  B-\epsilon_{1} \right \} + 1 \left \{ H\left ( \theta \right )> B-\epsilon_{1} \right \} \right ) \exp\left [ -\lambda \left ( A\left ( \theta \right ) +u H\left ( \theta \right ) \right ) \right ] d\pi\left ( \theta \right )}
	\\
	&\leq (B-\epsilon_{2}) + M_{h} \frac{\int_{\Theta} 1 \left \{H\left ( \theta \right )> B-\epsilon_{2} \right \} \exp\left [ -\lambda \left ( A\left ( \theta \right ) +u H\left ( \theta \right ) \right ) \right ] d\pi\left ( \theta \right ) }{\int_{\Theta} 1 \left \{ H\left ( \theta \right )\leq  B-\epsilon_{1} \right \}  \exp\left [ -\lambda \left ( A\left ( \theta \right ) +u H\left ( \theta \right ) \right ) \right ] d\pi\left ( \theta \right )}
	\\
	&\leq (B-\epsilon_{2}) + M_{h} \left ( \frac{\exp \left [ -\lambda u (B-\epsilon_{2}) \right ] }{\exp \left [ -\lambda u (B-\epsilon_{1}) \right ]} \right ) \left ( \frac{\exp \left [ \lambda M_{a} \right ]}{\eta \exp \left [ -\lambda M_{a} \right ]} \right )
	\\
	&= (B-\epsilon_{2}) + \exp\left [ -\lambda u \left ( \epsilon_{1}-\epsilon_{2} \right ) \right ] \left ( \frac{M_{h}\exp\left [ 2\lambda M_{a} \right ]}{\eta} \right ).
	\end{align*}
	As $\epsilon_{1}-\epsilon_{2}>0$, for large enough values of $u$ it holds that $\Lambda(u)<B$.  Then, as  $\Lambda (u)$ is continuous and strictly decreasing in $u$ it follows that there is a unique $\overline{u}_{B}>0$ such that $\Lambda (\overline{u}_{B})=B$.
	
	To finish the proof the property in \eqref{Lemma constrained KL: minimization problem}, we need to show that when $\Lambda(0)>B$, $\tilde{\rho}_{A,H,\lambda, \overline{u}_{B}}$ is the optimal probability measure on $\Theta$ for the minimization problem.  Replacing $A$ in Corollary \ref{Corollary KL} (a) with the $A+\overline{u}_{B}H$ as given above and noting that $\tilde{\rho}_{A,H,\lambda, \overline{u}_{B}}= \rho_{\lambda(A+\overline{u}_{B}H), \pi}$, we have that for any $\rho\in\mathcal{E}_{B}$,
	
	\begin{align}
	\notag 
	&\tilde{\rho}_{A,H,\lambda, \overline{u}_{B}} 
	\\
	&= \underset{\rho\in\mathcal{P}_{\pi}(\Theta)}{\arg\min } \left [ \int_{\Theta} \left \{  A\left ( \theta \right ) + \overline{u}_{B}H\left ( \theta \right ) \right \} d\rho\left ( \theta \right ) +\frac{1}{\lambda}D_{\mathrm{KL}}\left (\rho,\pi\right ) \right ]
	\\
	\notag 
	&=\underset{\rho\in\mathcal{P}_{\pi}(\Theta)}{\arg\min } \left [ \int_{\Theta} \left \{  A\left ( \theta \right ) + \overline{u}_{B}H\left ( \theta \right ) \right \} d\rho\left ( \theta \right ) +\frac{1}{\lambda}D_{\mathrm{KL}}\left (\rho,\pi\right ) -\overline{u}_{B}B \right ]
	\\
	\label{eqn: rho_tilde proof}
	& = \underset{\mathcal{E}_{H,B}}{\arg\min } \left [ \int_{\Theta}   A\left ( \theta \right ) d\rho\left ( \theta \right ) +\frac{1}{\lambda}D_{\mathrm{KL}}\left (\rho,\pi\right ) +\overline{u}_{B}\left (  \int_{\Theta} H\left ( \theta \right ) d\rho\left ( \theta \right ) -B \right ) \right ]
	\\
	\label{eqn: rho_tilde proof2}
	&= \underset{ \{\rho\in\mathcal{P}_{\pi}(\Theta): \int_{\Theta}H(\theta)d\rho(\theta)=B \} }{\arg\min } \left [ \int_{\Theta}   A\left ( \theta \right ) d\rho\left ( \theta \right ) +\frac{1}{\lambda}D_{\mathrm{KL}}\left (\rho,\pi\right ) +\overline{u}_{B}\left (  \int_{\Theta} H\left ( \theta \right ) d\rho\left ( \theta \right ) -B \right ) \right ],
	\end{align}
	where the third equality holds in our specific setting because $\tilde{\rho}_{A,H,\lambda, \overline{u}_{B}}\in\mathcal{E}_{H,B}$ and $\mathcal{E}_{H,B}\subset\mathcal{P}_{\pi}(\Theta)$. The fourth equality follows similar reasoning.  Next note that for any $\rho\in \mathcal{E}_{H,B}$, as $\overline{u}_{B}>0$,  
	
	\begin{align}
	\label{eqn: rho_tilde proof3}
	\int_{\Theta}  A\left ( \theta \right )  d\rho\left ( \theta \right ) +\frac{1}{\lambda}D_{\mathrm{KL}}\left (\rho,\pi\right ) &\geq  \int_{\Theta}  A\left ( \theta \right )  d\rho\left ( \theta \right ) +\frac{1}{\lambda}D_{\mathrm{KL}}\left (\rho,\pi\right ) + \overline{u}_{B}\left ( \int_{\Theta}H\left ( \theta \right ) d\rho\left (\theta \right )-B \right )
	\\
	\label{eqn: rho_tilde proof4}
	&\geq \int_{\Theta}  A\left ( \theta \right )  d\tilde{\rho}_{A,H,\lambda, \overline{u}_{B}}\left ( \theta \right ) +\frac{1}{\lambda}D_{\mathrm{KL}}\left (\tilde{\rho}_{A,H,\lambda, \overline{u}_{B}},\pi\right ),
	\end{align}
	where \eqref{eqn: rho_tilde proof4} follows from \eqref{eqn: rho_tilde proof} and the fact that $\int_{\Theta} H(\theta) d\tilde{\rho}_{A,H,\lambda, \overline{u}_{B}}(\theta)=B$.  Because the inequality in \eqref{eqn: rho_tilde proof3} is strict whenever $\int_{\Theta}H(\theta)d\rho < B$ it follows from \eqref{eqn: rho_tilde proof2} that $\tilde{\rho}_{A,H,\lambda, \overline{u}_{B}}$ is the argmin when $\Lambda(0)>B$, completing the proof of the property in \eqref{Lemma constrained KL: minimization problem}. 
	
	\medskip 
	
	Next we need to prove the property in \eqref{eqn: sup useful for expanding the budget}.  This property is trivial when $B=\infty$.  Assume $B<\infty$ for the remainder of the proof.   Let 
	\[h(u) = \int_{\Theta}A(\theta) d\tilde{\rho}_{A,H,\lambda, u}(\theta) + u\left ( \int_{\Theta}H(\theta)d\tilde{\rho}_{A,H,\lambda, u}(\theta)-B \right ) +\frac{1}{\lambda}D_{\mathrm{KL}} \left (\tilde{\rho}_{A,H,\lambda, u},\pi \right ).\]
	By the definition of $\overline{u}_{B}$, we need to show that the supremum of $h(u)$ over $u\geq 0$ is achieved at $\overline{u}_{B}$.
	Observe that by Corollary \ref{Corollary KL} (a), as $\tilde{\rho}_{A,H,\lambda, u}=\rho_{\lambda(A+uH), \pi}$ in the notation there,
	\begin{align}
	\notag
	&\int_{\Theta} \left \{ A(\theta) +u H(\theta) \right \} d\tilde{\rho}_{A,H,\lambda, u}(\theta) +\frac{1}{\lambda} D_{\mathrm{KL}}(\tilde{\rho}_{A,H,\lambda, u}, \pi )
	\\
	\label{Equation first in proof of Lemma KL useful for expanding budget}
	&= -\frac{1}{\lambda}\log \left [ \int_{\Theta}\exp\left [ -\lambda \left ( A(\theta)+uH(\theta) \right ) \right ] d\pi(\theta) \right ] .
	\end{align}
	Utilizing this it is straightforward to derive that
	\begin{align}
	\notag 
	\frac{d}{du} h(u) &= \frac{\int_{\Theta}H(\theta) \exp\left [ -\lambda \left ( A(\theta)+uH(\theta) \right ) \right ] d\pi(\theta) }{\int_{\Theta}\exp\left [ -\lambda \left ( A(\theta')+uH(\theta') \right ) \right ] d\pi(\theta')} -B
	\\
	\label{Equation in proof of Lemma KL useful for expanding the budget}
	&=\Lambda (u)-B,
	\end{align}
	where we may exchange differentiation and expectation following similar reasoning as before.  
	
	In the proof of the property in \eqref{Lemma constrained KL: minimization problem} it is shown that $\Lambda(u)$ is strictly decreasing on $[0,\infty)$ when $\Lambda(0)\geq B$.  When $\Lambda(0)> B$, by the definition of $\overline{u}_{B}$ we have $\Lambda(\overline{u}_{B})=B$ with $\overline{u}_{B}>0$.  It follows that the supremum of $h(u)$, which is continuous in $u$, is achieved at $\overline{u}_{B}$.  This is because, from \eqref{Equation in proof of Lemma KL useful for expanding the budget}, the derivative of $h(u)$ is positive on $[0,\overline{u}_{B})$, zero at $\overline{u}_{B}$ and decreasing on $(\overline{u}_{B},\infty)$.  If $\Lambda(0)=B$, we have that the supremum is achieved at $0$, which is $\overline{u}_{B}$ in this case by the definition $\overline{u}_{B}$, as the derivative of $h(u)$ is now zero at $u=\overline{u}_{B}=0$ and negative for $u\in (0,\infty)$.  Conversely, if $\Lambda(0)< B$, nearly identical steps to those in the proof of the property in \eqref{Lemma constrained KL: minimization problem} show that $\Lambda(u)$ is non-increasing in $u$ for $u\geq 0$.  Hence in this case the derivative of $h(u)$ is negative for $u\in [0,\infty)$ and the supremum is achieved at $0$, which by definition, is the value of $\overline{u}_{B}$ when $\Lambda(0)\leq B$.
\end{proof}

\bigskip 

\begin{proof}[Proof of Lemma \ref{Lemma: solutions to empirical Gibbs problem rho_hat u_hat and rho_hat}]
	Part (a).  Given Assumptions \ref{Assumption: measurability} and \ref{Assumption: CQ} (i) for $B\in\mathbb{R}\cup \{\infty\}$, this is an immediate corollary of Lemma \ref{Lemma constrained KL} taking $A(\theta)=R_{n}(\theta)$ and $H(\theta)=K_{n}(\theta)$.
	
	Part (b).  Again let  $A(\theta)=R_{n}(\theta)$, $H(\theta)=K_{n}(\theta)$, and $\tilde{\rho}_{A,H,\lambda, u}=\hat{\rho}_{\lambda, u}$ in the notation of Lemma \ref{Lemma constrained KL}.  Observe that, for a fixed sample $S$, as the distribution of $K_{n}(\theta)$ induced by $\theta\sim \hat{\rho}_{\lambda, u}$ is degenerate only when the distribution of $K_{n}(\theta)$ induced by $\pi $ is degenerate, which is assumed to not be the case (with probability one) by Assumption \ref{Assumption: CQ} (ii),  $P^{n}$ almost surely it holds that
	\[\int_{\Theta}K_{n}(\theta) d\hat{\rho}_{\lambda, u}\]
	cannot take any value $b$ that does not satisfy
	\[\pi \left ( \left  \{ \theta: K_{n}(\theta) < b \right  \} \right )>0.\]
	It follows that, $P^{n}$ almost surely,
	\[\pi \left ( \left  \{ \theta: K_{n}(\theta) < \widehat{B}(\hat{\rho}_{\lambda, u}) \right  \} \right )>0.\]
	Then, the result for part (b) follows by applying Lemma \ref{Lemma constrained KL} with $A(\theta)=R_{n}(\theta)$, $H(\theta)=K_{n}(\theta)$, and $B=\widehat{B}(\hat{\rho}_{\lambda, u})$.
\end{proof}

\bigskip 

\subsection{Proofs for Section \ref{sec: PAC Analysis} \label{Subsec: Proofs for PAC Analysis}}
The proofs of Theorems \ref{Theorem: Main Oracle Inequality Constrained Case} and \ref{Theorem: Oracle inequality w/ normal prior} will utilize the following lemmas that follow from Lemma \ref{Lemma constrained KL}.  We again utilize the notation in \eqref{Definition: E_B} and \eqref{Definition: E_B hat} for $\mathcal{E}_{B}$ and $\widehat{\mathcal{E}}_{B}$, respectively.

\begin{lemma}
	\label{Lemma: related to budget expansions in Appendix for Section 4}
	(a)  Let Assumptions \ref{Assumption: measurability} and \ref{Assumption: CQ} (i) hold for $B\in\mathbb{R}$.  For any $\lambda >0$ and $B'\geq B$, $P^{n}$ almost surely it holds that
	\begin{align*}
	&\underset{\widehat{\mathcal{E}}_{B'}}{\min} \left [ \int_{\Theta}R_{n}(\theta) d\rho(\theta) +\frac{1}{\lambda}D_{\mathrm{KL}}(\rho,\pi) \right ] 
	\\
	&= \sup_{u\geq 0} \left [ \int_{\Theta}R_{n}(\theta)d\hat{\rho}_{\lambda, u}(\theta) + u \left ( \int_{\Theta}K_{n}(\theta)d\hat{\rho}_{\lambda, u}(\theta) - B' \right ) +\frac{1}{\lambda}D_{\mathrm{KL}}\left ( \hat{\rho}_{\lambda, u}, \pi \right) \right ].
	\end{align*}
	
	(b)  Let Assumptions \ref{Assumption: measurability} and \ref{Assumption: CQ} (i) hold for $B\in\mathbb{R}$.  The following properties hold $P^{n}$ almost surely.  For any $\lambda >0$ and $B'\geq B$, $u^{*}(B',\lambda)$ exist, is unique, and satisfies that $u^{*}(B',\lambda)=0$ when $\int_{\Theta} K(\theta) d\rho^{*}_{\lambda, 0}(\theta) \leq B'$ whereas, when $\int_{\Theta} K(\theta) d\rho^{*}_{\lambda, 0}(\theta) > B'$,  $u^{*}(B',\lambda)$ is positive and satisfies $\int_{\Theta} K(\theta) d\rho^{*}_{\lambda, u^{*}(B',\lambda)}(\theta) = B'$.  Additionally,
	\begin{align*}
	\rho^{*}_{\lambda, u^{*}(B',\lambda)}= \underset{\mathcal{E}_{B'}}{\arg\min} \left [ \int_{\Theta}R(\theta) d\rho(\theta) +\frac{1}{\lambda}D_{\mathrm{KL}}(\rho,\pi) \right ],
	\end{align*}
	and
	\begin{align}
	\notag 
	&\underset{\mathcal{E}_{B'}}{\min} \left [ \int_{\Theta}R(\theta) d\rho(\theta) +\frac{1}{\lambda}D_{\mathrm{KL}}(\rho,\pi) \right ] 
	\\
	\notag 
	&= \sup_{u\geq 0} \left [ \int_{\Theta}R(\theta)d\rho^{*}_{\lambda, u}(\theta) + u \left ( \int_{\Theta}K(\theta)d\rho^{*}_{\lambda, u}(\theta) - B' \right ) +\frac{1}{\lambda}D_{\mathrm{KL}}\left ( \rho^{*}_{\lambda, u}, \pi \right) \right ].
	\end{align}
	
	(c)  Let Assumptions \ref{Assumption: measurability} and \ref{Assumption: CQ} (ii) hold.  
	For any $B'\geq B ( \hat{\rho}_{\lambda, u}  )$, the following event occurs $P^{n}$ almost surely
	\begin{align*}
	&\min_{\rho\in \mathcal{E}_{B'}}\left [ \int_{\Theta}R(\theta)d\rho(\theta) +\frac{1}{\lambda}D_{\mathrm{KL}} \left ( \rho,\pi \right  ) \right ] 
	\\
	& = \sup_{a \geq 0} \left [ \int_{\Theta}R(\theta)d\rho^{*}_{\lambda, a}(\theta) + a \left ( \int_{\Theta}K(\theta)d\rho^{*}_{\lambda, a}(\theta)-B' \right ) +\frac{1}{\lambda } D_{\mathrm{KL}}\left (  \rho^{*}_{\lambda, a} ,\pi \right )  \right ].
	\end{align*}
\end{lemma}

\begin{proof}[Proof of Lemma \ref{Lemma: related to budget expansions in Appendix for Section 4}]
	Part (a).  When 
	\[\pi\left ( \theta\in\Theta : K_{n}(\theta) < B \right ) >0\]
	holds for $B$, it also holds for $B'\geq B$.  Therefore part (a) follows from Lemma \ref{Lemma constrained KL} with $A(\theta)=R_{n}(\theta)$, $H(\theta)=K_{n}(\theta)$, $\tilde{\rho}_{A,H,\lambda, u}=\hat{\rho}_{\lambda, u}$ and combining the statements in \eqref{Lemma constrained KL: minimization problem} and \eqref{eqn: sup useful for expanding the budget}.
	
	Part (b).   When
	\[\pi\left ( \theta\in\Theta: K(\theta) < B \right ) >0 \]
	holds for $B$, it also holds for $B'\geq B$.  Then the result follows from Lemma \ref{Lemma constrained KL} with $A(\theta)=R(\theta)$, $H(\theta)=K(\theta)$, and $\tilde{\rho}_{A,H,\lambda, u}=\rho^{*}_{\lambda, u}$.
	
	Part (c).  Note that by Assumption \ref{Assumption: CQ} (ii), as $\hat{\rho}_{\lambda, u}$ and $\pi$ are each absolutely continuous with respect to the other, we have that the distribution of $K(\theta)$ induced by $\theta\sim \hat{\rho}_{\lambda, u}$ is not degenerate.  Therefore,  
	\[\pi\left ( \left  \{\theta: K(\theta) < \int_{\Theta}K(\theta) d\hat{\rho}_{\lambda, u} =B(\hat{\rho}_{\lambda, u})  \right \}  \right )>0.\]
	It follows that for any $B'\geq B(\hat{\rho}_{\lambda, u})$, we have
	\[\pi\left ( \{\theta: K(\theta) < B' \}  \right )>0.\]
	Then, the result of part (c) follows from applying Lemma \ref{Lemma constrained KL} with $A(\theta)=R(\theta)$, $H(\theta)=K(\theta)$, and $\tilde{\rho}_{A,H,\lambda, u}=\rho^{*}_{\lambda, u}$ and then combining equations \eqref{Lemma constrained KL: minimization problem} and \eqref{eqn: sup useful for expanding the budget} there.  
\end{proof}

\subsubsection{Proofs for Subsection \ref{subsec: Regret Bounds and Oracle Inequalities} \label{Subsec: Proofs for Regret Bounds and Oracle Inequalities}: Regret Bounds and Oracle-Type Inequalities}

\begin{proof}[Proof of Theorem \ref{Theorem: PAC-Bayesian Generalization bounds}]
	Part (a).   When $V_{n}(\theta)=R_{n}(\theta)$, $V(\theta)=R(\theta)$, and $M_{\ell}=M_{y}$, we have the setup for Theorem \ref{Theorem: General PAC-Bayesian Generalization Bound} with 
	\[\ell_{v}(Z,\theta)= \left ( \frac{YD}{e(X)} - \frac{Y(1-D)}{1-e(X)} \right )(f^{*}(X)-f_{\theta}(X)),\]
	$L(\theta)=V(\theta)$, and $L_{n}(\theta)=V_{n}(\theta)$.  Note that, by Assumption \ref{Assumption: treatment identification and boundedness}, parts (iii) and (iv), we have that $-M_{\ell}/2\kappa \leq \ell_{v}(Z,\theta) \leq M_{\ell}/2\kappa$ a.s.  
	Similarly, when $V_{n}(\theta)=K_{n}(\theta)$, $V(\theta)=K(\theta)$, and $M_{\ell}=M_{c}$, we have the setup for Theorem \ref{Theorem: General PAC-Bayesian Generalization Bound}  now with
	\[\ell_{v}(Z,\theta) = \left ( \frac{CD}{e(X)} - \frac{C(1-D)}{1-e(X)} \right ) f_{\theta}(X),\]
	and again taking $L(\theta)=V(\theta)$ and $L_{n}(\theta)=V_{n}(\theta)$.  Again Assumption \ref{Assumption: treatment identification and boundedness}, parts (iii) and (iv), yields that $-M_{\ell}/2\kappa \leq \ell_{v}(Z,\theta) \leq M_{\ell}/2\kappa$ a.s.  
	
	Given this setup, we apply Theorem \ref{Theorem: General PAC-Bayesian Generalization Bound} in the same way for either of the settings for $L(\theta), L_{n}(\theta)$ and $M_{\ell}$.   We need an appropriate choice for $D(\cdot, \cdot )$ and to then verify the condition in \eqref{Equation: condition in General PAC-Bayesian Generalization Bound}.  Importantly, in either setting we have that, for any $\theta\in\Theta$, $\ell_{v}(Z_{1},\theta),\dots, \ell_{v}(Z_{n},\theta)$ is an iid set of random variables taking values in $[-M_{\ell}/2\kappa , M_{\ell}/2\kappa ]$ almost surely.  For either $s\in\{-1,1\}$, take $D[L_{n}(\theta), L(\theta) ] = s(L_{n}(\theta)-L(\theta))$.  We need to verify the condition in \eqref{Equation: condition in General PAC-Bayesian Generalization Bound} and determine an appropriate $f(\lambda, n)$.    Start with $s=1$.  Then, by Hoeffding's lemma (see, for example, \cite{massart2007concentration}, page 21), for any $\theta\in\Theta$,
	\begin{align}
	\notag 
	E_{P^{n}} \left [  \exp \left ( \lambda \left [ L_{n}(\theta)-L(\theta) \right ] \right ) \right ] &= E_{P^{n}} \left [  \exp \left ( \frac{\lambda}{n} \sum_{i=1}^{n} \left ( \ell_{v}(Z_{i},\theta)- E_{P}\left [ \ell_{v}(Z_{i},\theta) \right ] \right ) \right ) \right ]
	\\
	\notag 
	&= \prod^{n}_{i=1}E_{P} \left [  \exp \left \{ \frac{\lambda}{n} \left ( \ell_{v}(Z_{i},\theta)- E_{P} \left [ \ell_{v}(Z_{i},\theta) \right ] \right ) \right \} \right ]
	\\
	\label{eqn: Hoeffding bound for pb gen bound in main oracle inequality proof s=1}
	&\leq \prod^{n}_{i=1} \exp \left (  \frac{\lambda^{2}M_{\ell}^{2}}{8\kappa^{2}n^{2}} \right ) = \exp \left ( \frac{\lambda^{2}M^{2}_{\ell}}{8\kappa^{2}n} \right )
	\end{align}
	Nearly identical steps in the $s=-1$ case, now applying Hoeffding's lemma to $-\ell_{v}(Z_{i},\theta)$ produce that
	\begin{equation}
	\label{eqn: Hoeffding bound for pb gen bound in main oracle inequality proof s=-1}
	E_{P^{n}} \left [  \exp \left ( \lambda \left [ L(\theta)-L_{n}(\theta) \right ] \right ) \right ] \leq \exp \left ( \frac{\lambda^{2}M^{2}_{\ell}}{8\kappa^{2}n} \right ).
	\end{equation} 
	Integrating with respect to $\pi$, \eqref{eqn: Hoeffding bound for pb gen bound in main oracle inequality proof s=1} and \eqref{eqn: Hoeffding bound for pb gen bound in main oracle inequality proof s=-1} yield that 
	\[\int_{\Theta} E_{P^{n}} \left [ \lambda s \left ( R_{n}(\theta) - R(\theta) \right ) \right ] d\pi(\theta) \leq \exp \left (  \frac{\lambda^{2}M^{2}_{\ell}}{8\kappa^{2}n} \right ), \ \ s\in\{-1,1\}. \]
	We can reverse the order of integration on the left-hand of the above inequality, as $\pi$ is independent of the sample by Assumption \ref{Assumption: prior indep of data}.  Therefore, condition \eqref{Equation: condition in General PAC-Bayesian Generalization Bound} in Theorem \ref{Theorem: General PAC-Bayesian Generalization Bound} holds with $f(\lambda, n)= \lambda^{2}M^{2}_{\ell}/(8 n \kappa^{2})$.  Applying Theorem \ref{Theorem: General PAC-Bayesian Generalization Bound} completes the proof for Part (a).
	
	\medskip 
	
	Part (b).  We utilize the same notation in terms of $\ell_{v}(Z,\theta)$ in the two scenarios for $V_{n}(\theta)$, $V(\theta)$, and $M_{\ell}$ as in part (a). Let $E_{1}$ denote the event that the following inequality holds,
	\begin{equation}
	\label{Dkl bound in part b of first pac generalization theorem sec 3}
	D_{\mathrm{KL}}\left ( \hat{\rho}_{\lambda, u}, \rho^{*}_{\lambda, u} \right ) \leq \frac{\lambda \sqrt{2} \left ( M_{y} +uM_{c} \right ) }{\kappa \sqrt{n}}\sqrt{ \log \left ( 2\sqrt{n} \right ) + \log\frac{2}{\epsilon }  } + \frac{\lambda^{2}\left ( M_{y}+uM_{c} \right )^{2}}{2 n \kappa^{2} }. 
	\end{equation}
	Note that by Lemma \ref{Lemma: high prob bound for KL(rho hat, rho star)},  $P^{n}(E_{1}) \geq 1-\epsilon/2$.
	
	Next, let $E_{2}$ denote the event that the following inequality holds,
	\begin{equation}
	\label{2nd event in part b of first pac generalization theorem sec 3}
	\left ( \int_{\Theta} \left [  V_{n}(\theta) -V(\theta) \right ] d\hat{\rho}_{\lambda,u}(\theta)  \right )^{2} \leq \frac{M_{\ell}^{2}}{2n\kappa^{2}} \left [ D_{\mathrm{KL}}\left ( \hat{\rho}_{\lambda,u}, \rho^{*}_{\lambda,u} \right ) + \log\left ( 2\sqrt{n} \right ) +\log\frac{2}{\epsilon} \right ].
	\end{equation}
	
	In the setup of Theorem \ref{Theorem: adaptation of Seeger's bound} (b), take 
	\[\ell(Z,\theta)= \left (\ell_{v}(Z,\theta)+\frac{M_{\ell}}{2\kappa }  \right ) \left (\frac{\kappa}{M_{\ell}} \right ).\]
	Then, for any $\theta\in \Theta $,  $\ell(Z,\theta)\in[0,1]$  ($P$ almost surely).  Applying Theorem \ref{Theorem: adaptation of Seeger's bound} (b) yields that $P^{n}(E_{2})\geq 1-\epsilon/2$.
	
	Then, the following a union bound argument,
	\begin{align*}
	P^{n}\left ( E_{1}  \cap E_{2} \right ) &=1-P^{n}\left (E_{1}^{c} \cup E_{2}^{c} \right )
	\\
	&\geq 1- P^{n}\left (  E_{1} \right )- P^{n}\left ( E_{2} \right )
	\\
	&\geq 1-\frac{\epsilon}{2}-\frac{\epsilon}{2} = 1-\epsilon,
	\end{align*}
	yields that events $E_{1}$ and $E_{2}$ occur jointly with probability greater than $1-\epsilon$.  In the intersection of these events, plugging \eqref{Dkl bound in part b of first pac generalization theorem sec 3} into \eqref{2nd event in part b of first pac generalization theorem sec 3} produces the result in part (b). 
	
	\medskip 
	
	Part (c).  The proof follows similar steps to that in part (b).  Define $E_{1}$ the same way as in part (b).  Now,  $E_{2}$ is defined to be the event that 
	\[ \int_{\Theta}s \left [ V_{n}(\theta) -  V(\theta) \right ] d\hat{\rho}_{\lambda, u}(\theta) \leq  \frac{1}{\lambda}  D_{\mathrm{KL}}(\rho,\pi) + \frac{1}{\lambda} \left [ \frac{\lambda^{2}M_{\ell}^{2}}{8n\kappa^{2}} +\log\frac{2}{\epsilon} \right ].   \] 
	By part (a), $P^{n}(E_{1})\geq 1-\epsilon/2$.  Then, event $E_{2}$ is defined the same way is in the proof of part (b), $P(E_{1}\cap E_{2})>1-\epsilon$ by a union bound argument, and combining the inequalities in $E_{1}$ and $E_{2}$ produces the statement of part (c).
\end{proof}

\bigskip

\begin{proof}[Proof of Theorem \ref{Theorem: Main Oracle Inequality Constrained Case}]
	
	Part (a).  Let $E_{1}$ denote the event that, for all $\mathcal{\rho}\in\mathcal{P}_{\pi}(\Theta)$ simultaneously it holds that
	\begin{equation}
	\label{Event E1}
	\int_{\Theta} R(\theta) d\rho(\theta) \leq \int_{\Theta}R_{n}(\theta) d\rho(\theta) +\frac{1}{\lambda}D_{\mathrm{KL}}(\rho,\pi ) +\frac{1}{\lambda} \left [ \frac{\lambda^{2}M^{2}_{y}}{8n\kappa^{2}} +\log\frac{3}{\epsilon} \right ].
	\end{equation}
	Let $E_{2}$ denote the event that, for all $\mathcal{\rho}\in\mathcal{P}_{\pi}(\Theta)$ simultaneously it holds that
	\begin{equation}
	\label{Event E2}
	\int_{\Theta} R_{n}(\theta) d\rho(\theta) \leq \int_{\Theta}R(\theta) d\rho(\theta) +\frac{1}{\lambda}D_{\mathrm{KL}}(\rho,\pi ) +\frac{1}{\lambda} \left [ \frac{\lambda^{2}M^{2}_{y}}{8n\kappa^{2}} +\log\frac{3}{\epsilon} \right ].
	\end{equation}
	Lastly, let $u^{*} = u^{*}(B,\lambda/2)$ as specified in Definition \ref{Def: u hat and u star} and let $E_{3}$ denote the event that
	\begin{equation}
	\label{Event E3}
	\int_{\Theta}K_{n}(\theta)d\rho^{*}_{\lambda/2, u^{*}}(\theta) - \int_{\Theta}K(\theta)d\rho^{*}_{\lambda/2, u^{*}}(\theta) \leq   \sqrt{\frac{M_{c}^{2}\log\frac{3}{\epsilon }}{2n\kappa^{2}} }, 
	\end{equation}
	where $\rho^{*}_{\lambda/2, u}$ is given in Definition \ref{Definition: optimal rho hat under a budget constraint}.  
	
	By Theorem \ref{Theorem: PAC-Bayesian Generalization bounds} (a), applied to each $s\in\{-1,1\}$ with $V_{n}(\theta)=R_{n}(\theta),  \ V(\theta)=R(\theta)$,  and by Lemma \ref{Lemma: McDiarmid inequality applied to sample cost with non-random probability measure}, respectively, we have
	\[P^{n} \left ( E_{1} \right )\geq 1-\frac{\epsilon}{3}, \ P^{n} \left ( E_{2} \right ) \geq 1-\frac{\epsilon}{3}, \ \mathrm{and} \  P^{n} \left (E_{3} \right )\geq 1-\frac{\epsilon}{3} . \]  
	Applying a union bound argument as in the proof of Theorem \ref{Theorem: PAC-Bayesian Generalization bounds} (b), it holds that $P^{n}(E_{1}\cap E_{2} \cap E_{3})\geq 1-\epsilon$.  From the remainder of the proof, we work assuming the intersection of these three events.  We show the event in Theorem \ref{Theorem: Main Oracle Inequality Constrained Case} (a) is implied by their intersection, hence the event of interest contains this intersection and has probability greater than or equal $1-\epsilon$.
	
	We consider two possible scenarios in conjuncture with events $E_{1}$, $E_{2}$, and $E_{3}$.  In the first scenario, suppose that
	\begin{equation}
	\label{1st Scenario General Oracle Ineq part a}
	\int_{\Theta} K_{n}(\theta) d\rho^{*}_{\lambda/2, u^{*}}(\theta) \leq B.
	\end{equation}
	In this case  $\rho^{*}_{\lambda/2, u^{*}}\in\widehat{\mathcal{E}}_{B}$, where $\widehat{\mathcal{E}}_{B}$ is given by \eqref{Definition: E_B hat}.  Starting from \eqref{Event E1} with $\rho=\hat{\rho}_{\lambda,\hat{u}}$,
	\begin{align*}
	\int_{\Theta} R(\theta) d\hat{\rho}_{\lambda,\hat{u}}(\theta) & \leq  \int_{\Theta}R_{n}(\theta) d\hat{\rho}_{\lambda,\hat{u}}(\theta) +\frac{1}{\lambda}D_{\mathrm{KL}}(\hat{\rho}_{\lambda,\hat{u}},\pi ) +\frac{1}{\lambda} \left [ \frac{\lambda^{2}M^{2}_{y}}{8n\kappa^{2}} +\log\frac{3}{\epsilon} \right ]
	\\
	&=  \min_{\rho\in \widehat{\mathcal{E}}_{B}} \left \{  \int_{\Theta}R_{n}(\theta) d\rho(\theta) +\frac{1}{\lambda}D_{\mathrm{KL}}(\rho,\pi ) \right \} +\frac{1}{\lambda} \left [ \frac{\lambda^{2}M^{2}_{y}}{8n\kappa^{2}} +\log\frac{3}{\epsilon}\right ] 
	\\
	&\leq  \int_{\Theta}R_{n}(\theta) d\rho^{*}_{\lambda/2, u^{*}}(\theta) +\frac{1}{\lambda}D_{\mathrm{KL}}(\rho^{*}_{\lambda/2, u^{*}},\pi ) +\frac{1}{\lambda} \left [ \frac{\lambda^{2}M^{2}_{y}}{8n\kappa^{2}} +\log\frac{3}{\epsilon}\right ].
	\end{align*}
	The second equality above follows from Lemma \ref{Lemma: solutions to empirical Gibbs problem rho_hat u_hat and rho_hat} (a).  Now, consider \eqref{Event E2} with $\rho = \rho^{*}_{\lambda/2, u^{*}}$.  Plugging this inequality into the right-hand side of the above inequality produces
	\begin{align*}
	\int_{\Theta} R(\theta) d\hat{\rho}_{\lambda,\hat{u}}(\theta) &\leq \int_{\Theta} R(\theta) d\rho^{*}_{\lambda/2, u^{*}}(\theta) + \frac{2}{\lambda }D_{\mathrm{KL}}(\rho^{*}_{\lambda/2, u^{*}}, \pi ) +\frac{2}{\lambda} \left [ \frac{\lambda^{2}M^{2}_{y}}{8n\kappa^{2}} +\log\frac{3}{\epsilon}\right ]
	\\
	&= \min_{\rho\in \mathcal{E}_{B}} \left \{ \int_{\Theta}R(\theta) d\rho(\theta) + \frac{2}{\lambda}D_{\mathrm{KL}}(\rho,\pi) \right \}  + \frac{2}{\lambda} \left [ \frac{\lambda^{2}M^{2}_{y}}{8n\kappa^{2}} +\log\frac{3}{\epsilon}\right ] 
	\\
	& \leq \min_{\rho\in \mathcal{E}_{B}} \left \{ \int_{\Theta}R(\theta) d\rho(\theta) + \frac{2}{\lambda}D_{\mathrm{KL}}(\rho,\pi) \right \}  + \frac{2}{\lambda} \left [ \frac{\lambda^{2}M^{2}_{y}}{8n\kappa^{2}} +\log\frac{3}{\epsilon}\right ]  + \hat{u} \sqrt{\frac{M_{c}^{2}\log\frac{3}{\epsilon }}{2n\kappa^{2}} } ,
	\end{align*}
	where the equality in the second row follows from Lemma \ref{Lemma: related to budget expansions in Appendix for Section 4} (b) and the final inequality holds as $\hat{u}\geq 0$.  Thus the result of part (a) holds in the first scenario described by \eqref{1st Scenario General Oracle Ineq part a}, noting that for $\rho\in\mathcal{P}(\Theta)$, $R(f_{G,\rho})=\int_{\Theta}R(\theta)d\rho(\theta)$.
	
	In the second and only remaining scenario, we consider when  
	\begin{equation}
	\label{2nd Scenario General Oracle Ineq part a}
	\int_{\Theta} K_{n}(\theta) d\rho^{*}_{\lambda/2, u^{*}}(\theta) > B.
	\end{equation}
	If we set
	\begin{equation}
	\label{1st orac inequality part (a) B'}
	B' = \int_{\Theta} K_{n}(\theta) d\rho^{*}_{\lambda/2, u^{*}}(\theta),
	\end{equation}
	then it holds that $\rho^{*}_{\lambda/2, u^{*}} \in \widehat{\mathcal{E}}_{B'}$.  Again starting from the event in \eqref{Event E1} with $\rho=\hat{\rho}_{\lambda,\hat{u}}$, we obtain
	\begin{align}
	\notag 
	&\int_{\Theta} R(\theta) d\hat{\rho}_{\lambda,\hat{u}}(\theta) 
	\\
	\notag 
	&\leq \int_{\Theta}R_{n}(\theta) d\hat{\rho}_{\lambda,\hat{u}}(\theta) +\frac{1}{\lambda}D_{\mathrm{KL}}(\hat{\rho}_{\lambda,\hat{u}},\pi ) +\frac{1}{\lambda} \left [ \frac{\lambda^{2}M^{2}_{y}}{8n\kappa^{2}} +\log\frac{3}{\epsilon} \right ]
	\\
	\label{1st orac ineq part (a) eq chain 0}
	&= \int_{\Theta}R_{n}(\theta) d\hat{\rho}_{\lambda,\hat{u}}(\theta) +\frac{1}{\lambda}D_{\mathrm{KL}}(\hat{\rho}_{\lambda,\hat{u}},\pi )+\hat{u}\left ( \int_{\Theta}K_{n}(\theta)d\hat{\rho}_{\lambda,\hat{u}}(\theta) -B \right ) +\frac{1}{\lambda} \left [ \frac{\lambda^{2}M^{2}_{y}}{8n\kappa^{2}} +\log\frac{3}{\epsilon} \right ]
	\\
	\notag 
	&= \int_{\Theta}R_{n}(\theta) d\hat{\rho}_{\lambda,\hat{u}}(\theta) + \hat{u} \left ( \int_{\Theta}K_{n}(\theta)d\hat{\rho}_{\lambda,\hat{u}}(\theta) - B' \right ) + \frac{1}{\lambda}D_{\mathrm{KL}}(\hat{\rho}_{\lambda,\hat{u}},\pi )
	\\
	\label{1st orac ineq part (a) eq chain 1}
	& \hspace{0.25in} + \hat{u}\left (  \int_{\Theta} K_{n}(\theta) d\rho^{*}_{\lambda/2, u^{*}}(\theta) - B \right ) + \frac{1}{\lambda} \left [ \frac{\lambda^{2}M^{2}_{y}}{8n\kappa^{2}} +\log\frac{3}{\epsilon} \right ]
	\\
	\notag 
	&\leq \sup_{u \geq 0} \left [ \int_{\Theta}R_{n}(\theta) d\hat{\rho}_{\lambda,u}(\theta) + u \left ( \int_{\Theta}K_{n}(\theta)d\hat{\rho}_{\lambda,u}(\theta) - B' \right ) + \frac{1}{\lambda}D_{\mathrm{KL}}(\hat{\rho}_{\lambda,u},\pi )  \right ]  
	\\
	\notag 
	& \hspace{0.25in} + \hat{u}\left (  \int_{\Theta} K_{n}(\theta) d\rho^{*}_{\lambda/2, u^{*}}(\theta) - B \right ) + \frac{1}{\lambda} \left [ \frac{\lambda^{2}M^{2}_{y}}{8n\kappa^{2}} +\log\frac{3}{\epsilon} \right ]
	\\
	\notag 
	&= \min_{\rho\in \widehat{\mathcal{E}}_{B'}} \left \{  \int_{\Theta}R_{n}(\theta) d\rho(\theta)  + \frac{1}{\lambda}D_{\mathrm{KL}}(\rho,\pi )  \right \} 
	\\
	\label{1st orac ineq part (a) eq chain 2} 
	&\hspace{0.25in} + \hat{u}\left (  \int_{\Theta} K_{n}(\theta) d\rho^{*}_{\lambda/2, u^{*}}(\theta) - B \right ) + \frac{1}{\lambda} \left [ \frac{\lambda^{2}M^{2}_{y}}{8n\kappa^{2}} +\log\frac{3}{\epsilon} \right ]
	\\
	\label{1st orac ineq part (a) eq chain 3}
	& \leq  \min_{\rho \in \widehat{\mathcal{E}}_{B'}} \left \{  \int_{\Theta}R_{n}(\theta) d\rho(\theta)  + \frac{1}{\lambda}D_{\mathrm{KL}}(\rho,\pi )  \right \} + \hat{u}\sqrt{\frac{M_{c}^{2}\log\frac{3}{\epsilon }}{2n\kappa^{2}} } + \frac{1}{\lambda} \left [ \frac{\lambda^{2}M^{2}_{y}}{8n\kappa^{2}} +\log\frac{3}{\epsilon} \right ]
	\\
	\label{1st orac ineq part (a) eq chain 4}
	&\leq  \int_{\Theta}R_{n}(\theta) d\rho^{*}_{\lambda/2, u^{*}}(\theta)  + \frac{1}{\lambda}D_{\mathrm{KL}}(\rho^{*}_{\lambda/2, u^{*}},\pi ) + \hat{u}\sqrt{\frac{M_{c}^{2}\log\frac{3}{\epsilon }}{2n\kappa^{2}} } + \frac{1}{\lambda} \left [ \frac{\lambda^{2}M^{2}_{y}}{8n\kappa^{2}} +\log\frac{3}{\epsilon} \right ]
	\\
	\label{1st orac ineq part (a) eq chain 6}
	&\leq \int_{\Theta}R(\theta) d\rho^{*}_{\lambda/2, u^{*}}(\theta)  + \frac{2}{\lambda}D_{\mathrm{KL}}(\rho^{*}_{\lambda/2, u^{*}},\pi ) + \hat{u}\sqrt{\frac{M_{c}^{2}\log\frac{3}{\epsilon }}{2n\kappa^{2}} } + \frac{2}{\lambda} \left [ \frac{\lambda^{2}M^{2}_{y}}{8n\kappa^{2}} +\log\frac{3}{\epsilon} \right ]
	\\
	\label{1st orac ineq part (a) eq chain 5}
	& =  \min_{\rho\in \mathcal{E}_{B}} \left \{ \int_{\Theta}R(\theta) d\rho(\theta) + \frac{2}{\lambda}D_{\mathrm{KL}}(\rho,\pi) \right \}  + \frac{2}{\lambda} \left [ \frac{\lambda^{2}M^{2}_{y}}{8n\kappa^{2}} +\log\frac{3}{\epsilon}\right ] + \hat{u} \sqrt{\frac{M_{c}^{2}\log\frac{3}{\epsilon }}{2n\kappa^{2}} } .
	\end{align}
	In the above, step \eqref{1st orac ineq part (a) eq chain 0} follows from the properties of $\hat{u}=\hat{u}(B,\lambda)$ in Lemma \ref{Lemma: solutions to empirical Gibbs problem rho_hat u_hat and rho_hat} (a).  In step \eqref{1st orac ineq part (a) eq chain 1} we simply added and subtracted $\hat{u} B'$ with $B'$ given in \eqref{1st orac inequality part (a) B'}.  Step \eqref{1st orac ineq part (a) eq chain 2} follows from Lemma \ref{Lemma: related to budget expansions in Appendix for Section 4} (a).  Step \eqref{1st orac ineq part (a) eq chain 3} follows from \eqref{Event E3} and the observation that $\int_{\Theta}K(\theta) d\rho^{*}_{\lambda/2, u^{*}}(\theta)$ is always less than or equal to $B$ by Lemma \ref{Lemma: related to budget expansions in Appendix for Section 4} (b).  Step \eqref{1st orac ineq part (a) eq chain 4} follows from fact that $\rho^{*}_{\lambda/2, u^{*}} \in \widehat{\mathcal{E}}_{B'}$ by the construction of $B'$ in \eqref{1st orac inequality part (a) B'}.  \eqref{1st orac ineq part (a) eq chain 6} follows from \eqref{Event E2} with $\rho = \rho^{*}_{\lambda/2, u^{*}}$ and lastly \eqref{1st orac ineq part (a) eq chain 5} follows from Lemma \ref{Lemma: related to budget expansions in Appendix for Section 4} (b).
	
	It follows that the result in part (a) also holds in the second scenario in \eqref{2nd Scenario General Oracle Ineq part a} which completes the proof for this part.

	\medskip 
	Part (b).  Now, let $E_{1}$ denote the event that, for all $\mathcal{\rho}\in\mathcal{P}_{\pi}(\Theta)$ simultaneously it holds that
	\begin{equation}
	\label{Event E1 (b)}
	\int_{\Theta} R(\theta) d\rho(\theta) \leq \int_{\Theta}R_{n}(\theta) d\rho(\theta) +\frac{1}{\lambda}D_{\mathrm{KL}}(\rho,\pi ) +\frac{1}{\lambda} \left [ \frac{\lambda^{2}M^{2}_{y}}{8n\kappa^{2}} +\log\frac{4}{\epsilon} \right ].
	\end{equation}
	Let $E_{2}$ denote the event that
	\begin{align}
	\notag 
	& \int_{\Theta}R_{n}(\theta) d\hat{\rho}_{\lambda, u}(\theta)+u\int_{\Theta}K_{n}(\theta)d\hat{\rho}_{\lambda,u}(\theta) +\frac{1}{\lambda} D_{\mathrm{KL}} \left ( \hat{\rho}_{\lambda,u},\pi  \right )  
	\\
	\label{Event E2 (b)}
	&  \leq   \int_{\Theta}R(\theta)d\rho^{*}_{\lambda,u}(\theta)+u\int_{\Theta}K(\theta)d\rho^{*}_{\lambda,u}(\theta)+\frac{1}{\lambda}D_{\mathrm{KL}}(\rho^{*}_{\lambda,u},\pi )  +\sqrt{\frac{(M_{y}+uM_{c})^{2}\log(4/\epsilon)}{2n\kappa^{2}}},
	\end{align}
	and let $E_{3}$ denote the event that 
	\begin{align}
	\notag 
	&\int_{\Theta}K(\theta)d\hat{\rho}_{\lambda, u}(\theta) - \int_{\Theta}K_{n}(\theta)d\hat{\rho}_{\lambda, u}(\theta)  
	\\
	\notag 
	&\leq \frac{\sqrt{2} \left ( M_{y} +uM_{c} \right ) }{\kappa \sqrt{n}}\sqrt{ \log \left ( 2\sqrt{n} \right ) + \log\frac{4}{\epsilon }  } + \frac{\lambda\left ( M_{y}+uM_{c} \right )^{2}}{2 n \kappa^{2} } + \frac{1}{\lambda} \left [ \frac{\lambda^{2}M_{c}^{2}}{8n\kappa^{2}} +\log\frac{4}{\epsilon} \right ]
	\\
	\label{Event E3 (b)}
	&=U_{1}\left ( \epsilon ; \lambda, u , n \right ) +\frac{1}{\lambda} \left [ \frac{\lambda^{2}M_{c}^{2}}{8n\kappa^{2}} +\log\frac{4}{\epsilon} \right ]
	\end{align}
	By Theorem \ref{Theorem: PAC-Bayesian Generalization bounds} (a), applied with $s=-1$, $V_{n}(\theta)=R_{n}(\theta)$, $V(\theta)=R(\theta)$, and $M_{\ell}=M_{y}$, we have that $P^{n}(E_{1})=\epsilon/4$.  By Lemma \ref{Lemma: Upper Bound via Rho star elements}, $P^{n}(E_{2})=\epsilon /4$.  And lastly, by Theorem \ref{Theorem: PAC-Bayesian Generalization bounds} (c) with $V_{n}(\theta)=K_{n}(\theta)$, $V(\theta)=K(\theta)$, and $M_{\ell}=M_{c}$, it holds that $P^{n}(E_{3})=\epsilon/2$.  Again applying a union bound argument similar to that in the proof of Theorem \ref{Theorem: PAC-Bayesian Generalization bounds} (b), we have
	\[P^{n}\left ( E_{1}\cap E_{2}\cap E_{3} \right ) \geq  1-\epsilon. \]
	As in part (a), we prove the result by showing that the intersection of these events implies the event in the result.
	
	Recall, 
	\[B \left (  \hat{\rho}_{\lambda, u}  \right )=\int_{\Theta}K(\theta) d \hat{\rho}_{\lambda, u}(\theta) \ \mathrm{and} \ \widehat{B} \left (\hat{\rho}_{\lambda, u} \right )= \int_{\Theta}K_{n}(\theta) d \hat{\rho}_{\lambda, u}(\theta). \]
	Then, the event $E_{3}$ described in \eqref{Event E3 (b)} can be stated
	\begin{equation}
	\label{implication on Event E4 (b)}
	B \left (  \hat{\rho}_{\lambda, u}  \right )-\widehat{B} \left (\hat{\rho}_{\lambda, u} \right ) \leq U_{1}\left ( \epsilon ; \lambda, u , n \right ) + \frac{1}{\lambda} \left [ \frac{\lambda^{2}M_{c}^{2}}{8n\kappa^{2}} +\log\frac{4}{\epsilon} \right ].
	\end{equation}
	Now, starting from \eqref{Event E1 (b)} with $\rho= \hat{\rho}_{\lambda, u}$, 
	\small 
	\begin{align}
	\notag 
	&\int_{\Theta}R(\theta) d\hat{\rho}_{\lambda, u}(\theta ) 
	\\
	\notag 
	&\leq \int_{\Theta}R_{n}(\theta)d\hat{\rho}_{\lambda, u }(\theta) + \frac{1}{\lambda}D_{\mathrm{KL}}(\hat{\rho}_{\lambda, u},\pi ) +\frac{1}{\lambda} \left [ \frac{\lambda^{2}M^{2}_{y}}{8n\kappa^{2}} +\log\frac{4}{\epsilon} \right ]
	\\
	\notag 
	&=\int_{\Theta}R_{n}(\theta)d\hat{\rho}_{\lambda, u }(\theta) + u\left ( \int_{\Theta}K_{n}(\theta)d\hat{\rho}_{\lambda, u} -\widehat{B} \left ( \hat{\rho}_{\lambda, u} \right ) \right ) + \frac{1}{\lambda}D_{\mathrm{KL}}(\hat{\rho}_{\lambda, u},\pi ) +\frac{1}{\lambda} \left [ \frac{\lambda^{2}M^{2}_{y}}{8n\kappa^{2}} +\log\frac{4}{\epsilon} \right ]
	\\
	\notag 
	&=\int_{\Theta}R_{n}(\theta)d\hat{\rho}_{\lambda, u }(\theta) + u \int_{\Theta}K_{n}(\theta)d\hat{\rho}_{\lambda, u} + \frac{1}{\lambda}D_{\mathrm{KL}}(\hat{\rho}_{\lambda, u},\pi ) 
	\\
	\notag 
	&\hspace{0.25in} - u B \left ( \hat{\rho}_{\lambda, u}  \right ) + u\left (B \left ( \hat{\rho}_{\lambda, u}  \right ) - \widehat{B} \left ( \hat{\rho}_{\lambda, u} \right ) \right ) +\frac{1}{\lambda} \left [ \frac{\lambda^{2}M^{2}_{y}}{8n\kappa^{2}} +\log\frac{4}{\epsilon} \right ]
	\\
	\notag 
	&\leq \int_{\Theta}R(\theta)d\rho^{*}_{\lambda,u}(\theta)+u\int_{\Theta}K(\theta)d\rho^{*}_{\lambda,u}(\theta)+\frac{1}{\lambda}D_{\mathrm{KL}}(\rho^{*}_{\lambda,u},\pi )  +\sqrt{\frac{(M_{y}+uM_{c})^{2}\log(4/\epsilon)}{2n\kappa^{2}}}
	\\
	\label{1st orac ineq part (b) eq chain 1}
	&\hspace{0.25in} - u B \left ( \hat{\rho}_{\lambda, u}  \right ) + u\left (B \left ( \hat{\rho}_{\lambda, u}  \right ) - \widehat{B} \left ( \hat{\rho}_{\lambda, u} \right ) \right ) +\frac{1}{\lambda} \left [ \frac{\lambda^{2}M^{2}_{y}}{8n\kappa^{2}} +\log\frac{4}{\epsilon} \right ]
	\\
	\notag 
	&= \int_{\Theta}R(\theta)d\rho^{*}_{\lambda,u}(\theta)+u \left ( \int_{\Theta}K(\theta)d\rho^{*}_{\lambda,u}(\theta)-B \left ( \hat{\rho}_{\lambda, u}  \right ) \right ) +\frac{1}{\lambda}D_{\mathrm{KL}}(\rho^{*}_{\lambda,u},\pi )  +\sqrt{\frac{(M_{y}+uM_{c})^{2}\log(4/\epsilon)}{2n\kappa^{2}}}
	\\
	\notag 
	&\hspace{0.25in} + u\left ( B \left ( \hat{\rho}_{\lambda, u}  \right ) - \widehat{B} \left ( \hat{\rho}_{\lambda, u} \right ) \right ) +\frac{1}{\lambda} \left [ \frac{\lambda^{2}M^{2}_{y}}{8n\kappa^{2}} +\log\frac{4}{\epsilon} \right ]
	\\
	\notag 
	&\leq \int_{\Theta}R(\theta)d\rho^{*}_{\lambda,u}(\theta)+u \left ( \int_{\Theta}K(\theta)d\rho^{*}_{\lambda,u}(\theta)- B \left ( \hat{\rho}_{\lambda, u}  \right ) \right ) +\frac{1}{\lambda}D_{\mathrm{KL}}(\rho^{*}_{\lambda,u},\pi )  +\sqrt{\frac{(M_{y}+uM_{c})^{2}\log(4/\epsilon)}{2n\kappa^{2}}}
	\\
	\label{1st orac ineq part (b) eq chain 2}
	&\hspace{0.25in} + u U_{1}\left ( \epsilon ; \lambda, u , n \right )  +\frac{1}{\lambda} \left [ \frac{\lambda^{2} \left ( M^{2}_{y} +u M_{c}^{2} \right )}{8n\kappa^{2}} +(1+u)\log\frac{4}{\epsilon} \right ]
	\\
	\notag 
	& \leq \sup_{a\geq 0}\left [ \int_{\Theta}R(\theta)d\rho^{*}_{\lambda,a}(\theta)+a \left ( \int_{\Theta}K(\theta)d\rho^{*}_{\lambda,a}(\theta)- B \left ( \hat{\rho}_{\lambda, u}  \right ) \right ) +\frac{1}{\lambda}D_{\mathrm{KL}}(\rho^{*}_{\lambda,a},\pi )  \right ]  
	\\
	\notag 
	&\hspace{0.25in} + \sqrt{\frac{(M_{y}+uM_{c})^{2}\log(4/\epsilon)}{2n\kappa^{2}}} + u U_{1}\left ( \epsilon ; \lambda, u , n \right )  +\frac{1}{\lambda} \left [ \frac{\lambda^{2} \left ( M^{2}_{y} +u M_{c}^{2} \right )}{8n\kappa^{2}} +(1+u)\log\frac{4}{\epsilon} \right ]
	\\
	\notag 
	&= \min_{\rho\in\mathcal{E}_{ B  ( \hat{\rho}_{\lambda, u}   ) }} \left \{ \vphantom{\frac{\lambda^{2}}{8\kappa^{2}}} \int_{\Theta}R(\theta)d\rho(\theta)+\frac{1}{\lambda}D_{\mathrm{KL}}(\rho,\pi ) \right \}
	\\
	\label{1st orac ineq part (b) eq chain 4}
	&\hspace{0.7in}  + \sqrt{\frac{(M_{y}+uM_{c})^{2}\log(4/\epsilon)}{2n\kappa^{2}}} + u U_{1}\left ( \epsilon ; \lambda, u , n \right )  +\frac{1}{\lambda} \left [ \frac{\lambda^{2} \left ( M^{2}_{y} +u M_{c}^{2} \right )}{8n\kappa^{2}} +(1+u)\log\frac{4}{\epsilon} \right ] .
	\end{align}
	\normalsize
	In the above, \eqref{1st orac ineq part (b) eq chain 1} follows from plugging in \eqref{Event E2 (b)}, \eqref{1st orac ineq part (b) eq chain 2} follows from plugging in \eqref{implication on Event E4 (b)}, and lastly \eqref{1st orac ineq part (b) eq chain 4} follows from Lemma \ref{Lemma: related to budget expansions in Appendix for Section 4} (c).   Switching  the notation to  $\int_{\Theta}R(\theta)d\rho(\theta)=R(f_{G,\rho})$ for $\rho\in\mathcal{P}(\Theta)$ and utilizing the  definition of $U_{2}\left ( \epsilon ; \lambda, u , n \right )$, the above yields the statement in part (b) of the Theorem.  
\end{proof}

\subsubsection{Proofs for Subsection \ref{subsec: Normal Prior}: Normal Prior}

\begin{proof}
	[Proof of Theorem \ref{Theorem: Oracle inequality w/ normal prior}]
	We will use the following properties in the proofs of part (a) and (b).  For treatment assignment rules of the form in \eqref{Example treatment model class}, when $\lVert \theta \rVert \neq 0$, it holds that $f_{\theta}(x)=f_{\theta/\lVert \theta \rVert}(x)$ for all $x\in\mathcal{X}$.   As we are presuming that $\overline{\theta}\neq 0$ and $\overline{\theta}_{u}\neq 0$ (almost surely), with probability one we can find a values $\overline{\theta}$ and $\overline{\theta}_{u}$ such that  $\lVert \overline{\theta} \rVert =1$ and $\lVert \overline{\theta}_{u} \rVert =1$.   We assume $\overline{\theta}$ and $\overline{\theta}_{u}$ are selected to have this property for the remainder of the proof.  Below, for integration over $\Theta=\mathbb{R}^{q}$,  we write $\int\dots $ in place of $\int_{\mathbb{R}^{q}}\dots$
	
	Observe that for $\theta, \theta_{1}\in\mathbb{R}^{q}$ such that $\lVert \theta_{1}\rVert =1$ and $\lVert \theta \rVert \neq 0$, 
	\begin{align}
	\label{R-R eq 1}
	R\left ( \theta  \right )- R\left ( \theta_{1} \right ) &=  W\left (f_{\theta} \right ) - W\left ( f_{\theta_{1}} \right )
	\\
	\notag 
	&= E_{Q} \left [   \left ( Y_{1}-Y_{0} \right ) \left ( f_{\theta}(X) - f_{\theta_{1}}(X) \right )   \right ]
	\\
	\label{R-R eq 2} 
	&\leq  M_{y} E_{P} \left [ \left \lvert  1 \left \{ \phi(X)^{\intercal}\theta >0  \right \} - 1\left \{ \phi(X)^{\intercal}\theta_{1} >0 \right \} \right \rvert  \right ]
	\\
	\notag 
	& = M_{y} P \left [ \left (  \phi(X)^{\intercal}\theta \right ) \left (  \phi(X)^{\intercal}\theta_{1} \right ) < 0  \right ]
	\\
	\notag 
	& = M_{y} P\left [ \left (  \phi(X)^{\intercal}\frac{\theta}{\lVert \theta \rVert}  \right ) \left (  \phi(X)^{\intercal}\theta_{1} \right ) < 0 \right ]
	\\
	\label{R-R eq 3}
	&\leq M_{y} \nu \left \lVert \frac{\theta}{\lVert \theta \rVert} - \theta_{1} \right \rVert
	\\
	\label{R-R eq 4}
	&\leq M_{y} 2 \nu  \left \lVert \theta - \theta_{1}  \right  \rVert,
	\end{align}
	where \eqref{R-R eq 1} follows from the definition of welfare regret, \eqref{R-R eq 2} follows from Assumption \ref{Assumption: treatment identification and boundedness} (iii) and the fact that the distribution of $X$ is determined by $P$ as well as $Q$, \eqref{R-R eq 3} follows from Assumption \ref{Assumption: technical condition for comparisons to EWM}, and \eqref{R-R eq 4} follows from the fact that with $\theta, \theta_{1}$ as above,  
	\[\left \lVert \frac{\theta}{\lVert\theta\rVert} - \theta_{1} \right  \rVert \leq \lVert \theta - \theta_{1} \rVert .\]
	As a consequence of \eqref{R-R eq 4}, for any $\sigma>0$,
	\begin{align}
	\notag 
	\int R(\theta) d\Phi_{\theta_{1},\sigma^{2}}(\theta) &= R\left ( \theta_{1} \right ) + \int \left [ R \left ( \theta \right ) - R\left ( \theta_{1} \right ) \right ]d\Phi_{\theta_{1},\sigma^{2}}(\theta)
	\\
	\notag 
	&\leq R\left ( \theta_{1}\right ) + 2M_{y}\nu \int \left \lVert \theta-\theta_{1} \right \rVert d\Phi_{\theta_{1},\sigma^{2}}(\theta)
	\\
	\label{Ineq: int R normal prior}
	&\leq R\left ( \theta_{1} \right ) + 2M_{y}\nu \sigma \sqrt{q}, 
	\end{align}
	where we have used the fact that for $\theta \sim \Phi_{\theta_{1},\sigma^{2}}$, $||\theta-\theta_{1}||\sim \sigma H^{1/2}$ with $H\sim \chi^{2}(q)$.  Then, by Jensen's inequality, $E \sigma H^{1/2} \leq \sigma (EH)^{1/2}=\sigma (q)^{1/2}$.
	
	Following nearly identical steps, now starting with the definition of the expected costs $K(\theta)$ and $K(\overline{\theta})$, it is straightforward to derive that, for $\theta, \theta_{1}\in\mathbb{R}^{q}$ such that $\lVert \theta_{1}\rVert =1$ and $\lVert \theta \rVert \neq 0$,
	\begin{equation*}
	K\left ( \theta \right ) - K\left (\theta_{1} \right )  \leq M_{c} 2\nu \left  \lVert \theta - \theta_{1}  \right \rVert,
	\end{equation*}
	and for $\sigma>0$,
	\begin{equation}
	\label{Ineq: int G normal prior} 
	\int K(\theta)d\Phi_{\theta_{1},\sigma^{2}}(\theta) \leq K\left ( \theta_{1} \right ) +2M_{c} \nu \sigma \sqrt{q}.
	\end{equation}
	
	Lastly, before considering part (a) and (b) separately, note that by Lemma \ref{Lemma: Normal KL}, with $\sigma_{\pi}=1/\sqrt{q}$, $\sigma_{\rho}=1/(2\sqrt{nq})$,  and $\lVert \theta_{1}\rVert =1$,
	\begin{align}
	\label{Eqn: D_kl in oracle ineq proof normal prior}
	D_{\mathrm{KL}}\left ( \Phi_{\theta_{1},\sigma^{2}_{\rho}}, \Phi_{0,\sigma^{2}_{\pi }} \right ) = \frac{q}{2}\left [ \frac{1}{4n}+\log\left (4n \right ) \right ].
	\end{align}
	
	\medskip 
	Part (a).  We consider the posterior distribution $\widetilde{\rho}= \Phi_{\overline{\theta},\sigma^{2}_{\rho}}$ with $\sigma_{\rho}=1/(2\sqrt{nq})$ so that $D_{\mathrm{KL}}(\widetilde{\rho},\pi)$ is given by \eqref{Eqn: D_kl in oracle ineq proof normal prior}.  Next, define
	\begin{equation}
	\label{B' in part (a) of normal oracle inequality proof}
	B' = B + \frac{\nu M_{c}}{\sqrt{n}}.
	\end{equation}
	
	Assumptions \ref{Assumption: measurability} and \ref{Assumption: prior indep of data} are met and Assumptions \ref{Assumption: technical condition for comparisons to EWM} and \ref{Assumption: CQ} are assumed to hold so we can apply Theorem \ref{Theorem: Main Oracle Inequality Constrained Case} (a).   Starting from there, with probability at least $1-\epsilon $ we have
	\begin{align}
	\notag 
	&\int R(\theta) d\hat{\rho}_{\lambda, \hat{u}}(\theta) 
	\\
	\notag 
	& \leq \min_{\rho\in\mathcal{E}_{B}} \left \{ \int R(\theta) d\rho(\theta) + \frac{2}{\lambda}D_{\mathrm{KL}}(\rho,\pi)   \right \} + \frac{2}{\lambda}\left [ \frac{\lambda^{2}M^{2}_{y}}{8n\kappa^{2}} +\log\frac{3}{\epsilon }  \right ] + \hat{u} \sqrt{\frac{M_{c}^{2}\log\frac{3}{\epsilon }}{2n\kappa^{2}} } 
	\\
	\label{equation string normal oracle (a) 1}
	&= \int R(\theta) d\rho^{*}_{\lambda/2,u^{*}}(\theta) + \frac{2}{\lambda}D_{\mathrm{KL}}(\rho^{*}_{\lambda/2,u^{*}},\pi) + \frac{2}{\lambda}\left [ \frac{\lambda^{2}M^{2}_{y}}{8n\kappa^{2}} +\log\frac{3}{\epsilon }  \right ] + \hat{u} \sqrt{\frac{M_{c}^{2}\log\frac{3}{\epsilon }}{2n\kappa^{2}} }  
	\\
	\notag 
	&=\int R(\theta) d\rho^{*}_{\lambda/2,u^{*}}(\theta) + u^{*}\left ( \int K(\theta)d\rho^{*}_{\lambda/2,u^{*}}(\theta) - B \right ) + \frac{2}{\lambda}D_{\mathrm{KL}}(\rho^{*}_{\lambda/2,u^{*}},\pi)
	\\
	\label{equation string normal oracle (a) 2}
	& \hspace{0.25in} + \frac{2}{\lambda}\left [ \frac{\lambda^{2}M^{2}_{y}}{8n\kappa^{2}} +\log\frac{3}{\epsilon }  \right ] + \hat{u} \sqrt{\frac{M_{c}^{2}\log\frac{3}{\epsilon }}{2n\kappa^{2}} }  
	\\
	\notag 
	& = \int R(\theta) d\rho^{*}_{\lambda/2,u^{*}}(\theta) + u^{*}\left ( \int K(\theta)d\rho^{*}_{\lambda/2,u^{*}}(\theta) - B' \right ) + \frac{2}{\lambda}D_{\mathrm{KL}}(\rho^{*}_{\lambda/2,u^{*}},\pi)
	\\
	\notag 
	&\hspace{0.25in} + u^{*}\left (B'-B \right ) + \frac{2}{\lambda}\left [ \frac{\lambda^{2}M^{2}_{y}}{8n\kappa^{2}} +\log\frac{3}{\epsilon }  \right ] + \hat{u} \sqrt{\frac{M_{c}^{2}\log\frac{3}{\epsilon }}{2n\kappa^{2}} }  
	\\
	\notag 
	& \leq  \sup_{u\geq 0} \left [ \int R(\theta) d\rho^{*}_{\lambda/2,u}+u\left ( \int K(\theta) d\rho^{*}_{\lambda,u}(\theta) - B' \right ) +\frac{2}{\lambda}D_{\mathrm{KL}}(\rho^{*}_{\lambda,u},\pi ) \right ] 
	\\
	\notag 
	&\hspace{0.25in} + u^{*}\left (B'-B \right )+ \frac{2}{\lambda}\left [ \frac{\lambda^{2}M^{2}_{y}}{8n\kappa^{2}} +\log\frac{3}{\epsilon }  \right ] + \hat{u} \sqrt{\frac{M_{c}^{2}\log\frac{3}{\epsilon }}{2n\kappa^{2}} }  
	\\
	\label{equation string normal oracle (a) 3}
	&=\min_{\rho\in\mathcal{E}_{B'}} \left \{ \int R(\theta)d\rho(\theta) + \frac{2}{\lambda} D_{\mathrm{KL}} \left ( \rho,\pi \right ) \right \} + u^{*}\frac{\nu M_{c}}{\sqrt{n}}+ \frac{2}{\lambda}\left [ \frac{\lambda^{2}M^{2}_{y}}{8n\kappa^{2}} +\log\frac{3}{\epsilon }  \right ] + \hat{u} \sqrt{\frac{M_{c}^{2}\log\frac{3}{\epsilon }}{2n\kappa^{2}} }  
	\end{align}
	In the above, \eqref{equation string normal oracle (a) 1} and \eqref{equation string normal oracle (a) 2} follow from Lemma \ref{Lemma: related to budget expansions in Appendix for Section 4} (b) while \eqref{equation string normal oracle (a) 3} follows from applying Lemma \ref{Lemma: related to budget expansions in Appendix for Section 4} (b) and the definition of $B'$ in \eqref{B' in part (a) of normal oracle inequality proof}. 
	
	From \eqref{Ineq: int G normal prior} with $\overline{\theta}$ in the place of $\theta_{1}$ and with $\sigma_{\rho}=1/(2\sqrt{nq})$, we have
	\begin{align}
	\label{Eqn: rho tilde in budget B' normal oracle theorem (a)} 
	\int K(\theta) d\widetilde{\rho}(\theta)= \int K(\theta)d\Phi_{\overline{\theta},\sigma^{2}_{\rho}}(\theta) \leq K\left ( \overline{\theta} \right ) +\frac{\nu M_{c}}{\sqrt{n}} \leq B'
	\end{align}
	as, by the definition of $\overline{\theta}$, $K(\overline{\theta})\leq B$.  Therefore $\widetilde{\rho} \in \mathcal{E}_{B'}$.  From \eqref{equation string normal oracle (a) 3}, we have, with probability at least $1-\epsilon$,
	\begin{align*}
	&\int R(\theta) d\hat{\rho}_{\lambda, \hat{u}}(\theta)
	\\ 
	& \leq \min_{\rho\in\mathcal{E}_{B'}} \left \{ \int R(\theta)d\rho(\theta) + \frac{2}{\lambda} D_{\mathrm{KL}} \left ( \rho,\pi \right ) \right \} + u^{*}\frac{\nu M_{c}}{\sqrt{n}}+ \frac{2}{\lambda}\left [ \frac{\lambda^{2}M^{2}_{y}}{8n\kappa^{2}} +\log\frac{3}{\epsilon }  \right ] + \hat{u} \sqrt{\frac{M_{c}^{2}\log\frac{3}{\epsilon }}{2n\kappa^{2}} }  
	\\
	&\leq \int R(\theta) d\widetilde{\rho}(\theta) + \frac{2}{\lambda}D_{\mathrm{KL}} \left ( \widetilde{\rho},\pi \right ) + u^{*}\frac{\nu M_{c}}{\sqrt{n}}+ \frac{2}{\lambda}\left [ \frac{\lambda^{2}M^{2}_{y}}{8n\kappa^{2}} +\log\frac{3}{\epsilon }  \right ] + \hat{u} \sqrt{\frac{M_{c}^{2}\log\frac{3}{\epsilon }}{2n\kappa^{2}} }  
	\\
	&\leq  R\left ( \overline{\theta} \right ) + \frac{\nu M_{y}}{\sqrt{n}}+ \frac{q}{\lambda }\left [ \frac{1}{4n}+\log\left (4n \right ) \right ] + u^{*}\frac{\nu M_{c}}{\sqrt{n}}+ \frac{2}{\lambda}\left [ \frac{\lambda^{2}M^{2}_{y}}{8n\kappa^{2}} +\log\frac{3}{\epsilon }  \right ] + \hat{u} \sqrt{\frac{M_{c}^{2}\log\frac{3}{\epsilon }}{2n\kappa^{2}} } .
	\end{align*}  
	In the last step, we have applied the properties in \eqref{Ineq: int R normal prior} and \eqref{Eqn: D_kl in oracle ineq proof normal prior} with $\overline{\theta}$ taking the role of $\theta_{1}$.  Plugging in $\lambda = \kappa  \sqrt{n q}/M_{y}$ and rearranging terms then produces the result in (a) with
	\[\overline{U}_{1}(n;q)=\sqrt{\frac{q}{n}} \left [ \frac{\nu M_{y} }{\sqrt{q}} +\frac{M_{y}}{\kappa}\left ( \frac{1}{4}+\frac{1}{4n} \right ) \right ]. \]
	
	\medskip 
	
	Part (b).  As a starting point, we utilize the setup and initial steps of the proof of Theorem \ref{Theorem: Main Oracle Inequality Constrained Case} (b).  Assume the same the definitions of events $E_{1}$, $E_{2}$ and $E_{3}$ as in \eqref{Event E1 (b)}, \eqref{Event E2 (b)}, \eqref{Event E3 (b)}, respectively.  Following that proof up to \eqref{1st orac ineq part (b) eq chain 2}, we have that with probability at least $1-\epsilon$,
	\begin{align}
	\notag 
	&\int R(\theta) d\hat{\rho}_{\lambda, u}(\theta ) 
	\\
	\notag 
	&\leq \int R(\theta)d\rho^{*}_{\lambda,u}(\theta)+u \left ( \int K(\theta)d\rho^{*}_{\lambda,u}(\theta)-B \left ( \hat{\rho}_{\lambda, u}  \right ) \right ) +\frac{1}{\lambda}D_{\mathrm{KL}}(\rho^{*}_{\lambda,u},\pi )  
	\\
	\label{1st step normal oracle inequality proof part (b)}
	& \hspace{0.1in}+\sqrt{\frac{(M_{y}+uM_{c})^{2}\log(4/\epsilon)}{2n\kappa^{2}}} + u U_{1}\left ( \epsilon ; \lambda, u, n \right )   +\frac{1}{\lambda} \left [ \frac{\lambda^{2} \left ( M^{2}_{y} +u M_{c}^{2} \right )}{8n\kappa^{2}} +(1+u)\log\frac{4}{\epsilon} \right ],
	\end{align}
	where $B ( \hat{\rho}_{\lambda, u}  )=\int K(\theta)d\hat{\rho}_{\lambda, u} = K(f_{G,\hat{\rho}_{\lambda, u}})$ and $U_{1}\left ( \epsilon ; \lambda, u, n \right ) $ is defined in Theorem \ref{Theorem: Main Oracle Inequality Constrained Case} (b). 
	
	Now we will consider the posterior $\widetilde{\rho}=\Phi_{\overline{\theta}_{u},\sigma^{2}_{\rho}} $ with $\sigma_{\rho}=1/(2\sqrt{nq})$.  Utilizing \eqref{Eqn: D_kl in oracle ineq proof normal prior} with $\pi$ as described in the theorem, now with $\overline{\theta}_{u}$ in place of $\theta_{1}$, with probability one we have
	\begin{equation}
	\label{DKL normal prior oracle ineq part b}
	D_{\mathrm{KL}}(\widetilde{\rho},\pi) = \frac{q}{2}\left [ \frac{1}{4n}+ \log\left (4n \right ) \right ].
	\end{equation}
	Additionally, we now define 
	\begin{equation}
	\label{B' in normal oracle ineq part b}
	B' = B \left ( \hat{\rho}_{\lambda, u}  \right )+ \frac{\nu M_{c}}{\sqrt{n}}. 
	\end{equation}
	From \eqref{Ineq: int G normal prior} with $\overline{\theta}_{u}$ in the place of $\theta_{1}$ and with $\sigma_{\rho}=1/(2\sqrt{nq})$, with probability one we have
	\begin{align}
	\label{Eqn: rho tilde in budget B' normal oracle theorem (b)} 
	\int K(\theta) d\widetilde{\rho}(\theta)= \int K(\theta)d\Phi_{\overline{\theta}_{u},\sigma^{2}_{\rho}}(\theta) \leq K\left ( \overline{\theta}_{u} \right ) +\frac{\nu M_{c}}{\sqrt{n}} \leq B',
	\end{align}
	because, by the definition of $\overline{\theta}_{u}$ we have $K(\overline{\theta}_{u})\leq B \left ( \hat{\rho}_{\lambda, u}  \right )$ (a.s.).  It follows that with probability one, $\widetilde{\rho}\in \mathcal{E}_{B'}$.
	
	Returning to \eqref{1st step normal oracle inequality proof part (b)}, we have, with probability at least $1-\epsilon$,
	\small 
	\begin{align}
	\notag 
	&\int R(\theta) d\hat{\rho}_{\lambda, u}(\theta ) 
	\\
	\notag 
	&\leq \int R(\theta)d\rho^{*}_{\lambda,u}(\theta)+u \left ( \int K(\theta)d\rho^{*}_{\lambda,u}(\theta)-B \left ( \hat{\rho}_{\lambda, u}  \right ) \right ) +\frac{1}{\lambda}D_{\mathrm{KL}}(\rho^{*}_{\lambda,u},\pi )  +\sqrt{\frac{(M_{y}+uM_{c})^{2}\log(4/\epsilon)}{2n\kappa^{2}}}
	\\
	\notag 
	&\hspace{0.25in} + u U_{1}\left ( \epsilon ; \lambda, u, n \right )   +\frac{1}{\lambda} \left [ \frac{\lambda^{2} \left ( M^{2}_{y} +u M_{c}^{2} \right )}{8n\kappa^{2}} +(1+u)\log\frac{4}{\epsilon} \right ]
	\\
	\notag 
	&\leq \int R(\theta)d\rho^{*}_{\lambda,u}(\theta)+u \left ( \int K(\theta)d\rho^{*}_{\lambda,u}(\theta)-B' \right ) +\frac{1}{\lambda}D_{\mathrm{KL}}(\rho^{*}_{\lambda,u},\pi )  +\sqrt{\frac{(M_{y}+uM_{c})^{2}\log(4/\epsilon)}{2n\kappa^{2}}}
	\\
	\notag 
	&\hspace{0.25in} +u\left ( B'-B \left ( \hat{\rho}_{\lambda, u}  \right ) \right ) + u U_{1}\left ( \epsilon ; \lambda, u, n \right )   +\frac{1}{\lambda} \left [ \frac{\lambda^{2} \left ( M^{2}_{y} +u M_{c}^{2} \right )}{8n\kappa^{2}} +(1+u)\log\frac{4}{\epsilon} \right ]
	\\
	\notag 
	&\leq \sup_{a\geq 0} \left [ \int R(\theta)d\rho^{*}_{\lambda,a}(\theta)+u \left ( \int K(\theta)d\rho^{*}_{\lambda,a}(\theta)-B' \right ) +\frac{1}{\lambda}D_{\mathrm{KL}}(\rho^{*}_{\lambda,a},\pi )  \right ] 
	\\
	\label{equation string normal oracle (b) 1}
	&\hspace{0.25in}  +\sqrt{\frac{(M_{y}+uM_{c})^{2}\log(4/\epsilon)}{2n\kappa^{2}}} + u \left ( \frac{\nu M_{c}}{\sqrt{n}} \right ) + u U_{1}\left ( \epsilon ; \lambda, u, n \right )   +\frac{1}{\lambda} \left [ \frac{\lambda^{2} \left ( M^{2}_{y} +u M_{c}^{2} \right )}{8n\kappa^{2}} +(1+u)\log\frac{4}{\epsilon} \right ]
	\\
	\notag 
	& = \inf_{\rho\in\mathcal{E}_{B'}} \left [ \int R(\theta)d\rho(\theta) +\frac{1}{\lambda}D_{\mathrm{KL}}\left ( \rho, \pi \right ) \right ] +\sqrt{\frac{(M_{y}+uM_{c})^{2}\log(4/\epsilon)}{2n\kappa^{2}}} 
	\\
	\label{equation string normal oracle (b) 2}
	&\hspace{0.25in}   + u \left ( \frac{\nu M_{c}}{\sqrt{n}} \right ) + u U_{1}\left ( \epsilon ; \lambda, u, n \right )   +\frac{1}{\lambda} \left [ \frac{\lambda^{2} \left ( M^{2}_{y} +u M_{c}^{2} \right )}{8n\kappa^{2}} +(1+u)\log\frac{4}{\epsilon} \right ]
	\\
	\notag 
	& \leq \int R(\theta) d\widetilde{\rho}(\theta) + \frac{1}{\lambda} D_{\mathrm{KL}}\left ( \widetilde{\rho},\pi \right )  +\sqrt{\frac{(M_{y}+uM_{c})^{2}\log(4/\epsilon)}{2n\kappa^{2}}} 
	\\
	\label{equation string normal oracle (b) 3}
	&\hspace{0.25in}   + u \left ( \frac{\nu M_{c}}{\sqrt{n}} \right ) + u U_{1}\left ( \epsilon ; \lambda, u, n \right )   +\frac{1}{\lambda} \left [ \frac{\lambda^{2} \left ( M^{2}_{y} +u M_{c}^{2} \right )}{8n\kappa^{2}} +(1+u)\log\frac{4}{\epsilon} \right ]
	\\
	\notag 
	&\leq R\left ( \overline{\theta}_{u} \right ) + \frac{\nu M_{y}}{\sqrt{n}}+ \frac{q}{2\lambda}\left [ \frac{1}{4n}+ \log\left (4n \right ) \right ] + \sqrt{\frac{(M_{y}+uM_{c})^{2}\log(4/\epsilon)}{2n\kappa^{2}}} 
	\\
	\label{equation string normal oracle (b) 4}
	&\hspace{0.25in}   + u \left ( \frac{\nu M_{c}}{\sqrt{n}} \right ) + u U_{1}\left ( \epsilon ; \lambda, u, n \right )   +\frac{1}{\lambda} \left [ \frac{\lambda^{2} \left ( M^{2}_{y} +u M_{c}^{2} \right )}{8n\kappa^{2}} +(1+u)\log\frac{4}{\epsilon} \right ]
	\end{align}
	\normalsize
	Above, \eqref{equation string normal oracle (b) 1} follows from \eqref{B' in normal oracle ineq part b} and the fact that the supremum there is greater than or equal to the object it replaces,  \eqref{equation string normal oracle (b) 2} follows from Lemma \ref{Lemma: related to budget expansions in Appendix for Section 4} (c), \eqref{equation string normal oracle (b) 3} follows from having $\widetilde{\rho}\in\mathcal{E}_{B'}$ with probability one, and lastly \eqref{equation string normal oracle (b) 4} follows from \eqref{Ineq: int R normal prior}, with $\overline{\theta}_{u}$ in place of $\theta_{1}$ and $\sigma_{\rho}=1/(2\sqrt{nq})$ in place of $\sigma$, and utilizing \eqref{DKL normal prior oracle ineq part b}.  
	Plugging in $\lambda$ as given in part (b), straightforward manipulations of the expression in \eqref{equation string normal oracle (b) 4} show that the inequality can be written
	\begin{align*}
	R\left ( f_{G,\hat{\rho}_{\lambda,u}} \right ) \leq R\left ( \overline{\theta}_{u} \right )+ \frac{M_{y}+uM_{c}}{\kappa} \left [ \overline{U}_{2}(n;q,u,\epsilon) + \overline{U}_{3}(n;q,u,\epsilon)+\overline{U}_{4}(n;q,u) \right ],
	\end{align*}
	where
	\[\overline{U}_{2}(n;q,u,\epsilon) = \frac{\sqrt{q}\log\left ( 2\sqrt{n} \right )+ \sqrt{2}u \sqrt{\log\left ( 2\sqrt{n}  \right ) +\log\frac{4}{\epsilon}  }}{\sqrt{n}} = \mathcal{O}\left (\frac{\log n}{\sqrt{n}} \right ), \]
	\[\overline{U}_{3}(n;q,u,\epsilon) = \frac{\sqrt{\frac{\log \left ( 4/\epsilon \right )}{2}}+\frac{1}{\sqrt{q}}(1+u)\log\frac{4}{\epsilon}}{\sqrt{n}} = \mathcal{O}\left ( \frac{1}{\sqrt{n}} \right ),\]
	and
	\[\overline{U}_{4}(n;q,u)= \frac{\kappa \nu +\sqrt{q} \left ( \frac{1}{8n} +\frac{u}{2} \right ) }{\sqrt{n}}  +  \sqrt{\frac{q}{n}} \left ( \frac{M^{2}_{y}+uM^{2}_{c}}{8\left (M_{y}+uM_{c} \right )^{2}} \right )  = \mathcal{O}\left ( \frac{1}{\sqrt{n}}\right ).\]
\end{proof}

\subsubsection{Proofs for Subsection \ref{subsec: The Majority Vote TR}: The Majority Vote Treatment Rule}
\begin{proof}[Proof of Theorem \ref{Theorem: mv loss bounded by twice the budget penalized welfare regret}]
	First, note that
	\begin{equation}
	\label{equation:  mv leq 2 gibbs}
	f_{\mathrm{mv},\rho}(x) = 1\left \{ \int_{\Theta} f_{\theta}(x) d\rho(\theta) >\frac{1}{2} \right \} \leq 2 \int_{\Theta}f_{\theta}(x)d\rho(\theta) = 2 f_{G,\rho}(x).
	\end{equation}
	
	To see this, note that when $\int_{\Theta}f(x)d\rho(\theta) \leq 1/2$, $f_{\mathrm{mv},\rho}=0$ hence the left-hand size of the above inequality is zero while the right-hand size is non-negative and the inequality holds.  When $\int_{\Theta}f_{\theta}(x)d\rho(\theta)>1/2$, the left hand side is $1$ while the right hand side must be greater than $1$, so the inequality holds in all cases.  
	
	Next we will show that for any $x\in\mathcal{X}$, 
	\begin{align}
	\label{equation: for all x less than 2}
	\left ( \delta_{y}(x)-\eta_{B(\rho)} \delta_{c}(x) \right ) \left (f^{*}_{B(\rho)}(x)-f_{\mathrm{mv},\rho }(x) \right ) \leq 2 \left ( \delta_{y}(x)-\eta_{B(\rho)} \delta_{c}(x) \right ) \left (f^{*}_{B(\rho)}(x)-f_{G,\rho }(x) \right ).
	\end{align}
	To see this, first consider $x\in \mathcal{X}$ such that $f^{*}_{B(\rho)}(x)=1\{\delta_{y}(x)-\eta_{B(\rho)}\delta_{c}(x)>0\}=0.$  In this case, $\delta_{y}(x)-\eta_{B(\rho)}\delta_{c}(x)\leq 0$ and we have
	\begin{align*}
	\left ( \delta_{y}(x)-\eta_{B(\rho)} \delta_{c}(x) \right ) \left (f^{*}_{B(\rho)}(x)-f_{\mathrm{mv},\rho }(x) \right )  &= \left \vert \delta_{y}(x)-\eta_{B(\rho)} \delta_{c}(x) \right \vert f_{\mathrm{mv},\rho}(x) 
	\\
	&\leq  2 \left \vert \delta_{y}(x)-\eta_{B(\rho)} \delta_{c}(x) \right \vert f_{G,\rho}(x) 
	\\
	& = 2 \left ( \delta_{y}(x)-\eta_{B(\rho)} \delta_{c}(x) \right ) \left (f^{*}_{B(\rho)}(x)-f_{G,\rho }(x) \right ),
	\end{align*}
	where the inequality follows from \eqref{equation:  mv leq 2 gibbs}.  To verify \eqref{equation: for all x less than 2},  we now need to check that it holds for $x\in\mathcal{X}$ such that $f^{*}_{B(\rho)}(x)=1\{\delta_{y}(x)-\eta_{B(\rho)}\delta_{c}(x)>0\}=1.$  In this case, $\delta_{y}(x)-\eta_{B(\rho)}\delta_{c}(x)>0$, so \eqref{equation: for all x less than 2} reduces to
	\begin{equation}
	\left (f^{*}_{B(\rho)}(x)-f_{\mathrm{mv},\rho }(x) \right ) \leq 2  \left (f^{*}_{B(\rho)}(x)-f_{G,\rho }(x) \right ).
	\end{equation}
	First consider $x\in\mathcal{X}$ such that $f_{\mathrm{mv},\rho}(x)=1$.  Then the left-hand size is zero while the right hand side is non-negative as $f_{G,\rho}\in[0,1]$ and $f^{*}_{B(\rho)(x)}=1$ in the current assumed setting, so the condition holds.  Lastly, if $f_{\mathrm{mv},\rho}(x)=0$, so that $\int_{\Theta}f_{\theta}(x)d\rho(\theta)\leq 1/2$, in the current setting with $f^{*}_{B(\rho)}(x)=1$ we then have
	\begin{align*}
	2 \left ( f^{*}_{B(\rho)}(x)-f_{G,\rho }(x) \right ) &= 2\left ( 1-\int_{\Theta}f_{\theta}(x)d\rho(\theta) \right )
	\\
	&\geq 2 \left (1-\frac{1}{2} \right )
	\\
	&=1
	\\
	&= f^{*}_{B(\rho)}(x)-f_{\mathrm{mv},\rho}(x).
	\end{align*}
	Hence \eqref{equation: for all x less than 2} holds for all $x\in\mathcal{X}$.  Taking the expectation of both sides of that inequality with respect to a draw of $X$ from $Q$ then yields that 
	\begin{equation}
	\label{Equation L_B of mv  leq 2 L_B Gibbs}
	L_{B(\rho)}\left ( f_{\mathrm{mv},\rho} \right ) \leq 2 L_{B(\rho)}\left ( f_{G,\rho} \right ). 
	\end{equation}
	To complete the proof, we need to verify that 
	\begin{equation}
	\label{equation: 2x theorem last step}
	L_{B(\rho)}\left ( f_{G,\rho} \right ) = R_{B(\rho)}\left ( f_{G,\rho} \right ). 
	\end{equation}
	Now there are two possibilities to consider.  The first is when $\eta_{B(\rho)}=0$.  In this case, we have
	\begin{align*}
	L_{B(\rho)}\left ( f_{G,\rho} \right ) &= E_{Q}\left [ \delta_{y}(X)\left ( f_{B(\rho)}^{*}-f_{G,\rho} \right ) \right ]
	\\
	&=W\left (f_{B(\rho)}^{*} \right ) - W\left ( f_{G,\rho} \right ) = R_{B(\rho)}\left ( f_{G,\rho} \right ).
	\end{align*}
	And in the only remaining case, when $\eta_{B(\rho)}>0$, we also have $K(f_{B(\rho)}^{*})=B(\rho)$ by Theorem \ref{Theorem: theoretically optimal policy choice}.  As, by the definition of $B(\rho)$, it also holds that $K(f_{G,\rho})=B(\rho)$, we have
	\begin{align*}
	L_{B(\rho)}\left ( f_{G,\rho} \right ) &= E_{Q}\left [ \delta_{y}(X) \left ( f_{B(\rho)}^{*}-f_{G,\rho} \right )  \right ] -\eta_{B(\rho)} E_{Q}\left [ \delta_{c}(X) \left ( f_{B(\rho)}^{*}-f_{G,\rho} \right )  \right ]
	\\
	&=E_{Q}\left [ \delta_{y}(X) \left ( f_{B(\rho)}^{*}-f_{G,\rho} \right )  \right ] -\eta_{B(\rho)} \left [ K\left ( f^{*}_{B(\rho)} \right ) - K\left ( f_{G,\rho} \right ) \right ]
	\\
	&=E_{Q}\left [ \delta_{y}(X) \left ( f_{B(\rho)}^{*}-f_{G,\rho} \right )  \right ] 
	\\
	&= R_{B(\rho)}\left ( f_{G,\rho} \right ).
	\end{align*}
	Hence \eqref{equation: 2x theorem last step} holds and combined with \eqref{Equation L_B of mv  leq 2 L_B Gibbs} this completes the proof.
\end{proof}

\newpage

\bibliographystyle{apalike}
\bibliography{PB_EWM_constrained}

\end{document}